\documentclass[letterpaper,journal]{IEEEtran}

\usepackage{amsmath,amsfonts,amssymb,amsthm}
\usepackage{algorithmic}
\usepackage{algorithm}
\usepackage{array}
\usepackage{booktabs}
\usepackage{cite}
\usepackage[dvipsnames,table]{xcolor}
\usepackage{enumitem}
\usepackage{graphicx}
\usepackage{makecell}
\usepackage{multirow}
\usepackage{orcidlink}
\usepackage{placeins}
\usepackage{stfloats}
\usepackage{tabularx}
\usepackage{textcomp}
\usepackage{url}

\graphicspath{{figures/}}

\newtheorem{lemma}{Lemma}
\newtheorem{theorem}{Theorem}
\newtheorem{proposition}{Proposition}
\newtheorem{corollary}{Corollary}

\newtheorem{assumption}{Assumption}
\theoremstyle{definition}

\theoremstyle{plain}

\newcommand{\sg}{\operatorname{sg}}
\newcommand{\anchor}{a}
\newcommand{\bestresult}[1]{{\bfseries\boldmath #1}}

\title{
Federated Semantic Residuals for \\ Multimodal Recommendation
}

\author{
  Zhiwei Li\,\orcidlink{0000-0002-8897-8905},
  Guodong Long\,\orcidlink{0000-0003-3740-9515},
  Chengqi Zhang\,\orcidlink{0000-0001-5715-7154},~\IEEEmembership{Fellow,~IEEE},
  and Jing Jiang\,\orcidlink{0000-0001-5301-7779}%
  \thanks{Zhiwei Li, Guodong Long, and Jing Jiang are with the Australian Artificial Intelligence Institute, Faculty of Engineering and Information Technology, University of Technology Sydney, Sydney, NSW, Australia (e-mail: zhw.li@outlook.com; guodong.long@uts.edu.au; jing.jiang@uts.edu.au).}%
  \thanks{Chengqi Zhang is with the Department of Data Science and Artificial Intelligence, The Hong Kong Polytechnic University, Hong Kong Special Administrative Region, China (e-mail: chengqi.zhang@polyu.edu.hk).}
}

\hypersetup{hidelinks,
pdftitle={Federated Semantic Residuals for Multimodal Recommendation},
pdfauthor={Zhiwei Li, Guodong Long, Chengqi Zhang, Jing Jiang}}
\begin{document}

\maketitle

\begin{abstract}
  Federated Recommendation Systems \mbox{(FedRSs)} learn collaborative scores from interaction histories on device and rank items by them, so a higher score means a higher rank.
  Item text and images add semantic information that can reorder the list.
  Existing Federated Multimodal Recommenders \mbox{(FedMRs)} fuse this content into their models, which can obscure how content reorders the list.
  We therefore propose the \textbf{SE}mantic \textbf{R}esidual \textbf{A}dapter (SERA), which keeps content out of an existing FedRS and adds semantic residual scores to its collaborative scores.
  Content can then lift an item above another only when its semantic residual score exceeds the other's by more than their collaborative score gap.
  This score grows with the item's content similarity to the user's interaction history and is capped at a residual bound.
  We prove that for any content, an item can enter the list only if its collaborative score is within twice this bound of the lowest one in the list.
  A larger bound lets content lift more items, including irrelevant ones, so the bound trades off how far content reaches against accuracy.
  With a validated bound, SERA improves all three FedRSs on all six datasets without adding personalized parameters. It outperforms FedMRs in most comparisons and lifts some FedRSs above centralized baselines.
\end{abstract}


\section{Introduction}\label{sec:intro}

Federated Recommendation Systems (FedRSs) learn collaborative scores from interaction histories stored on user devices \cite{zhang2025personalized} and rank items by these scores, so an item appears above another if it scores higher.
Observed interactions provide the main supervision \cite{ammad2019federated,koren2009matrix}, and each user or item embedding is learned from the gradients of its own interactions \cite{li2024fedrap}.
As a result, users and items with sparse histories have too few interactions to learn reliable ID-based scores \cite{li2024fedcore,zhang2024ifedrec}.

Item text and images ease this sparsity, because they exist before an item collects many interactions \cite{du2018personalized}.
They add semantic information to interaction histories \cite{zhang2023beyond} and help multimodal recommenders refine collaborative rankings \cite{xu2021multi,zhang2021lattice}.

\begin{figure}[!t]
  \centering
  \includegraphics[width=\columnwidth]{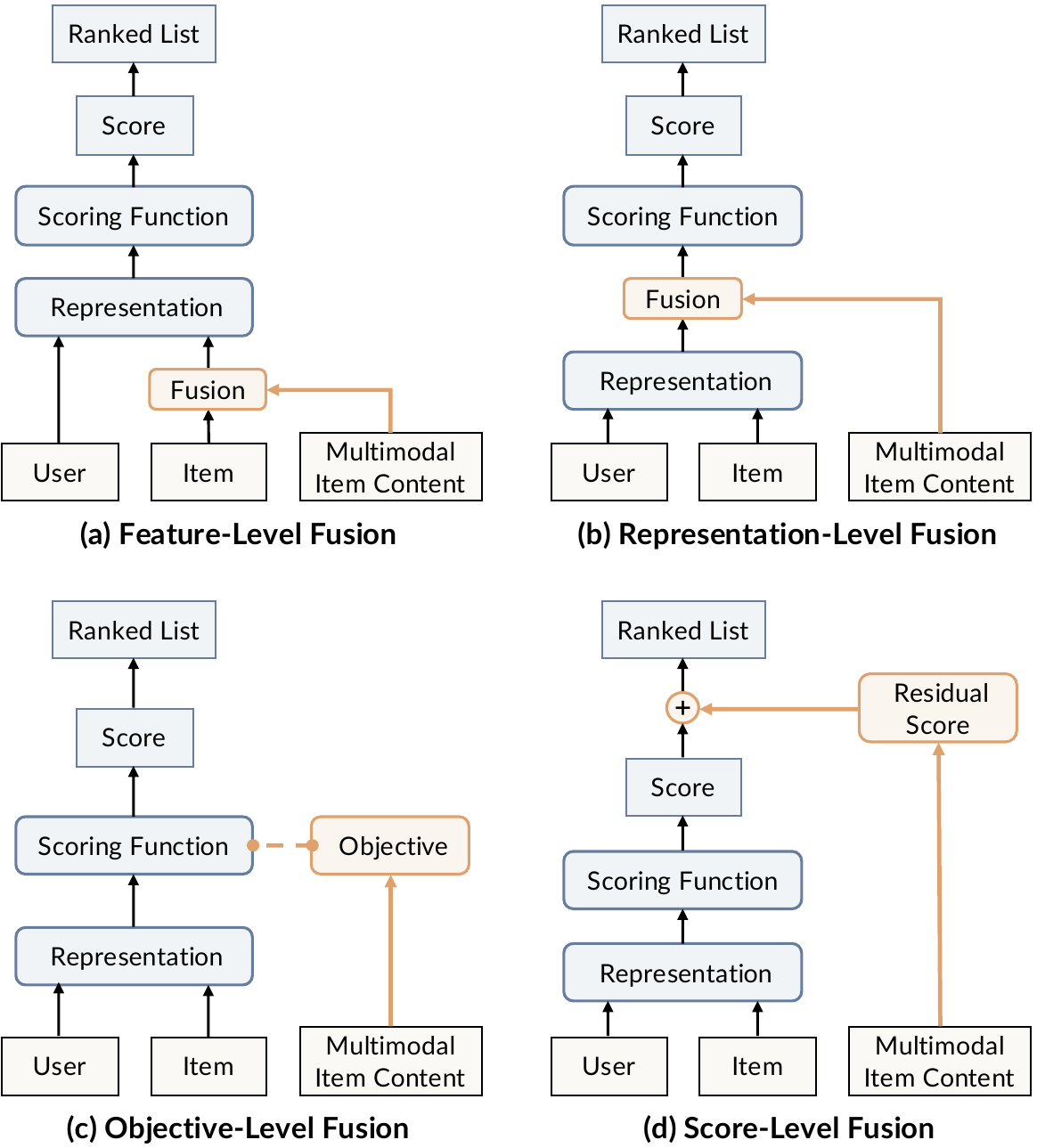}
  \caption{Fusion levels in FedMRs.
  (a) Feature-level fusion combines inputs.
  (b) Representation-level fusion combines representations.
  (c) Objective-level fusion uses content in training objectives.
  (d) Score-level fusion combines scores, including SERA's semantic residual scores.
  Blue and orange show the collaborative and content paths, and dashed links show training signals.}
  \label{fig:fusion-taxonomy}
\end{figure}

Fig.~\ref{fig:fusion-taxonomy} illustrates how Federated Multimodal Recommenders (FedMRs) add this content.
(a)~Feature-level fusion combines content with the raw inputs through early fusion, as in AMMFRS \cite{feng2024robust}.
The model can then learn fine-grained cross-modal interactions, and the FedRS has to be rebuilt and retrained.
(b)~Representation-level fusion merges content into the learned representations through intermediate fusion \cite{zhou2026comprehensive}.
FedVLR \cite{li2024fedvlr} uses it to adapt content to each user with fusion parameters on every client.
(c)~Content also enters training through extra loss terms \cite{zhao2024fusion}, and we treat this as objective-level fusion.
FedMMR \cite{li2024towards} and GFMFR \cite{zhang2026multimodal} add alignment and distillation losses, which alter the training objective of the FedRS and leave inference unchanged.
All three levels build content into the model behind the collaborative scores, so the final ranking mixes the effect of content with the collaborative signal.

Score-level fusion in Fig.~\ref{fig:fusion-taxonomy}(d) instead keeps the collaborative scores and adds a content score to each.
An item then rises above another only when its content score exceeds the other's by more than their collaborative score gap.
Each rank change thus traces to an item pair, and capping these scores limits how far content can reach into the collaborative ranking.

We therefore propose SERA, a federated \textbf{SE}mantic \textbf{R}esidual \textbf{A}dapter that brings score-level late fusion to existing FedRSs.
SERA keeps content out of the inputs, representations, and training objective of a FedRS and adds a semantic residual score to each collaborative score.
Collaborative scores weight each user's history items to form a personalized content query, and each candidate's semantic residual score grows with its content similarity to this query.
A residual budget, scaled by the standard deviation of the user's collaborative scores, caps each semantic residual score at a user-adaptive residual bound.

Combining the pairwise condition with this cap, we prove that, for any content, an item outside the collaborative top-$K$ list can enter the final top-$K$ list only if its collaborative score is within twice the residual bound of the $K$-th item.
The bound can trade off reach against accuracy, since a larger bound lets content lift more items, including some irrelevant ones.

Our contributions span method, theory, and experiments:
\begin{itemize}
\item We propose SERA, a federated adapter that adds semantic residual scores to the collaborative scores of an existing FedRS. SERA uses these scores both to weight each user's history items and to scale the residuals per user.
\item Every rank change under SERA traces to item pairs, and we prove a reach condition that holds for any content.
\item On six datasets, SERA improves all three FedRSs with shared parameters alone, beats existing FedMRs in most comparisons, and shows the predicted trade-off.
\end{itemize}

\section{Related Work}\label{sec:related}

\subsection{Federated Recommendation Systems}
FedRSs learn from interaction histories stored on user devices by aggregating client updates \cite{mcmahan2017communication, pillutla2022partial}.
Federated Collaborative Filtering keeps user factors local and aggregates item-factor gradients to learn shared item representations \cite{ammad2019federated}.
FedNCF extends this design to neural models and aggregates sparse item embeddings separately from dense neural parameters \cite{perifanis2022federated}.
Differences in client interaction distributions motivate personalized models \cite{kairouz2021advances,collins2021fedrep}.
PFedRec retains a local user model and adapts shared item representations to local interactions \cite{zhang2023dual}.
FedRAP decomposes item representations into globally shared and user-specific components \cite{li2024fedrap}, while GPFedRec builds a user-relation graph from client-specific item embeddings for personalized aggregation \cite{zhang2024gpfedrec}.
To accommodate device resources, HeteFedRec transfers knowledge across models with different embedding dimensions through dual-task learning and relation-based distillation \cite{yuan2024hetefedrec}.
These FedRSs provide the collaborative scores that SERA builds on.

\subsection{Multimodal Recommendation}
Multimodal recommenders supplement interaction supervision with item text and images \cite{zhang2023beyond,wei2019mmgcn}.
VBPR projects visual features into a user-specific visual preference term added to collaborative scores \cite{he2016vbpr}, while ACF uses item and component attention to aggregate multimodal interaction histories \cite{chen2017attentive}.
In item-based collaborative filtering, NAIS weights history items by relevance to the candidate \cite{he2018nais}. SERA weights history content by collaborative scores to form one query for all candidates.
MMCN refines item representations through attention-based residuals and hierarchical contrastive learning \cite{wang2025mmcn}.
MICRO learns modality-aware item graphs and combines shared and modality-specific signals through contrastive learning \cite{zhang2023micro}, while FREEDOM freezes content-derived graphs and denoises interaction graphs \cite{zhou2023freedom}.
BM3 combines interaction reconstruction with cross-modal alignment \cite{zhou2023bootstrap}, and MMSSL contrasts collaborative and multimodal views \cite{wei2023mmssl}.
DropoutNet reconstructs cold-item representations from content \cite{volkovs2017dropoutnet}, while CLCRec aligns content and collaborative spaces \cite{wei2021clcrec}.
PCKD reduces preference inconsistencies across distillation projection through a ranking regularization term \cite{zhu2025pckd}.

Content reach also relates to studies that limit changes to existing lists.
Additive corrections to a frozen ranker and score regularization trade accuracy gains against the share of queries whose rankings change \cite{li2020stabilizing}.
CertDR uses randomized smoothing to certify top-$K$ exclusion under word-substitution attacks \cite{wu2022certified}.
Reachability analyses ask whether a user can make an item recommended by changing their own interactions \cite{dean2020recommendations,curmei2021quantifying}.
SERA instead bounds the collaborative score gaps that semantic residual scores can overcome for any content.

\subsection{Federated Multimodal Recommendation}
FedMRs integrate item content under federated resource and privacy constraints \cite{li2026survey}.
Their designs differ in fusion level and in the item representations that clients and servers exchange.
IFedRec learns an attribute network from aggregated item embeddings and aligns its output with local collaborative representations to serve cold-start items \cite{zhang2024ifedrec}.
AMMFRS fuses text, image, identity, and behavior features through client-side attention \cite{feng2024robust}.
FedMMR aligns a server-side multimodal model with lightweight collaborative models on clients \cite{li2024towards}.
FedVLR generates vision-language views on the server and combines them through a client-specific mixture of experts \cite{li2024fedvlr}.
GFMFR uses server-side group-level multimodal representations as distillation targets for local training \cite{zhang2026multimodal}.
SERA instead fuses content at the score level and personalizes its semantic residual scores with local interaction histories.

\section{Preliminaries}\label{sec:prelim}

\subsection{Notation}

Let $\mathcal{U}$ contain $N$ users and let $\mathcal{I}$ contain $M$ items. We use $u\in\mathcal{U}$ for a user and $i\in\mathcal{I}$ for an item.
The candidate set $\mathcal{C}_u\subseteq\mathcal{I}$ contains the items ranked for user $u$, and $M_u=|\mathcal{C}_u|$ is its size.
A FedRS maps each user-item pair $(u,i)$ to a collaborative score $\anchor(u,i)=\anchor(u,i;\theta^{(u)})$.
The parameter $\theta^{(u)}=[\theta_{\mathrm{s}},\theta_u]$ contains shared parameters $\theta_{\mathrm{s}}$ and personalized parameters $\theta_u$.
Client $u$ stores its interaction history $\mathcal{H}_u$ on the device.
The client updates $\theta^{(u)}$ locally and sends the update associated with the shared parameters $\theta_{\mathrm{s}}$ to the server.

\subsection{Federated Recommendation Systems}

For user $u$, let $i^{+}\in\mathcal{H}_u$ be an observed item and $i^{-}$ be a sampled negative item.
Under the Bayesian Personalized Ranking (BPR) loss \cite{rendle2009bpr}, a score function $s$ has client risk:
\begin{equation}
\mathcal{L}_u(s) =
\mathbb{E}_{
  i^{+}\sim \mathcal{H}_u,\; i^{-}\sim \mathcal{D}
}
\!
\left[
  -\log \sigma\!\big( s(u,i^{+}) - s(u,i^{-}) \big)
\right].
\label{eq:clientrisk}
\end{equation}
Here, $\mathcal{D}$ is the negative-item sampling distribution, and $\sigma(\cdot)$ is the logistic function.
The federated objective aggregates the client risks using a nonnegative weight $\alpha_u$ for each user:
\begin{equation}
  \mathcal{L}(s) = \sum_{u} \alpha_u \mathcal{L}_u(s),
  \;
  \text{where } \sum_u \alpha_u = 1.
\label{eq:fedloss}
\end{equation}
SERA trains its semantic residual scores on this aggregated risk.
Each FedRS retains its own training objective and its own rules for updating and sharing model parameters.

\section{Methodology}
\label{sec:method}

SERA in Fig.~\ref{fig:framework} fuses item content with an existing FedRS at the score level.
It learns semantic residual scores from local histories and caps them with a global residual budget.

\begin{figure}[!t]
  \centering
  \includegraphics[width=.9\columnwidth]{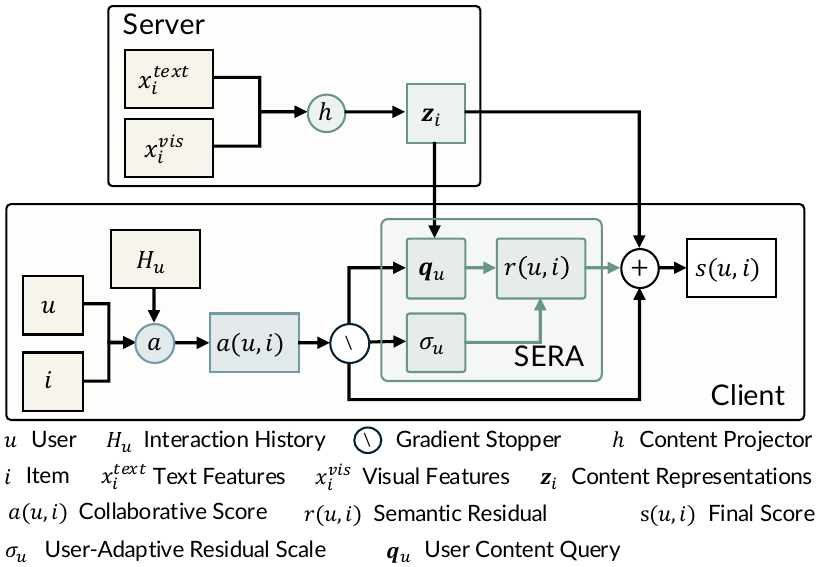}
  \caption{Personalized content-based scoring in SERA.
  A shared projector maps item text and images to content representations.
  Collaborative scores weight history items to form a content query. 
  The standard deviation of these scores across items sets the user-adaptive residual scale.
  SERA compares candidates with this query and adds bounded semantic residual scores to the collaborative scores.
  Local histories and collaborative scores personalize these scores.}
  \label{fig:framework}
\end{figure}

\subsection{Residual Scoring}
\label{subsec:reranking}

SERA starts from the collaborative score $\anchor(u,i)$ that an ID-based FedRS assigns to pair $(u,i)$.
SERA forms the final score $s(u,i)$ by adding to $\anchor(u,i)$ a \textit{semantic residual score} $r(u,i)$ derived from the text and images of the candidate item $i$:
\begin{equation}
s(u,i) = \sg\!\big[\anchor(u,i)\big] + r(u,i).
\label{eq:score}
\end{equation}
The stop-gradient operator $\sg[\cdot]$ keeps the value of the collaborative score and blocks its gradient, so SERA's loss never updates the FedRS. Item content can thus change the final ranking only through the semantic residual score $r(u,i)$.

\subsection{Item Content Representations}
Each item $i$ provides processed text features $\mathbf{x}_i^{\mathrm{text}}$ and visual features $\mathbf{x}_i^{\mathrm{vis}}$.
The shared projector $h_{\phi}$ with parameters $\phi$ maps their concatenation to normalized content representation $\mathbf{z}_i$:
\begin{equation}
\mathbf{z}_i = \Pi\!\big( h_{\phi}([\mathbf{x}_i^{\mathrm{text}}, \mathbf{x}_i^{\mathrm{vis}}]) \big).
\label{eq:basis}
\end{equation}
In Eq.~\eqref{eq:basis}, $\Pi(\mathbf{v})=\mathbf{v}/\max\{\|\mathbf{v}\|_2,\varepsilon\}$ is L2 normalization, and $\varepsilon>0$ avoids division by zero.
During training, the item features remain frozen while $h_{\phi}$ learns content representations.

\subsection{History-Weighted Content Query}
SERA combines the representations $\mathbf{z}_i$ of history items into a \emph{content query}.
Collaborative scores set the \textit{history weights}:
\begin{equation}
\beta_{ui}
=
\frac{\exp\!\big(\sg[\anchor(u,i)]/\tau_q\big)}
{\sum_{j\in\mathcal{H}_u}\exp\!\big(\sg[\anchor(u,j)]/\tau_q\big)},
\quad
\text{where } i\in\mathcal{H}_u.
\label{eq:weights}
\end{equation}
The temperature $\tau_q>0$ controls how strongly the history weights favor high-scoring history items.
These weights combine the history items' content into the content query $\mathbf{q}_u$:
\begin{equation}
\mathbf{q}_u
=
\Pi\!\left(\sum_{i\in\mathcal{H}_u}\beta_{ui}\mathbf{z}_i\right).
\label{eq:query}
\end{equation}
Users share the projector parameters, while their local histories and collaborative scores determine personalized queries.

\subsection{User-Adaptive Bounded Residual}
\label{subsec:residual}
SERA compares the content query $\mathbf{q}_u$ with each content representation $\mathbf{z}_i$ through their content similarity $c_{ui}$:
\begin{equation}
c_{ui} = \langle \mathbf{q}_u, \mathbf{z}_i \rangle .
\label{eq:compat}
\end{equation}
As both vectors are normalized, $c_{ui}\in[-1,1]$.
The \textit{residual function} maps this content similarity to the \textit{shape term} $g_{ui}$:
\begin{equation}
g_{ui} = \tanh\!\left( \frac{c_{ui}}{\tau_r} \right).
\label{eq:shape}
\end{equation}
The residual temperature $\tau_r>0$ controls saturation. 
A small value drives $g_{ui}$ toward $\pm1$, and a large value keeps it close to linear in $c_{ui}$.
Let $v_u$ be the population standard deviation of user $u$'s collaborative scores over the full item set $\mathcal I$. SERA then multiplies $g_{ui}$ by a residual budget $\kappa\ge0$ and the user-adaptive residual scale $\sigma_u=\max\{v_u,\varepsilon\}>0$, which gives:
\begin{equation}
r(u,i) = \kappa \, \sigma_u \, g_{ui}.
\label{eq:residual}
\end{equation}
\begin{proposition}[User-Adaptive Residual Bound]
\label{prop:residual-bound}
For $\kappa>0$ and $\tau_r>0$, the semantic residual score is bounded as follows:
\begin{equation}
|r(u,i)| \le \kappa\,\sigma_u\tanh\!\left(\frac{1}{\tau_r}\right) < \kappa\,\sigma_u,
\;
\forall i\in\mathcal{I}.
\label{eq:bound}
\end{equation}
\end{proposition}
\begin{proof}[Proof Sketch]
The Cauchy--Schwarz inequality applied to Eq.~\eqref{eq:compat} gives $|c_{ui}|\le1$.
The oddness and strict monotonicity of $\tanh$ in Eqs.~\eqref{eq:shape} and \eqref{eq:residual} give the bound in Eq.~\eqref{eq:bound}.
\end{proof}
We denote this residual bound by $R_u = \kappa\,\sigma_u\tanh(1/\tau_r)$.
When $\kappa=0$, the semantic residual score vanishes and $s(u,i)=\anchor(u,i)$.
The analysis below assumes $\kappa>0$.

\subsection{Objective Function}
Let $\mathcal{L}_{\mathrm{base}}(\anchor)$ be the FedRS's training objective for collaborative scores $\anchor(u,i)$.
SERA's residual objective applies the client risk in Eq.~\eqref{eq:clientrisk} to the final score $s$ of each user:
\begin{equation}
\mathcal{J}(\phi)
=
\lambda\sum_{u\in\mathcal{U}}\alpha_u\mathcal{L}_u(s),
\;
\text{where }
\sum_{u\in\mathcal{U}}\alpha_u=1.
\label{eq:objective_sera}
\end{equation}
Here, $\lambda>0$ scales the residual loss.
SERA applies $\sg[\cdot]$ to every collaborative score that enters $s(u,i)$, $\beta_{ui}$, or $\sigma_u$, so each residual update treats the collaborative scores as constants and updates only $\phi$.
Adding $\mathcal{J}(\phi)$ in Eq.~\eqref{eq:objective_sera} to $\mathcal{L}_{\mathrm{base}}(\anchor)$ forms a joint objective for the FedRS and the shared projector:
\begin{equation}
\min_{\theta_{\mathrm{s}},\{\theta_u\}_{u \in \mathcal{U}},\phi} \,
\mathcal{L}_{\mathrm{base}}(\anchor) + \mathcal{J}(\phi).
\label{eq:objective}
\end{equation}
Eq.~\eqref{eq:objective} trains the FedRS on interactions alone and the projector through the final scores.
SERA therefore gives any FedRS personalized semantic residual scores bounded in proportion to the standard deviation of each user's collaborative scores.

\section{Algorithm Analysis}
\label{sec:theory}

\subsection{Federated Optimization}
Algorithm~\ref{alg:fedsera} splits the optimization of Eq.~\eqref{eq:objective} between clients and the server, which holds the item features and the projector $h_\phi$.
At round $t$, the server computes the content table $\mathbf{Z}^t=[\mathbf{z}_i]_{i\in\mathcal{I}}$ by Eq.~\eqref{eq:basis}, samples clients $\mathcal{S}_t$, and normalizes their data counts into weights $\alpha_{u,t}$ that sum to one.
Each sampled client receives $\mathbf{Z}^t$ with the shared collaborative parameters, performs $E$ local steps of the FedRS, and keeps its personalized parameters on the device.
Over these updates, it accumulates the gradient $\mathbf{G}_u^t$ of its residual loss with respect to $\mathbf{Z}^t$ and uploads it with its collaborative model update.

The server applies the FedRS's aggregation rule, averages the uploaded gradients with weights $\alpha_{u,t}$, and back-propagates the average to update $h_\phi$.
Every local step reads the table built from $\phi^t$. For fixed collaborative parameters, the update thus follows an unbiased gradient of Eq.~\eqref{eq:objective_sera} with $\alpha_u$ replaced by $\mathbb{E}[\alpha_{u,t}]$.
If this objective is smooth with bounded gradient variance and the step size scales as $T^{-1/2}$, standard nonconvex analysis~\cite{bottou2018optimization} shows that its average squared gradient norm is $O(T^{-1/2})$.
Algorithm~\ref{alg:fedsera} computes $\sigma_u$ from all $M$ collaborative scores.
At inference, the client builds its query from its history and the stored table and ranks the candidate set $\mathcal{C}_u$ by $s(u,i)$.

\begin{algorithm}[t]
  \caption{Training a FedRS with SERA}
  \label{alg:fedsera}
  \begin{algorithmic}[1]
  \REQUIRE Hyperparameters of SERA $(h_\phi,\kappa,\tau_r,\tau_q,\lambda)$, item features $\{\mathbf{x}_i^{\mathrm{text}}, \mathbf{x}_i^{\mathrm{vis}}\}$ on the server, and a FedRS
  \REQUIRE Communication rounds $T$, local steps $E$, and server learning rate $\eta_{\mathrm{s}}$
  \ENSURE Shared collaborative parameters $\theta_{\mathrm{s}}^{T}$, SERA projector $\phi^{T}$, and personalized parameters $\{\theta_u^{T}\}_{u\in\mathcal{U}}$
  \STATE Initialize the FedRS and the SERA projector $\phi^{0}$
  \FOR{$t=0$ \TO $T-1$}
    \STATE Compute $\mathbf{Z}^{t}=[\mathbf{z}_i]_{i\in\mathcal{I}}$ with $h_{\phi^{t}}$ by Eq.~\eqref{eq:basis}
    \STATE Sample clients $\mathcal{S}_t\subseteq\mathcal{U}$ and set weights $\alpha_{u,t}$
    \STATE Broadcast $\theta_{\mathrm{s}}^{t}$ and $\mathbf{Z}^{t}$ to $\mathcal{S}_t$
    \FORALL{clients $u\in\mathcal{S}_t$ \textbf{in parallel}}
      \STATE Set $\theta_{\mathrm{s}}^{(u)}\leftarrow\theta_{\mathrm{s}}^{t}$ and $\mathbf{G}_u^{t}\leftarrow\mathbf{0}$
      \FOR{$e=1$ \TO $E$}
        \STATE Draw positives $\mathcal{B}_u^{+}$ and sampled negatives $\mathcal{B}_u^{-}$
        \STATE Compute $\anchor(u,\mathcal{I})$ with the FedRS, then $\sigma_u$
        \STATE Compute $\mathbf{q}_u$ by Eqs.~\eqref{eq:weights} and \eqref{eq:query} from rows of $\mathbf{Z}^{t}$
        \STATE Compute $s(u,i)$ for $i\in\mathcal{B}_u^{+}\cup\mathcal{B}_u^{-}$ by Eq.~\eqref{eq:score}
        \STATE Update $\mathbf{G}_u^{t}\leftarrow\mathbf{G}_u^{t}+\partial\big(\lambda\mathcal{L}_u(s)\big)/\partial\mathbf{Z}^{t}$
        \STATE Update $\theta_{\mathrm{s}}^{(u)}$ and $\theta_u$ by one FedRS step on $\mathcal{L}_{\mathrm{base}}$
      \ENDFOR
      \STATE Upload $\theta_{\mathrm{s}}^{(u)}$ and $\mathbf{G}_u^{t}$ to the server
    \ENDFOR
    \STATE $\theta_{\mathrm{s}}^{t+1}\leftarrow\operatorname{Aggregate}\big(\{\theta_{\mathrm{s}}^{(u)}\}_{u\in\mathcal{S}_t},\{\alpha_{u,t}\}_{u\in\mathcal{S}_t}\big)$
    \STATE $\phi^{t+1}\leftarrow\phi^{t}-\eta_{\mathrm{s}}\sum_{i\in\mathcal{I}}\big(\partial\mathbf{z}_i/\partial\phi\big)^{\top}\big[\sum_{u\in\mathcal{S}_t}\alpha_{u,t}\mathbf{G}_u^{t}\big]_i$
  \ENDFOR
  \end{algorithmic}
\end{algorithm}

\subsection{Content Reach Boundaries}
\label{subsec:reach}

The residual bound restricts ranking changes after training.
For user $u$ and items $i,j\in\mathcal{I}$, the \textit{collaborative score gap} is:
\begin{equation}
\gamma_{u,ij} = \anchor(u,i)-\anchor(u,j).
\label{eq:pairgap}
\end{equation}
According to Eq.~\eqref{eq:pairgap}, for items with different collaborative scores, a strict ranking change requires a residual difference with the sign opposite to this gap and a larger magnitude.

To link pairwise ranking changes to list membership, the collaborative and final rankings share the candidate set $\mathcal{C}_u$ and a fixed tie-breaking rule.
For $1\le K< M_u$, let $i_{u,K}$ be the item at collaborative rank $K$, with \textit{collaborative top-$K$ boundary score} $\anchor^{(K)}_{u}=\anchor(u,i_{u,K})$.
For a candidate $j$ outside the collaborative top-$K$ list, its \textit{collaborative boundary gap} is:
\begin{equation}
\gamma^{(K)}_{u,j} = \gamma_{u,i_{u,K}j} = \anchor^{(K)}_{u} - \anchor(u,j).
\label{eq:gap}
\end{equation}
Eq.~\eqref{eq:gap} measures how far $j$ falls below the collaborative top-$K$ boundary score.
To locate ranking changes, we fix the trained model, history, and candidate set and vary $\kappa>0$.
The collaborative scores, user-adaptive residual scale, and shape terms remain constant, so each pairwise score difference is affine in the budget.
Lemma~\ref{lem:flip} gives the change points.

\begin{lemma}[Change Point]
\label{lem:flip}
For user $u \in \mathcal{U}$ and items $i,j \in \mathcal{I}$, let $\delta_{u,ij}=g_{ui}-g_{uj}\ne 0$ be their shape-term difference.
The \emph{change point} of the pair $(i,j)$ is then defined as follows:
\begin{equation}
\kappa_{u,ij} = -\frac{\gamma_{u,ij}}{\sigma_u\,\delta_{u,ij}} .
\label{eq:flippoint}
\end{equation}
With this change point, the score difference factors as:
\begin{equation}
s(u,i)-s(u,j) = \sigma_u\,\delta_{u,ij}\big(\kappa-\kappa_{u,ij}\big).
\label{eq:flip}
\end{equation}
Since $\sigma_u>0$, the pair changes ranking at some positive budget exactly when $\gamma_{u,ij}$ and $\delta_{u,ij}$ have opposite signs.
\end{lemma}
\begin{proof}[Proof Sketch]
Eqs.~\eqref{eq:score} and \eqref{eq:residual} give $s(u,i)-s(u,j)=\gamma_{u,ij}+\kappa\sigma_u\delta_{u,ij}$, which is affine in $\kappa$ with slope $\sigma_u\delta_{u,ij} \neq 0$.
Factoring out the slope yields Eq.~\eqref{eq:flip}, where only the factor $\kappa-\kappa_{u,ij}$ depends on $\kappa$, and it changes sign at $\kappa_{u,ij}$.
\end{proof}

Each pair thus changes ranking at most once. Its change point depends on the collaborative score gap, user-adaptive residual scale, and learned shape-term difference.
Each semantic residual score has magnitude at most $R_u$, so the residual difference of two items has magnitude at most $2R_u$.
Comparing this bound with the collaborative score gap gives a sufficient condition for keeping the order of two items.

\begin{corollary}[Order-Preservation Certificate]
\label{cor:certificate}
For user $u$ and items $i,j\in\mathcal{I}$, if $|\gamma_{u,ij}| > 2R_u$, then $i$ and $j$ keep their collaborative order under $s$.
The threshold is tight: if $\anchor(u,j)-\anchor(u,i)=2R_u$, $c_{ui}=1$, and $c_{uj}=-1$, then $s(u,i)=s(u,j)$.
\end{corollary}
\begin{proof}[Proof Sketch]
By Eq.~\eqref{eq:bound} and the triangle inequality, $|\kappa\sigma_u\delta_{u,ij}|=|r(u,i)-r(u,j)|\le 2R_u<|\gamma_{u,ij}|$, so adding the semantic residual scores preserves the sign of $\gamma_{u,ij}$.
\end{proof}

An outside item enters the top-$K$ list only by outranking an item of the collaborative top-$K$ list, so Corollary~\ref{cor:certificate} bounds the collaborative boundary gap of the entering item.

\begin{theorem}[Reach Condition]
\label{thm:reach}
Let $j\in\mathcal{C}_u$ lie outside the collaborative top-$K$ list and enter the final one, both formed from $\mathcal{C}_u$ with fixed tie-breaking.
Its collaborative boundary gap in Eq.~\eqref{eq:gap} must then satisfy the following \textit{reach condition}:
\begin{equation}
\gamma^{(K)}_{u,j} \le 2R_u .
\label{eq:reachcond}
\end{equation}
\end{theorem}
\begin{proof}[Proof Sketch]
Let $\mathcal{T}_u\subseteq\mathcal{C}_u$ be the collaborative top-$K$ set.
Since $j\notin\mathcal{T}_u$ enters the final top $K$, at least one item $i\in\mathcal{T}_u$ is displaced.
Thus $s(u,j)\ge s(u,i)$, and in a tie the fixed tie-breaking rule favors $j$.
By Eq.~\eqref{eq:bound} and $\anchor(u,i)\ge\anchor^{(K)}_{u}$:
\begin{equation}
  s(u,j)\le\anchor(u,j)+R_u,\;
  s(u,i)\ge\anchor(u,i)-R_u\ge\anchor^{(K)}_{u}-R_u.
\end{equation}
Combining these inequalities gives $\anchor(u,j)+R_u\ge\anchor^{(K)}_{u}-R_u$, so the collaborative boundary gap satisfies $\gamma^{(K)}_{u,j}\le2R_u$.
\end{proof}

\begin{figure}[!t]
  \centering
  \includegraphics[width=.9\columnwidth]{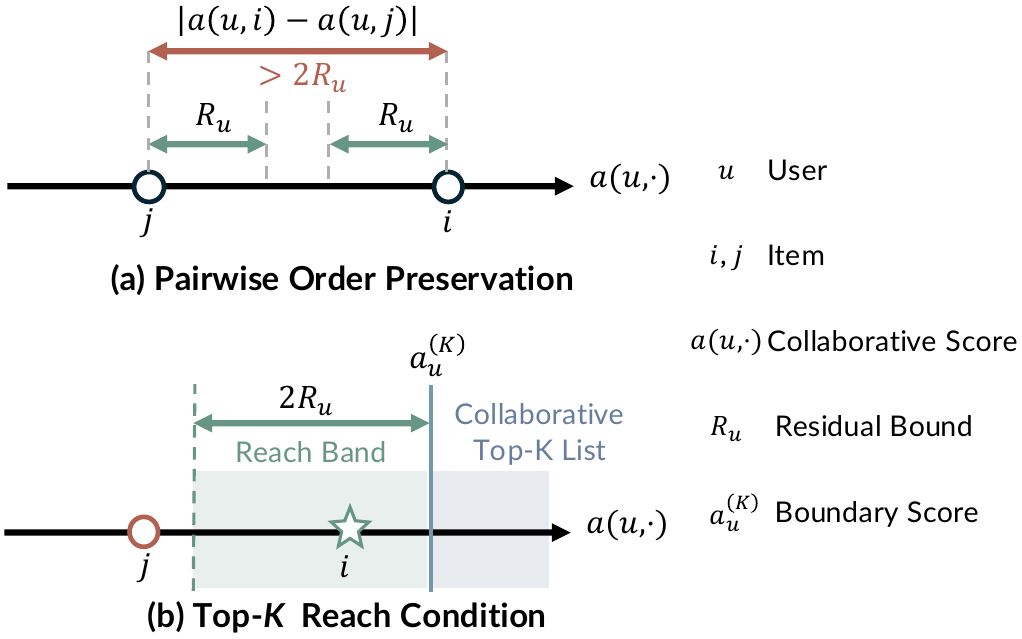}
  \caption{Limits on ranking changes under SERA's residual bound.
  (a) Two items retain their collaborative order when their collaborative score gap exceeds the maximum residual difference $2R_u$.
  (b) An item outside the collaborative top-$K$ list can enter only if its collaborative boundary gap is at most $2R_u$.
  Within this $2R_u$ range, semantic residual scores decide entry.}
  \label{fig:mechanism}
\end{figure}

The reach band in Fig.~\ref{fig:mechanism} contains the outside candidates with collaborative scores between $\anchor_u^{(K)}-2R_u$ and $\anchor_u^{(K)}$.
The same argument bounds the final rank of every item.
Let $\rho_u(j)$ and $\tilde\rho_u(j)$ denote the collaborative and final ranks of $j$ in $\mathcal{C}_u$.

\begin{corollary}[Content-Agnostic Rank Range]
\label{cor:rank-range}
For user $u$ and item $j\in\mathcal{C}_u$, let $n^{+}_{u,j}$ count the candidates with collaborative scores above \mbox{$\anchor(u,j)+2R_u$}, and $n^{-}_{u,j}$ those below \mbox{$\anchor(u,j)-2R_u$}.
For any content and projector parameters:
\begin{equation}
1+n^{+}_{u,j}\le\tilde\rho_u(j)\le M_u-n^{-}_{u,j}.
\label{eq:rank-range}
\end{equation}
\end{corollary}
\begin{proof}[Proof Sketch]
Normalization in Eqs.~\eqref{eq:basis} and \eqref{eq:query} keeps Proposition~\ref{prop:residual-bound} valid for any content, and Corollary~\ref{cor:certificate} then keeps the order between $j$ and each counted candidate unchanged.
\end{proof}

Theorem~\ref{thm:reach} is the top-$K$ case of Eq.~\eqref{eq:rank-range}, because $j$ enters only if $n^{+}_{u,j}<K$, which holds exactly when $\anchor^{(K)}_u-\anchor(u,j)\le2R_u$.
Since content cannot change the collaborative scores, the range also limits how high untrusted content or projector parameters can raise an item.

\subsection{Ranking Analysis}
\label{subsec:ranking-analysis}
We analyze when Theorem~\ref{thm:reach} limits the ranking and how pairwise ranking changes move each item's rank as the budget grows.
Because $\sigma_u$ is at least the standard deviation of the $M$ collaborative scores, every collaborative score gap satisfies $|\gamma_{u,ij}|\le\sigma_u\sqrt{2M}$.
The coverage condition $2\kappa\tanh(1/\tau_r)\ge\sqrt{2M}$ thus makes Theorem~\ref{thm:reach} and Corollary~\ref{cor:certificate} vacuous.
It depends only on $\kappa$, $\tau_r$, and $M$, so it can be checked before training.
Actual ranking changes depend on the learned residual differences between competing items.
For a fixed model and budget $\kappa$, a target item $i\in\mathcal{C}_u$ has two rank counts:
\begin{equation}
\begin{aligned}
P_i &= \big|\{\,j\ne i:\ \anchor(u,j)>\anchor(u,i),\ s(u,j)<s(u,i)\,\}\big|, \\
Q_i &= \big|\{\,j\ne i:\ \anchor(u,j)<\anchor(u,i),\ s(u,j)>s(u,i)\,\}\big|.
\end{aligned}
\label{eq:pq}
\end{equation}
Both counts range over $j\in\mathcal{C}_u\setminus\{i\}$.
They record strict ranking changes: $P_i$ counts competitors moved below $i$, and $Q_i$ counts competitors moved above it.
Moving $j$ below $i$ adds one to $P_i$ and one to $Q_j$, so $P_i$ and $Q_i$ have equal sums over $i\in\mathcal{C}_u$.
As $\kappa$ grows, these counts give the exact rank of the target.

\begin{proposition}[Rank Path]
\label{prop:rank-path}
Fix user $u$, a trained model, and a target $i\in\mathcal{C}_u$ whose collaborative score differs from all others.
For every budget $\kappa>0$ other than the change points $\kappa_{u,ij}$ of $i$, $P_i$ counts the competitors with higher collaborative scores whose change points with $i$ lie in $(0,\kappa)$, and $Q_i$ counts those with lower ones. The final rank of target $i$ then satisfies:
\begin{equation}
\tilde\rho_u(i)=\rho_u(i)-P_i+Q_i.
\label{eq:rank-path}
\end{equation}
For $\kappa$ above all positive change points of $i$, its final rank is its rank by similarity $c_{ui}$, and collaborative order breaks ties.
\end{proposition}
\begin{proof}[Proof Sketch]
The items above $i$ in the final ranking are those with higher collaborative scores that $i$ has not passed and those with lower ones that have passed it.
By Lemma~\ref{lem:flip}, each of these pairs changes ranking only at a positive change point.
\end{proof}

By Proposition~\ref{prop:rank-path}, an item $j$ with $\rho_u(j)>K$ lies in the final top-$K$ list exactly when $P_j-Q_j\ge\rho_u(j)-K$, which refines the reach condition of Theorem~\ref{thm:reach} into an exact test.

An item can also rise and later fall.
Each positive change point with a higher competitor raises the item by one position, and each one with a lower competitor lowers it by one, so an outside item can enter the top-$K$ list and later leave it.

\subsection{Complexity Analysis}

\paragraph{Time Complexity}
Let $\mathcal{B}_u=\mathcal{B}_u^{+}\cup\mathcal{B}_u^{-}$ hold the mini-batch items, let $d_x$ and $d$ be the projector's input and output dimensions, and let the projector have $|\phi|$ parameters, with $|\phi|=\Theta(d_xd)$ for a linear projection.
After computing all $M$ collaborative scores, a local step of Algorithm~\ref{alg:fedsera} adds $\Theta(M+(|\mathcal{H}_u|+|\mathcal{B}_u|)d)$ client time.
The $\Theta(M)$ term computes $\sigma_u$, and the $\Theta((|\mathcal{H}_u|+|\mathcal{B}_u|)d)$ term reads the history and batch rows of $\mathbf{Z}^t$, builds the query, scores the batch items, and accumulates their gradients.
If the FedRS update scores only sampled items, scoring all $M$ items for $\sigma_u$ is extra work.
Each round, the server evaluates and differentiates $h_\phi$ on the item set in $\Theta(M|\phi|)$ time and averages the uploaded gradients in $\Theta(|\mathcal{S}_t|Md)$ time.
Given the collaborative scores, recommendation over $\mathcal{C}_u$, including the FedRS's own top-$K$ selection, takes $\Theta(M+(|\mathcal{H}_u|+M_u)d+M_u\log K)$ time.

\paragraph{Space Complexity}
Each client stores the content table in $\Theta(Md)$ space and uses $\Theta(M+(|\mathcal{H}_u|+|\mathcal{B}_u|)d+Md)$ working memory for the collaborative scores, the query and batch rows, and the accumulated gradient.
The server stores the features, the projector with its optimizer state, and the full content table in $\Theta(Md_x+|\phi|+Md)$ space in total.

\paragraph{Communication Complexity}
Beyond the FedRS's communication, broadcasting $\mathbf{Z}^t$ and uploading $\mathbf{G}_u^t$ add $\Theta(|\mathcal{S}_t|Md)$ values per round in each direction. The item features and the projector stay on the server, so this cost does not grow with history lengths or the projector size $|\phi|$.

\subsection{Privacy Discussion}
\label{subsec:privacy}
SERA keeps user interaction histories $\mathcal{H}_u$ and personalized parameters $\theta_u$ on the client, where each client builds its content query $\mathbf{q}_u$, computes the semantic residual scores $r(u,i)$, and ranks items by $s(u,i)$.
The server receives the collaborative updates and the uploaded gradients $\mathbf{G}_u^t$ of Algorithm~\ref{alg:fedsera}.
The nonzero rows of $\mathbf{G}_u^t$ can reveal indices of history or mini-batch items, as the FedRS's item updates do.
The projector update needs only the weighted sum of $\mathbf{G}_u^t$, so these gradients can be combined through secure aggregation~\cite{kairouz2021advances}.
We evaluate the effects of client-side clipping and Gaussian noise in Sec.~\ref{sec:privacy-noise}.

\begin{table}[!t]
  \centering
  \caption{Statistics of the six datasets.
  \#Users and \#Items are the counts of users and items.
  \#Ratings is the number of observed interactions.
  Sparsity is $1-\#\text{Ratings}/(\#\text{Users}\times\#\text{Items})$.
  }
  \label{tab:datasets}
  \setlength{\tabcolsep}{8pt}
  \begin{tabular}{lrrrr}
  \toprule
  Dataset & \#Users & \#Items & \#Ratings & Sparsity \\
  \midrule
  Beauty & 253    & 356   & 2{,}535   & 97.19\% \\
  Cards  & 377    & 129   & 2{,}429   & 95.01\% \\
  ML     & 610    & 3{,}650 & 90{,}274 & 95.95\% \\
  Dance  & 10{,}231 & 1{,}676 & 80{,}086 & 99.53\% \\
  Food   & 5{,}990  & 1{,}125 & 36{,}482 & 99.46\% \\
  Movie  & 15{,}908 & 2{,}528 & 111{,}091 & 99.72\% \\
  \bottomrule
  \end{tabular}
\end{table}

\section{Experiments}
\label{sec:exp}

Following Secs.~\ref{sec:method} and~\ref{sec:theory}, we evaluate SERA's performance and test the implications of the reach and rank-path analysis.

\subsection{Experimental Setup}

\paragraph{Datasets}
We use six public recommendation datasets from three domains.
\textbf{Beauty} and \textbf{Cards} are the All Beauty and Gift Cards subsets of Amazon Reviews 2023\footnote{\url{https://amazon-reviews-2023.github.io/}} and contain product-review interactions and item metadata \cite{hou2024bridging}.
MovieLens-latest-small (\textbf{ML})\footnote{\url{https://grouplens.org/datasets/movielens/latest/}} contains movie ratings and item metadata \cite{harper2015movielens}.
\textbf{Dance}, \textbf{Food}, and \textbf{Movie} are short-video subsets of NineRec\footnote{\url{https://github.com/westlake-repl/NineRec}} with visual and textual content for each item \cite{zhang2024ninerec}.
We reserve one validation and one test interaction per user under a chronological leave-one-out split \cite{he2017neural}, and rank the test item against the full item set \cite{krichene2020sampled}.
We also retain users and items with at least five interactions before splitting to avoid cold-start cases.
Table~\ref{tab:datasets} reports the dataset statistics.

\paragraph{Metrics}
We measure ranking quality with Recall and Normalized Discounted Cumulative Gain (NDCG).
Recall@$K$ measures the proportion of users whose test item appears among the top-$K$ recommendations \cite{he2017neural}.
NDCG@$K$ gives higher scores when the test item ranks closer to the top \cite{jarvelin2002cumulated}.
We set $K=50$ and report means and standard deviations across five random seeds, with both metrics in percent and their differences in percentage points.
The hyperparameters of all methods are selected by NDCG@50 on the validation set.

\subsection{Baselines}
The baselines serve three roles.
\textbf{Centralized multimodal recommenders} show the accuracy when a model learns from interactions and item content in a centralized setting.
\begin{itemize}[leftmargin=*,itemsep=0pt,topsep=2pt]
\item \textbf{VBPR} \cite{he2016vbpr} extends BPR with pretrained visual item features for visually informed user preference estimation.
\item \textbf{BM3} \cite{zhou2023bootstrap} learns multimodal representations through self-supervised reconstruction without negative sampling.
\item \textbf{MGCN} \cite{yu2023multiview} denoises modality features and propagates them on modality-specific graphs for behavior-aware fusion.
\item \textbf{FITMM} \cite{yang2025fitmm} decomposes multimodal representations into frequency bands and regularizes their information content.
\item \textbf{PGL} \cite{yu2025principal} extracts principal user-item subgraphs to preserve local interaction structure during graph message passing.
\item \textbf{MSCA} \cite{li2026msca} aligns collaborative and modality-augmented views through semantic contrastive representation learning.
\end{itemize}
\textbf{ID-based FedRSs} provide the backbones for measuring the effect of adding multimodal item content.
\begin{itemize}[leftmargin=*,itemsep=0pt,topsep=2pt]
\item \textbf{FedAvg} \cite{mcmahan2017communication} aggregates shared model parameters through client-weighted averaging across participating clients.
\item \textbf{PFedRec} \cite{zhang2023dual} learns lightweight client models with locally personalized user and item representations on the device.
\item \textbf{FedRAP} \cite{li2024fedrap} decomposes item representations into shared and personalized components for federated personalization.
\end{itemize}
\textbf{FedMRs} provide federated multimodal alternatives that add item content at the feature, representation, or objective level.
\begin{itemize}[leftmargin=*,itemsep=0pt,topsep=2pt]
\item \textbf{FedMMR} \cite{li2024towards} aligns client collaborative and server multimodal item representations through a contrastive objective.
\item \textbf{AMMFRS} \cite{feng2024robust} fuses text, image, identity, and behavior features through client-side attention before federated training.
\item \textbf{FedVLR} \cite{li2024fedvlr} constructs server-side multimodal views and selects them through a client-specific mixture of experts.
\item \textbf{GFMFR} \cite{zhang2026multimodal} distills group-level multimodal preferences into each client's own local recommendation objective.
\end{itemize}
We attach SERA to the same FedAvg, PFedRec, and FedRAP backbones that FedVLR and GFMFR extend.
All methods use the same framework, data splits, and evaluation protocols.

\begin{table*}[!t]
  \centering
  \caption{Recommendation performance on Beauty, Cards, and ML in Recall@50 and NDCG@50, reported as means $\pm$ standard deviations in percent across five seeds.
  \textbf{Bold} and \underline{underlining} mark the highest and second-highest federated means.
  SERA improves all three backbones, narrows their gaps with the best centralized methods, and lifts some above centralized baselines.}
  \label{tab:main}
  \small
  \setlength{\tabcolsep}{4pt}
  \renewcommand{\arraystretch}{0.95}
  \begin{tabular}{c|l cc cc cc}
    \toprule
    \multicolumn{2}{c}{\textbf{Method}} & \multicolumn{2}{c}{\textbf{Beauty}} & \multicolumn{2}{c}{\textbf{Cards}} & \multicolumn{2}{c}{\textbf{ML}} \\
    \multicolumn{2}{c}{} & \textbf{Recall@50} & \textbf{NDCG@50} & \textbf{Recall@50} & \textbf{NDCG@50} & \textbf{Recall@50} & \textbf{NDCG@50} \\
    \midrule
    \multirow{6}{*}{\rotatebox{90}{Centralized}}
    & VBPR  & 23.79 $\pm$ 2.08 & 6.77 $\pm$ 0.61 & 74.96 $\pm$ 1.20 & 29.05 $\pm$ 0.63 & 21.44 $\pm$ 1.14 & 6.43 $\pm$ 0.31 \\
    & BM3   & 10.12 $\pm$ 2.66 & 2.58 $\pm$ 0.72 & 80.64 $\pm$ 1.23 & 34.38 $\pm$ 1.53 & 16.56 $\pm$ 3.01 & 4.69 $\pm$ 0.73 \\
    & MGCN  & 19.13 $\pm$ 1.82 & 5.51 $\pm$ 0.64 & 79.42 $\pm$ 0.81 & 33.84 $\pm$ 0.56 & 22.66 $\pm$ 0.57 & 6.54 $\pm$ 0.12 \\
    & FITMM & 19.68 $\pm$ 1.13 & 5.92 $\pm$ 0.54 & 81.91 $\pm$ 1.53 & 35.16 $\pm$ 0.19 & 22.23 $\pm$ 1.21 & 6.72 $\pm$ 0.18 \\
    & PGL   & 22.06 $\pm$ 1.35 & 5.74 $\pm$ 0.54 & 82.02 $\pm$ 2.04 & 34.52 $\pm$ 0.75 & 15.34 $\pm$ 6.84 & 4.81 $\pm$ 1.84 \\
    & MSCA  & 21.74 $\pm$ 1.96 & 6.47 $\pm$ 0.41 & 74.43 $\pm$ 1.37 & 31.41 $\pm$ 0.43 & 20.33 $\pm$ 1.73 & 5.94 $\pm$ 0.51 \\
    \midrule
    & FedMMR         & 13.83 $\pm$ 2.85 & 3.40 $\pm$ 0.73 & 74.59 $\pm$ 1.97 & 31.12 $\pm$ 0.93 & 8.56 $\pm$ 0.55 & 2.74 $\pm$ 0.16 \\
    & AMMFRS         & 16.28 $\pm$ 4.12 & 4.46 $\pm$ 1.42 & 69.18 $\pm$ 2.57 & 30.13 $\pm$ 0.49 & 9.54 $\pm$ 1.12 & 2.90 $\pm$ 0.44 \\
    & FedAvg         & 13.04 $\pm$ 0.56 & 3.77 $\pm$ 0.17 & 36.47 $\pm$ 0.13 & 10.89 $\pm$ 0.03 & 6.34 $\pm$ 5.36 & 1.99 $\pm$ 1.77 \\
    & \quad + FedVLR & 16.60 $\pm$ 0.94 & 3.99 $\pm$ 0.39 & 64.67 $\pm$ 1.21 & 24.54 $\pm$ 1.82 & 9.08 $\pm$ 0.17 & 2.42 $\pm$ 0.07 \\
    & \quad + GFMFR  & 13.12 $\pm$ 1.46 & 3.28 $\pm$ 0.54 & 69.39 $\pm$ 3.61 & 21.54 $\pm$ 5.23 & 9.97 $\pm$ 2.09 & 2.94 $\pm$ 0.55 \\
    \rowcolor{RoyalBlue!6}
    \cellcolor{white} & \quad + SERA   & 21.26 $\pm$ 3.63 & 5.94 $\pm$ 0.93 & \underline{74.85 $\pm$ 1.84} & \underline{31.86 $\pm$ 0.49} & \underline{11.54 $\pm$ 1.98} & \underline{3.65 $\pm$ 0.41} \\
    & PFedRec        & 15.97 $\pm$ 0.32 & 3.97 $\pm$ 0.07 & 32.63 $\pm$ 0.64 & 8.67 $\pm$ 0.81 & 1.64 $\pm$ 0.14 & 0.42 $\pm$ 0.02 \\
    & \quad + FedVLR & 17.94 $\pm$ 1.38 & 4.82 $\pm$ 0.44 & 71.17 $\pm$ 1.35 & 27.65 $\pm$ 0.60 & 7.15 $\pm$ 1.21 & 2.14 $\pm$ 0.25 \\
    & \quad + GFMFR  & 14.78 $\pm$ 6.13 & 3.83 $\pm$ 1.51 & 53.26 $\pm$ 2.93 & 20.16 $\pm$ 0.74 & 7.18 $\pm$ 1.31 & 2.01 $\pm$ 0.37 \\
    \rowcolor{RoyalBlue!6}
    \cellcolor{white} & \quad + SERA   & \underline{24.58 $\pm$ 2.08} & \underline{7.05 $\pm$ 0.70} & 70.34 $\pm$ 6.25 & 24.66 $\pm$ 6.68 & 10.49 $\pm$ 0.81 & 3.48 $\pm$ 0.70 \\
    & FedRAP         & 16.21 $\pm$ 1.02 & 4.12 $\pm$ 0.27 & 42.49 $\pm$ 0.11 & 11.09 $\pm$ 0.02 & 1.71 $\pm$ 0.22 & 0.69 $\pm$ 0.11 \\
    & \quad + FedVLR & \bestresult{34.55 $\pm$ 1.74} & \bestresult{11.98 $\pm$ 0.33} & 57.56 $\pm$ 0.56 & 22.65 $\pm$ 0.27 & 10.75 $\pm$ 0.64 & 3.60 $\pm$ 0.06 \\
    & \quad + GFMFR  & 14.94 $\pm$ 1.62 & 4.14 $\pm$ 0.51 & 47.85 $\pm$ 7.70 & 13.47 $\pm$ 3.60 & 9.28 $\pm$ 1.47 & 2.76 $\pm$ 0.44 \\
    \rowcolor{RoyalBlue!6}
    \cellcolor{white}\multirow{-14}{*}{\rotatebox{90}{Federated}}
    & \quad + SERA   & 20.47 $\pm$ 4.95 & 5.70 $\pm$ 1.53 & \bestresult{82.02 $\pm$ 0.44} & \bestresult{33.31 $\pm$ 1.15} & \bestresult{11.70 $\pm$ 0.59} & \bestresult{3.72 $\pm$ 0.13} \\
    \bottomrule
  \end{tabular}
\end{table*}

\begin{table*}[!t]
  \centering
  \caption{Recommendation performance on Dance, Food, and Movie in Recall@50 and NDCG@50, reported as means $\pm$ standard deviations in percent across five seeds.
  Formatting follows Table~\ref{tab:main}.
  AMMFRS exceeds the available memory on these datasets.
  SERA improves every federated backbone and outperforms FedVLR and GFMFR in both metrics across all three short-video datasets.}
  \label{tab:main-shortvideo}
  \small
  \setlength{\tabcolsep}{4pt}
  \renewcommand{\arraystretch}{0.95}
  \begin{tabular}{c|l cc cc cc}
    \toprule
    \multicolumn{2}{c}{\textbf{Method}} & \multicolumn{2}{c}{\textbf{Dance}} & \multicolumn{2}{c}{\textbf{Food}} & \multicolumn{2}{c}{\textbf{Movie}} \\
    \multicolumn{2}{c}{} & \textbf{Recall@50} & \textbf{NDCG@50} & \textbf{Recall@50} & \textbf{NDCG@50} & \textbf{Recall@50} & \textbf{NDCG@50} \\
    \midrule
    \multirow{6}{*}{\rotatebox{90}{Centralized}}
    & VBPR  & 23.07 $\pm$ 0.22 & 7.56 $\pm$ 0.08 & 21.12 $\pm$ 0.21 & 6.75 $\pm$ 0.12 & 12.64 $\pm$ 0.16 & 4.24 $\pm$ 0.07 \\
    & BM3   & 24.21 $\pm$ 0.40 & 8.13 $\pm$ 0.11 & 25.03 $\pm$ 0.22 & 8.28 $\pm$ 0.12 & 17.62 $\pm$ 0.23 & 6.00 $\pm$ 0.06 \\
    & MGCN  & 22.17 $\pm$ 2.48 & 7.60 $\pm$ 0.89 & 24.88 $\pm$ 0.29 & 8.39 $\pm$ 0.07 & 16.77 $\pm$ 0.14 & 5.83 $\pm$ 0.06 \\
    & FITMM & 20.74 $\pm$ 0.34 & 7.31 $\pm$ 0.07 & 22.62 $\pm$ 0.22 & 7.58 $\pm$ 0.03 & 15.33 $\pm$ 0.11 & 5.51 $\pm$ 0.05 \\
    & PGL   & 24.55 $\pm$ 0.44 & 8.45 $\pm$ 0.16 & 24.84 $\pm$ 0.27 & 8.28 $\pm$ 0.12 & 17.18 $\pm$ 0.18 & 6.19 $\pm$ 0.03 \\
    & MSCA  & 21.00 $\pm$ 0.41 & 7.23 $\pm$ 0.08 & 22.72 $\pm$ 3.22 & 7.76 $\pm$ 0.84 & 12.95 $\pm$ 0.07 & 4.74 $\pm$ 0.04 \\
    \midrule
    & FedMMR         & 13.45 $\pm$ 0.65 & 4.66 $\pm$ 0.31 & 11.18 $\pm$ 0.84 & 3.53 $\pm$ 0.37 & 7.35 $\pm$ 0.33 & 2.44 $\pm$ 0.11 \\
    & FedAvg         & 4.49 $\pm$ 1.03 & 1.20 $\pm$ 0.53 & 4.79 $\pm$ 1.59 & 1.34 $\pm$ 0.43 & 4.19 $\pm$ 1.86 & 1.09 $\pm$ 0.46 \\
    & \quad + FedVLR & 9.76 $\pm$ 0.49 & 2.57 $\pm$ 0.31 & 9.91 $\pm$ 0.05 & 2.70 $\pm$ 0.02 & 5.64 $\pm$ 0.00 & 1.61 $\pm$ 0.00 \\
    & \quad + GFMFR  & 8.61 $\pm$ 0.56 & 2.22 $\pm$ 0.10 & 8.81 $\pm$ 0.88 & 2.38 $\pm$ 0.21 & 5.58 $\pm$ 0.63 & 1.79 $\pm$ 0.21 \\
    \rowcolor{RoyalBlue!6}
    \cellcolor{white} & \quad + SERA   & 16.69 $\pm$ 2.44 & 4.48 $\pm$ 2.07 & \underline{14.30 $\pm$ 3.61} & \underline{4.35 $\pm$ 1.56} & 7.23 $\pm$ 0.29 & 2.11 $\pm$ 0.11\\
    & PFedRec        & 2.52 $\pm$ 0.63 & 0.68 $\pm$ 0.17 & 3.27 $\pm$ 0.78 & 0.93 $\pm$ 0.21 & 1.70 $\pm$ 0.35 & 0.42 $\pm$ 0.11 \\
    & \quad + FedVLR & 7.92 $\pm$ 0.99 & 1.92 $\pm$ 0.17 & 9.96 $\pm$ 0.89 & 2.49 $\pm$ 0.29 & 5.59 $\pm$ 0.81 & 1.56 $\pm$ 0.19 \\
    & \quad + GFMFR  & 5.94 $\pm$ 0.34 & 1.65 $\pm$ 0.09 & 6.39 $\pm$ 0.28 & 1.74 $\pm$ 0.11 & 3.96 $\pm$ 0.15 & 1.15 $\pm$ 0.05 \\
    \rowcolor{RoyalBlue!6}
    \cellcolor{white} & \quad + SERA   & \bestresult{22.57 $\pm$ 0.49} & \bestresult{7.68 $\pm$ 0.22} & \bestresult{19.07 $\pm$ 1.67} & \bestresult{6.19 $\pm$ 0.60} & \bestresult{12.11 $\pm$ 0.14} & \bestresult{4.30 $\pm$ 0.07} \\
    & FedRAP         & 2.37 $\pm$ 0.74 & 0.53 $\pm$ 0.23 & 5.48 $\pm$ 1.32 & 1.63 $\pm$ 0.45 & 1.68 $\pm$ 0.65 & 0.44 $\pm$ 0.13 \\
    & \quad + FedVLR & 7.27 $\pm$ 1.76 & 1.91 $\pm$ 0.54 & 8.58 $\pm$ 0.15 & 2.20 $\pm$ 0.03 & 6.02 $\pm$ 0.40 & 1.65 $\pm$ 0.09 \\
    & \quad + GFMFR  & 5.02 $\pm$ 1.16 & 1.40 $\pm$ 0.40 & 7.45 $\pm$ 1.04 & 2.70 $\pm$ 0.36 & 5.09 $\pm$ 1.71 & 1.45 $\pm$ 0.50 \\
    \rowcolor{RoyalBlue!6}
    \cellcolor{white}\multirow{-13}{*}{\rotatebox{90}{Federated}}
    & \quad + SERA   & \underline{17.33 $\pm$ 2.87} & \underline{5.23 $\pm$ 1.09} & 13.52 $\pm$ 1.63 & 3.98 $\pm$ 0.39 & \underline{9.88 $\pm$  0.86} & \underline{3.14 $\pm$ 0.27} \\
    \bottomrule
  \end{tabular}
\end{table*}

\subsection{Performance Comparison}

Tables~\ref{tab:main} and \ref{tab:main-shortvideo} compare centralized and federated recommendation methods across six datasets.
Comparing each backbone row with the indented SERA row below it shows that SERA improves every backbone on each dataset.
These gains narrow the gap with the best centralized methods, and SERA lets some backbones outperform centralized baselines.

Within the same backbone architecture, SERA exceeds GFMFR across all datasets and FedVLR in most comparisons.
SERA trails FedVLR only with FedRAP on Beauty and PFedRec on Cards.
The gains also vary across backbones and change which backbone leads on most datasets.
This backbone dependence is built into SERA: collaborative scores set the collaborative ranking, history weights, and user-adaptive residual scale.
Each backbone thus sets the semantic residual scores and the collaborative score gaps they must exceed.

SERA's gains extend from the product-review datasets to the larger item sets and sparser interaction matrices of the short-video datasets.
Both domains have sparse interaction histories on average, as shown in Table~\ref{tab:datasets}.
The shared projector learns from aggregated user updates, so a short local history can use content representations learned from other users' interactions.
SERA scores candidates with these representations, so its semantic residual scores also cover items with few interactions.

\begin{table*}[!t]
  \centering
  \caption{Parameter comparison of FedMRs with latent dimension 64 and feature dimension 512 per modality.
  SERA uses a single linear projection with bias.
  Backbone rows report total personalized parameters per client, and indented rows report additions.
  K denotes thousands.
  SERA derives personalized queries from local histories and shares all added trainable parameters.
  It adds zero personalized parameters and fewer shared parameters than FedVLR and GFMFR across all three backbones.}
  \label{tab:param-analysis}
  \begin{tabular}{l*{6}{c}|*{3}{c}}
    \toprule
    & \multicolumn{6}{c|}{\makecell{Per-client personalized parameters}}
    & \multicolumn{3}{c}{Global multimodal additions} \\
    Method
    & Beauty & Cards & ML & Dance & Food & Movie
    & \makecell{Client}
    & \makecell{Server}
    & \makecell{Representation\\per item} \\
    \midrule
    FedMMR
    & 64 & 64 & 64 & 64 & 64 & 64
    & 0 & 66K & 0 \\
    AMMFRS
    & 1027K & 1020K & 1261K & 1750K & 1443K & 2168K
    & 0 & 0 & 0 \\
    \midrule
    FedAvg
    & 193 & 193 & 193 & 193 & 193 & 193
    & 0 & 0 & 0 \\
    \quad + FedVLR
    & 13K & 13K & 13K & 13K & 13K & 13K
    & 124K & 0 & 64 \\
    \quad + GFMFR
    & 0 & 0 & 0 & 0 & 0 & 0
    & 0 & 134K & 0 \\
    \rowcolor{RoyalBlue!6}\quad + SERA
    & 0 & 0 & 0 & 0 & 0 & 0
    & 0 & 66K & 64 \\
    PFedRec
    & 65 & 65 & 65 & 65 & 65 & 65
    & 0 & 0 & 0 \\
    \quad + FedVLR
    & 13K & 13K & 13K & 13K & 13K & 13K
    & 124K & 0 & 64 \\
    \quad + GFMFR
    & 0 & 0 & 0 & 0 & 0 & 0
    & 0 & 134K & 0 \\
    \rowcolor{RoyalBlue!6}\quad + SERA
    & 0 & 0 & 0 & 0 & 0 & 0
    & 0 & 66K & 64 \\
    FedRAP
    & 23K & 8K & 234K & 107K & 72K & 162K
    & 0 & 0 & 0 \\
    \quad + FedVLR
    & 13K & 13K & 13K & 13K & 13K & 13K
    & 124K & 0 & 64 \\
    \quad + GFMFR
    & 0 & 0 & 0 & 0 & 0 & 0
    & 8K & 134K & 0 \\
    \rowcolor{RoyalBlue!6}\quad + SERA
    & 0 & 0 & 0 & 0 & 0 & 0
    & 0 & 66K & 64 \\
    \bottomrule
  \end{tabular}
\end{table*}

\subsection{Parameter Analysis}

Table~\ref{tab:param-analysis} shows that SERA adds no personalized or client-side trainable parameters on any backbone or dataset.
When validation selects a multilayer projector, SERA's server-side projector has up to 410K parameters.
Linear SERA with PFedRec has personalized and shared parameters comparable to FedMMR, while every SERA variant has fewer personalized parameters than AMMFRS.
Each query in Eqs.~\eqref{eq:weights} and \eqref{eq:query} comes from local histories and collaborative scores, so shared parameters yield personalized semantic residual scores.

\subsection{Content and Personalized Ranking}
\label{sec:content-ranking}

\subsubsection{Content and Modality Controls}
Fig.~\ref{fig:content-ablation} compares each full model with SERA retrained on altered content.
Shuffling permutes features across items, and the random-feature control replaces them with random vectors of the same mean and variance.
Both controls usually reduce mean accuracy.
Image removal causes more consistent losses than text removal.
On Cards, image removal lowers mean accuracy with every backbone, and text removal raises it for FedAvg and PFedRec.

The shuffled and random-feature controls detach each item from its own features, so the query in Eq.~\eqref{eq:query} and the similarity in Eq.~\eqref{eq:compat} stop reflecting the content of the user's history.
The modalities also contribute unequally to $\mathbf{z}_i$ in Eq.~\eqref{eq:basis}.
On Cards and Food, the image features vary more across items than the text features, and many Cards items share nearly identical title features.
Removing images discards the more variable modality, while removing text retains it and produces smaller gains or losses in ranking accuracy across backbones.

\begin{figure*}[!t]
  \centering
  \includegraphics[width=\textwidth]{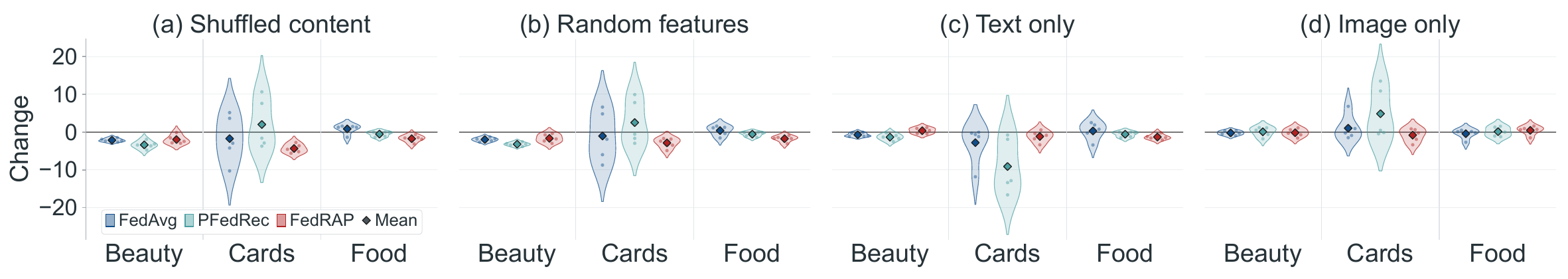}
  \caption{Content and modality ablations of SERA as NDCG@50 changes in percent from the full model in each setting.
  Each control is trained separately.
  (a) Shuffled content and (b) random features usually reduce mean accuracy.
  (c) Removing images causes more consistent losses than (d) removing text.}
  \label{fig:content-ablation}
\end{figure*}

\subsubsection{Residual and Content Contributions}
Fig.~\ref{fig:content-gains} compares SERA with separately trained zero-residual and shuffled-content controls.
The zero-residual control uses collaborative scores alone.
Shuffling keeps the residual computation unchanged, so this control tests whether each item needs its own features.
SERA achieves positive mean gains with FedRAP over both controls on every displayed dataset, while some gains with FedAvg and PFedRec are near zero or negative.

\begin{figure}[!t]
  \centering
  \includegraphics[width=\columnwidth]{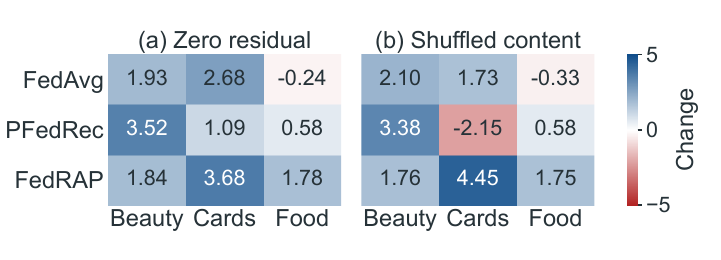}
  \caption{Contributions of SERA's semantic residual scores and correct item features to NDCG@50 in percent.
  Gains are mean differences across five seeds between full SERA and separately trained controls using (a) collaborative scores alone or (b) shuffled item content.
  Mean gains are positive for all displayed FedRAP configurations and vary in sign for other backbones.}
  \label{fig:content-gains}
\end{figure}

\subsubsection{User-Level Ranking Changes}
Fig.~\ref{fig:user-effects} holds each trained model fixed and compares users' NDCG@50 with and without its semantic residual scores.
More users improve than decline.
An item overtakes a competitor when their residual difference exceeds their collaborative score gap in Eqs.~\eqref{eq:score} and \eqref{eq:pairgap}.
Pairwise ranking changes can improve accuracy by promoting the test item within the top 50 or into it.
By Corollary~\ref{cor:rank-range}, NDCG@50 stays fixed for a user when the test item has only one possible final rank or can only rank below the top 50.

\begin{figure}[!t]
  \centering
  \includegraphics[width=\columnwidth]{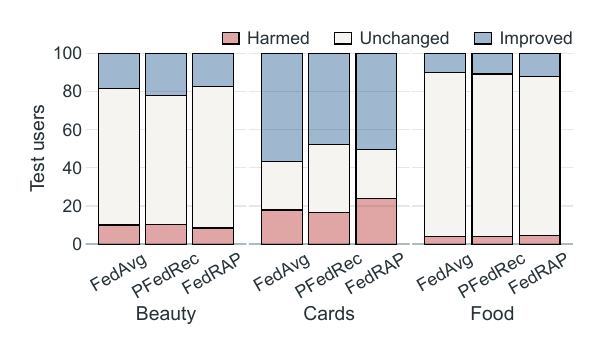}
  \caption{Percentages of Improved, Unchanged, and Harmed users, whose NDCG@50 rises, stays, or falls when SERA adds its semantic residual scores.
  More users benefit than decline in every configuration.
  The largest group is \textit{Unchanged} on Beauty and Food and \textit{Improved} on Cards.
  Ranks outside the top 50 have zero NDCG, so users whose test items stay there are \textit{Unchanged}.}
  \label{fig:user-effects}
\end{figure}

\subsection{Shared Projector Learning}
\label{sec:projector-learning}

\subsubsection{Bias and Weight Learning}
Fig.~\ref{fig:projector-learning} measures the gain of bias learning over a fixed projector and the added gain of weight learning.
Bias learning usually improves accuracy.
Weight learning often lowers it, except for FedRAP on Cards.

The frozen features lie in a narrow cone~\cite{liang2022mind}, and a fixed projector in Eq.~\eqref{eq:basis} keeps all items close to their mean direction.
Every item's residual gradient updates the same bias vector, so the shared biases can remove this common direction and spread the representations apart.
The weights far outnumber the biases and must fit the sparse feedback in Table~\ref{tab:datasets}, so their updates rarely improve on the bias shift.
Bias-only tuning likewise rivals full tuning of pretrained transformers on small data~\cite{benzaken2022bitfit}.
Eqs.~\eqref{eq:query} and \eqref{eq:compat} use these representations, so bias updates change all users' semantic residual scores at once.

\begin{figure}[!t]
  \centering
  \includegraphics[width=\columnwidth]{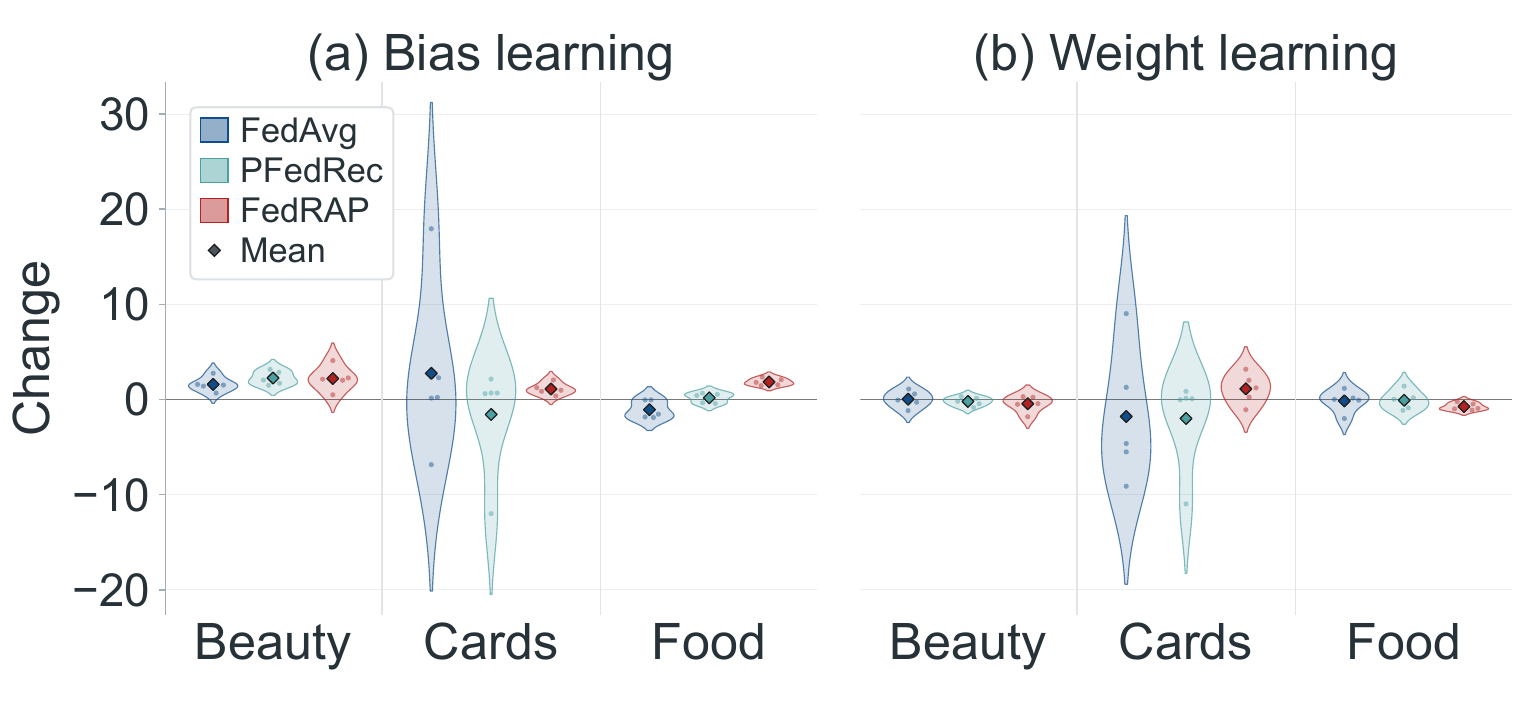}
  \caption{Contributions of SERA's shared biases and projection weights to NDCG@50 in percent.
  Differences are (a) bias learning minus a fixed projector and (b) full training minus bias learning alone. Bias learning usually improves accuracy, while additional weight learning often lowers it.}
  \label{fig:projector-learning}
\end{figure}

\subsubsection{Feature Centering}
Fig.~\ref{fig:bias-centering} tests whether subtracting the mean item feature reduces SERA's gain from bias learning.
Centering removes most of this gain on Food, where SERA uses a linear projector.
On Beauty and Cards, most of the gain remains in SERA's multilayer projectors, whose hidden-layer biases continue to adjust the content representations after input centering, altering both queries and candidate similarities.

\begin{figure}[!t]
  \centering
  \includegraphics[width=\columnwidth]{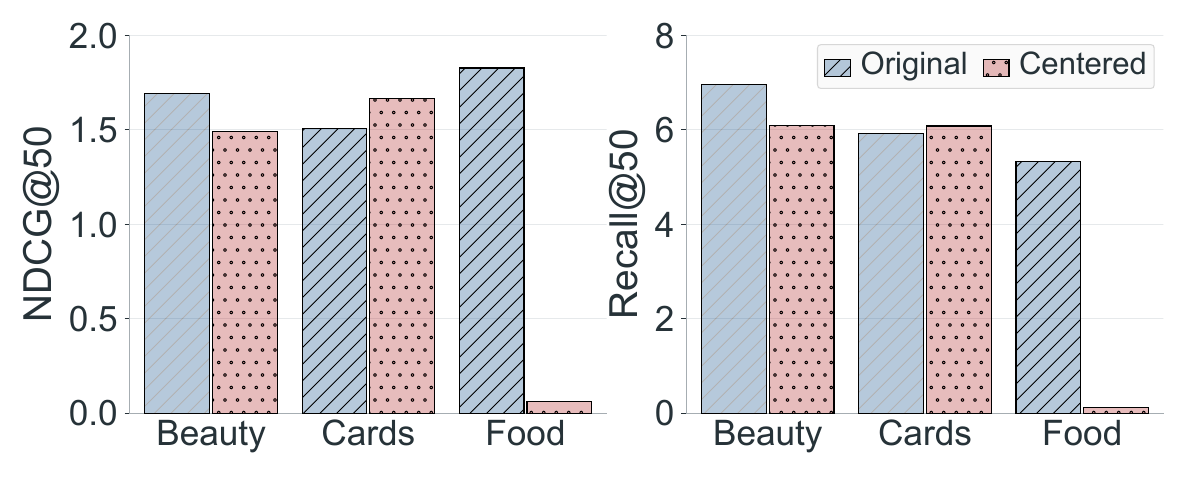}
  \caption{Effect of feature centering on SERA's bias gains in NDCG@50 and Recall@50 in percent. 
  Gains are learned-bias minus fixed-bias scores for FedRAP with fixed projection weights over 15 paired seeds. 
  Centering subtracts the feature mean over the item set. Gains persist on Beauty and Cards with multilayer projectors and shrink on Food with a linear one.}
  \label{fig:bias-centering}
\end{figure}

\subsection{Collaborative Training and Content Learning}
\label{sec:collaborative-content}

SERA's projector learns from scores whose collaborative part changes in training.
Fig.~\ref{fig:training-effects} measures the accuracy difference between intact and shuffled content with collaborative parameters fixed at initialization or updated during training.
Continued training usually reduces these gains and makes their means negative in some FedAvg and PFedRec configurations, which matches the weaker gains in Fig.~\ref{fig:content-gains}.
SERA's mean gains stay positive with FedRAP on all displayed datasets.
Collaborative training affects content learning through the pairwise loss in Eq.~\eqref{eq:clientrisk}, which becomes less sensitive to score changes as the margin between observed items and sampled negatives grows.
Collaborative updates therefore change which training pairs still give the projector large gradients.
Through Eqs.~\eqref{eq:weights}, \eqref{eq:residual}, and \eqref{eq:pairgap}, they also reweight history items, rescale semantic residual scores, and change the gaps that these scores must overcome.
The comparison measures their combined effect on content gains.
The stop-gradient in Eq.~\eqref{eq:score} only blocks residual gradients from flowing into the backbone.

\begin{figure}[!t]
  \centering
  \includegraphics[width=\columnwidth]{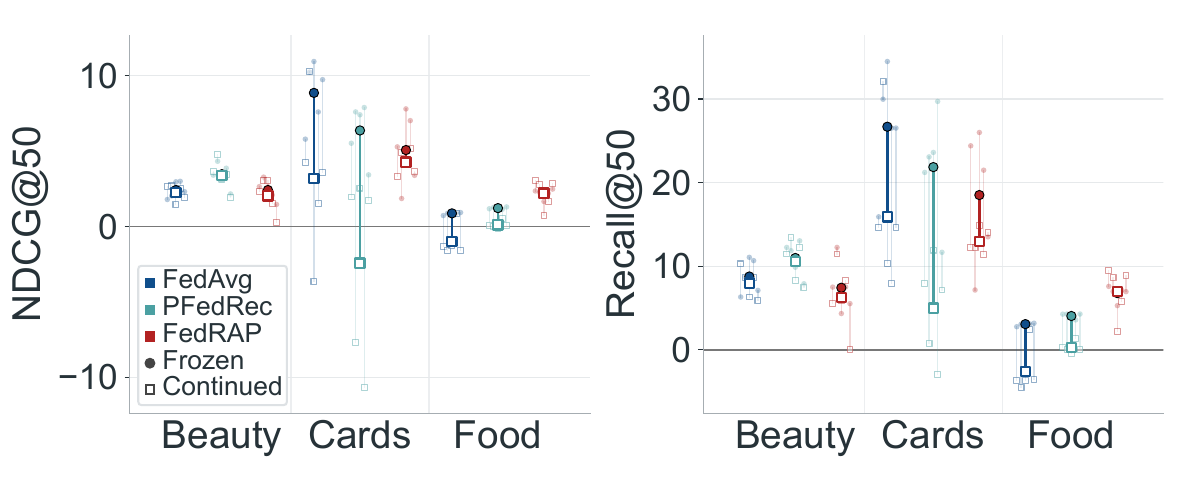}
  \caption{Effect of collaborative training on SERA's content gains in NDCG@50 and Recall@50 in percent.
  Gains are intact minus shuffled content, and the content projector learns in both conditions.
  Fixing collaborative parameters at initialization usually yields larger gains than continued training does.}
  \label{fig:training-effects}
\end{figure}

\subsection{Residual Budget and Residual Function}
\label{sec:residual-control}

\subsubsection{Residual Budget}
Figs.~\ref{fig:bound-ndcg} and \ref{fig:budget-users} vary the budget with the trained model fixed.
Their axes use the ratio $\kappa/\kappa_{\mathrm{train}}$, where $\kappa_{\mathrm{train}}$ is the training budget.
Gains initially rise, then plateau or decline in several configurations as the harmed share increases.
By Proposition~\ref{prop:rank-path}, as the budget grows, the test item passes some competitors and is passed by others, so its rank can rise and then fall. The final ranking approaches the ranking by similarity.
Pairwise ranking changes that push the test item down within the top 50 or out of it lower its NDCG.
At large budgets, FedRAP on Cards improves more users than it harms but lowers mean NDCG.
NDCG weights higher ranks more, so a few large losses can outweigh many smaller gains.

\begin{figure}[!t]
  \centering
  \includegraphics[width=\columnwidth]{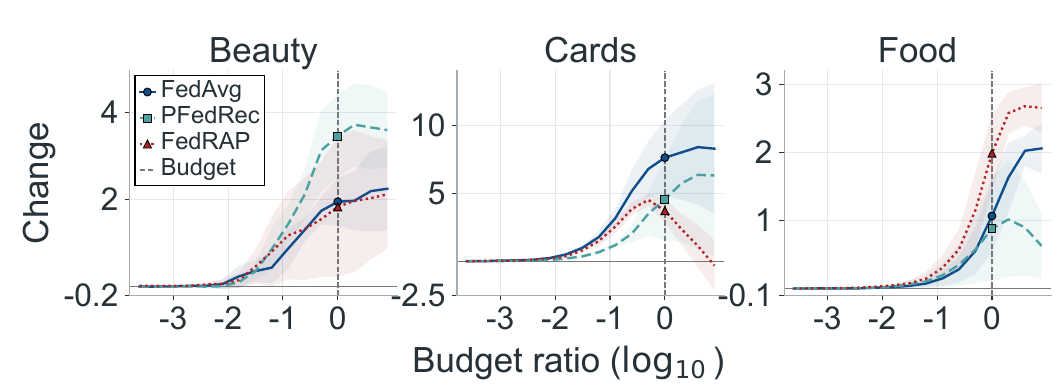}
  \caption{Effect of SERA's residual budget on NDCG@50 changes in percent.
  Only the budget varies after training, and changes are relative to zero residual. The horizontal axis is $\log_{10}(\kappa/\kappa_{\mathrm{train}})$, so zero retains the trained budget.
  Shading shows one standard deviation across five seeds. Accuracy initially improves and later plateaus or declines in several displayed configurations.}
  \label{fig:bound-ndcg}
\end{figure}

\begin{figure}[!t]
  \centering
  \includegraphics[width=\columnwidth]{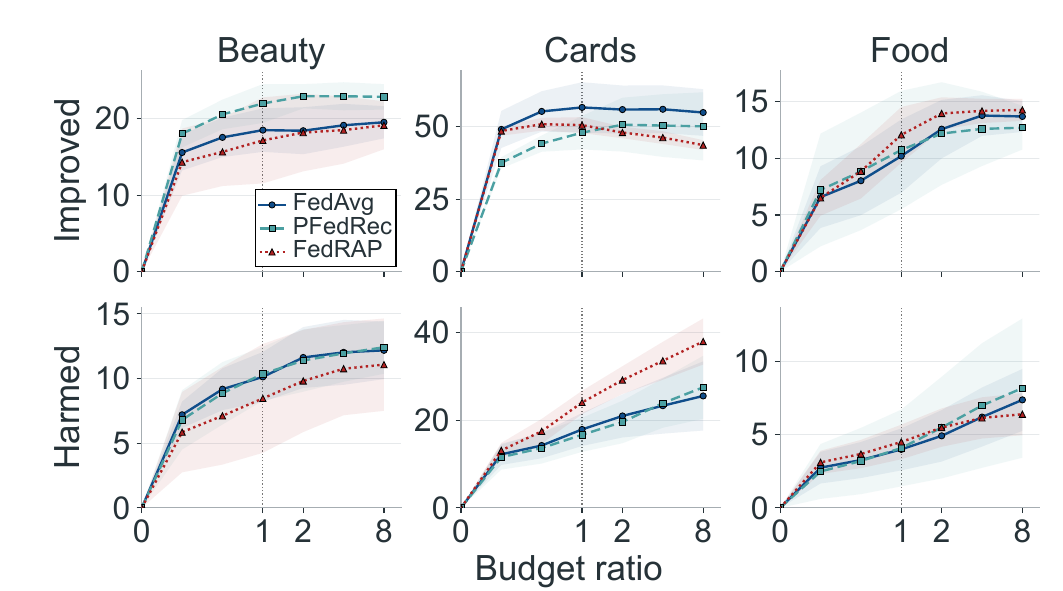}
  \caption{User proportions in percent for NDCG@50 gains and losses across SERA's residual budgets. The horizontal axis is the budget ratio $\kappa/\kappa_{\mathrm{train}}$, and changes are relative to zero residual. Shading shows one standard deviation across five seeds. The harmed share rises as the improved share often plateaus.}
  \label{fig:budget-users}
\end{figure}

\subsubsection{Residual Function}
Within the residual bound, the residual function determines how similarity differences become score differences.
Fig.~\ref{fig:shape-comparison} compares separately trained tanh residuals with linear residuals $\kappa\sigma_u\tanh(1/\tau_r)\,c_{ui}$, which share the residual bound in Eq.~\eqref{eq:bound}.
Mean accuracy is similar on Beauty, while the linear residual is better with all three backbones on Food.
Separate training also gives each function its own representations and training trajectory.
Both residual functions preserve the order by similarity for fixed representations.
At the same residual bound, the linear residual scales all similarity differences equally, and the slope of Eq.~\eqref{eq:shape} varies with similarity.
The two functions thus give different residual differences and change the ranking of different pairs.

\begin{figure}[!t]
  \centering
  \includegraphics[width=\columnwidth]{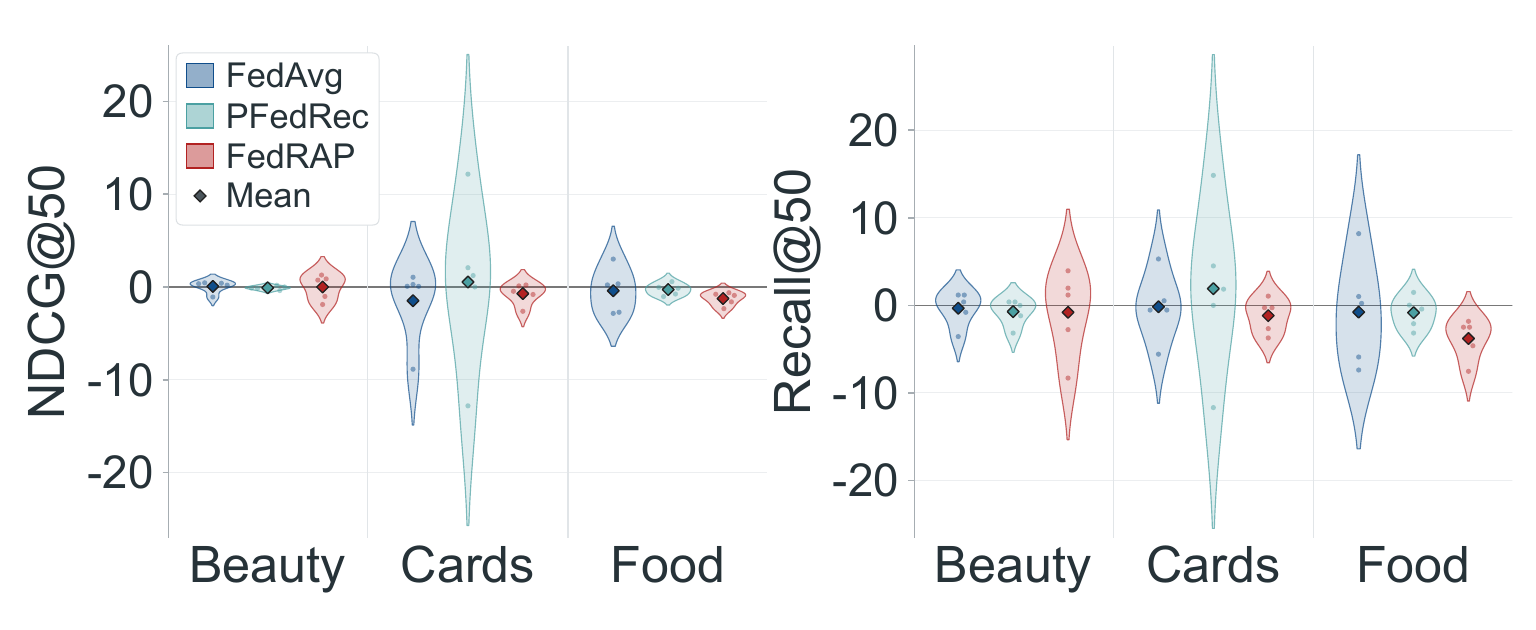}
  \caption{SERA's nonlinear versus linear residuals in NDCG@50 and Recall@50 in percent. Differences are tanh minus linear residuals, each trained separately with the same ratio of residual bound to user-adaptive residual scale. Linear residuals usually have higher mean accuracy, notably on Food.}
  \label{fig:shape-comparison}
\end{figure}

\subsubsection{Residual Temperature}
The temperature $\tau_r$ in Eq.~\eqref{eq:shape} controls the saturation of SERA's nonlinear residual.
Fig.~\ref{fig:temperature-response} compares models trained at different temperatures against a linear residual at matched residual bounds in Eq.~\eqref{eq:bound}.
At the lowest temperature, SERA falls below the linear residual in every displayed configuration.
Higher temperatures generally reduce these mean deficits.
Lower temperatures drive the shape term in Eq.~\eqref{eq:shape} toward $\pm1$, so highly similar items receive nearly equal semantic residual scores.
Higher temperatures reduce saturation and bring the tanh residual closer to the linear residual at the matched residual bound in Eq.~\eqref{eq:bound}.

\begin{figure}[!t]
  \centering
  \includegraphics[width=\columnwidth]{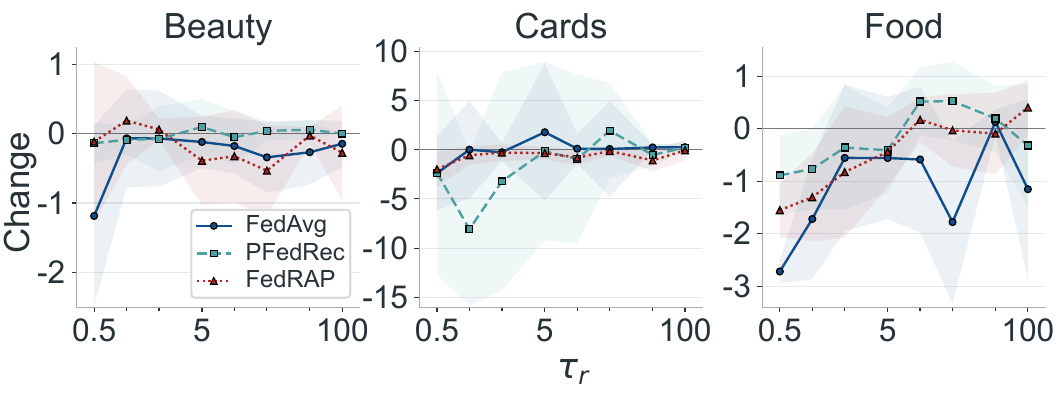}
  \caption{SERA's residual temperature and NDCG@50 differences from a linear residual in percent. All conditions are trained separately with the same ratio $R_u/\sigma_u$. Shading shows one standard deviation across five seeds. Mean accuracy differences generally shrink toward zero at higher temperatures.}
  \label{fig:temperature-response}
\end{figure}

\subsection{Privacy-Utility Trade-Off}
\label{sec:privacy-noise}
In Fig.~\ref{fig:privacy-noise}, each client clips its accumulated gradient $\mathbf{G}_u^t$ of Algorithm~\ref{alg:fedsera} and adds independent Gaussian noise to all $Md$ entries before weighted aggregation.
The horizontal axis gives the target signal-to-noise ratio (SNR) used to set a fixed noise multiplier before training.
Accuracy falls as the target SNR decreases and then levels off over the evaluated noise levels.
The accounting value $\varepsilon$ reaches 10 at a higher target SNR on Cards than on Beauty and at backbone-specific targets on Food. By then, accuracy has mostly leveled off.

Noise in the uploaded gradients disrupts how the shared projector learns semantic residual scores.
Through the query construction in Eq.~\eqref{eq:query} and the candidate comparison in Eq.~\eqref{eq:compat}, this disturbance changes which items receive larger residual scores.
The stop-gradient in Eq.~\eqref{eq:score} keeps the noisy residual updates out of collaborative training, while local histories and collaborative scores continue to personalize each query.
Eq.~\eqref{eq:bound} caps each semantic residual score even when the projector is poorly learned.
Corollary~\ref{cor:certificate} shows that pairs retain their collaborative order when the score gap exceeds twice this cap.
The noise level affects how accurately SERA reorders the remaining competitors, while the residual budget limits the collaborative score gaps these scores can overcome.

\begin{figure}[!t]
  \centering
  \includegraphics[width=\columnwidth]{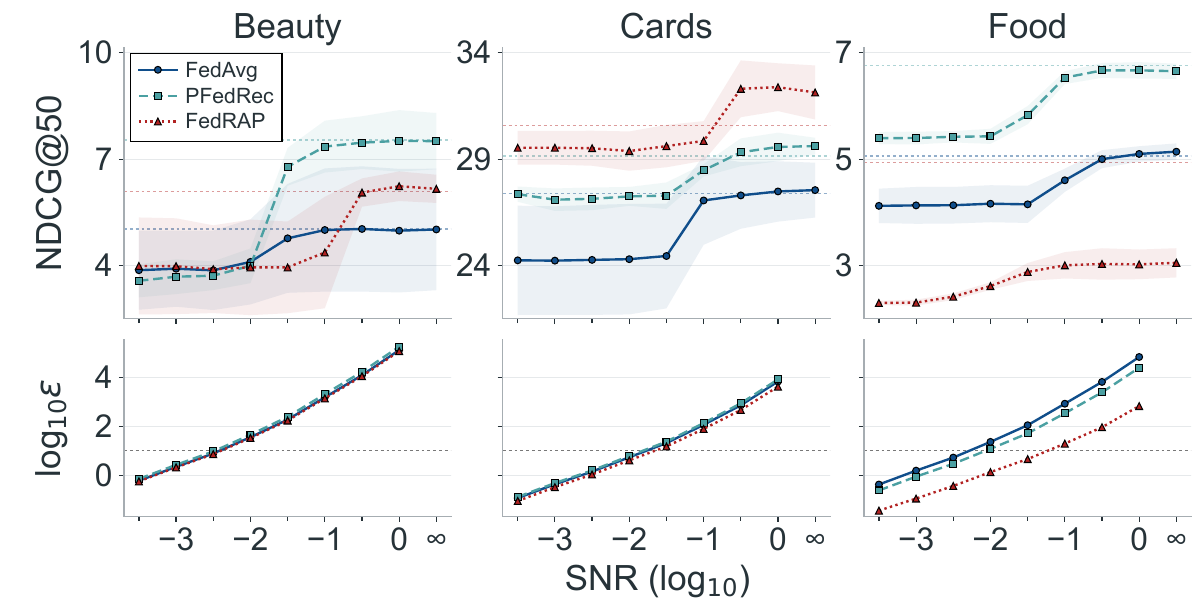}
  \caption{SERA's privacy-utility trade-off with clipped, noised uploads. The horizontal axis gives target SNR.
  Top: test NDCG@50 in percent over the last four rounds with one standard deviation across seeds. Dashed lines omit clipping and noise, and SNR $\infty$ is clipping alone.
  Bottom: $\log_{10}\varepsilon$ at $\delta=10^{-5}$ and a dotted line at $\varepsilon=10$. These conditional RDP accounting values do not establish differential privacy for the full training procedure.}
  \label{fig:privacy-noise}
\end{figure}

\section{Conclusion}
\label{sec:conclusion}
SERA fuses item content with a FedRS at the score level, adding semantic residual scores that are learned from local histories and capped at a user-adaptive residual bound.
Each rank change is then a pairwise ranking change in which a residual difference overcomes a collaborative score gap, and each pair changes ranking at most once as the budget grows.
For any content, an item outside the collaborative top-$K$ list can enter the top $K$ only if its collaborative boundary gap is at most twice the residual bound.
Larger budgets let content cross wider gaps but can lower accuracy by moving competitors above the test items.
SERA improves all three FedRSs on all six datasets without adding personalized parameters.
Federated semantic residuals thus bring content-based personalization to existing FedRSs with explicit control of content reach.

\bibliographystyle{IEEEtran}
\bibliography{references}

\begin{thebibliography}{10}
\providecommand{\url}[1]{#1}
\csname url@samestyle\endcsname
\providecommand{\newblock}{\relax}
\providecommand{\bibinfo}[2]{#2}
\providecommand{\BIBentrySTDinterwordspacing}{\spaceskip=0pt\relax}
\providecommand{\BIBentryALTinterwordstretchfactor}{4}
\providecommand{\BIBentryALTinterwordspacing}{\spaceskip=\fontdimen2\font plus
\BIBentryALTinterwordstretchfactor\fontdimen3\font minus
  \fontdimen4\font\relax}
\providecommand{\BIBforeignlanguage}[2]{{%
\expandafter\ifx\csname l@#1\endcsname\relax
\typeout{** WARNING: IEEEtran.bst: No hyphenation pattern has been}%
\typeout{** loaded for the language `#1'. Using the pattern for}%
\typeout{** the default language instead.}%
\else
\language=\csname l@#1\endcsname
\fi
#2}}
\providecommand{\BIBdecl}{\relax}
\BIBdecl

\bibitem{zhang2025personalized}
C.~Zhang, G.~Long, Z.~Zhang, Z.~Li, H.~Zhang, Q.~Yang, and B.~Yang,
  ``Personalized recommendation models in federated settings: A survey,''
  \emph{IEEE Trans. Knowl. Data Eng.}, vol.~37, no.~11, pp. 6562--6581, 2025.

\bibitem{ammad2019federated}
M.~{Ammad-ud-din}, E.~Ivannikova, S.~A. Khan, W.~Oyomno, Q.~Fu, K.~E. Tan, and
  A.~Flanagan, ``Federated collaborative filtering for privacy-preserving
  personalized recommendation system,'' \emph{arXiv preprint arXiv:1901.09888},
  2019.

\bibitem{koren2009matrix}
Y.~Koren, R.~Bell, and C.~Volinsky, ``Matrix factorization techniques for
  recommender systems,'' \emph{Computer}, vol.~42, no.~8, pp. 30--37, 2009.

\bibitem{li2024fedrap}
\BIBentryALTinterwordspacing
Z.~Li, G.~Long, and T.~Zhou, ``Federated recommendation with additive
  personalization,'' in \emph{Proc. Int. Conf. Learn. Represent.}, 2024.
  [Online]. Available: \url{https://openreview.net/forum?id=xkXdE81mOK}
\BIBentrySTDinterwordspacing

\bibitem{li2024fedcore}
Z.~Li, X.~Wu, W.~Pan, Y.~Ding, Z.~Wu, S.~Tan, Q.~Xu, Q.~Yang, and Z.~Ming,
  ``{FedCORE}: Federated learning for cross-organization recommendation
  ecosystem,'' \emph{IEEE Trans. Knowl. Data Eng.}, vol.~36, no.~8, pp.
  3817--3831, 2024.

\bibitem{zhang2024ifedrec}
C.~Zhang, G.~Long, T.~Zhou, Z.~Zhang, P.~Yan, and B.~Yang, ``When federated
  recommendation meets cold-start problem: Separating item attributes and user
  interactions,'' in \emph{Proc. ACM Web Conf.}, 2024, pp. 3632--3642.

\bibitem{du2018personalized}
X.~Du, H.~Yin, L.~Chen, Y.~Wang, Y.~Yang, and X.~Zhou, ``Personalized video
  recommendation using rich contents from videos,'' \emph{IEEE Trans. Knowl.
  Data Eng.}, vol.~32, no.~3, pp. 492--505, 2020.

\bibitem{zhang2023beyond}
X.~Zhang, B.~Xu, F.~Ma, C.~Li, L.~Yang, and H.~Lin, ``Beyond co-occurrence:
  Multi-modal session-based recommendation,'' \emph{IEEE Trans. Knowl. Data
  Eng.}, vol.~36, no.~4, pp. 1450--1462, 2024.

\bibitem{xu2021multi}
Y.~Xu, L.~Zhu, Z.~Cheng, J.~Li, Z.~Zhang, and H.~Zhang, ``Multi-modal discrete
  collaborative filtering for efficient cold-start recommendation,'' \emph{IEEE
  Trans. Knowl. Data Eng.}, vol.~35, no.~1, pp. 741--755, 2023.

\bibitem{zhang2021lattice}
J.~Zhang, Y.~Zhu, Q.~Liu, S.~Wu, S.~Wang, and L.~Wang, ``Mining latent
  structures for multimedia recommendation,'' in \emph{Proc. ACM Int. Conf.
  Multimedia}, 2021, pp. 3872--3880.

\bibitem{feng2024robust}
C.~Feng, D.~Feng, G.~Huang, Z.~Liu, Z.~Wang, and X.-G. Xia, ``Robust
  privacy-preserving recommendation systems driven by multimodal federated
  learning,'' \emph{IEEE Trans. Neural Netw. Learn. Syst.}, vol.~36, no.~5, pp.
  8896--8910, 2025.

\bibitem{zhou2026comprehensive}
H.~Zhou, X.~Zhou, Z.~Zeng, L.~Zhang, and C.~Miao, ``A comprehensive survey on
  multimodal recommender systems: Taxonomy, evaluation, and future
  directions,'' \emph{IEEE Trans. Pattern Anal. Mach. Intell.}, 2026.

\bibitem{li2024fedvlr}
Z.~Li, G.~Long, J.~Jiang, C.~Zhang, and Q.~Yang, ``Federated
  vision-language-recommendation with personalized fusion,'' in \emph{Proc.
  AAAI Conf. Artif. Intell.}, vol.~40, no.~28, 2026, pp. 23\,337--23\,345.

\bibitem{zhao2024fusion}
F.~Zhao, C.~Zhang, and B.~Geng, ``Deep multimodal data fusion,'' \emph{ACM
  Comput. Surv.}, vol.~56, no.~9, 2024, {Art. no.} 216.

\bibitem{li2024towards}
G.~Li, X.~Ding, L.~Yuan, L.~Zhang, and Q.~Rong, ``Towards resource-efficient
  and secure federated multimedia recommendation,'' in \emph{Proc. IEEE Int.
  Conf. Acoust., Speech Signal Process.}, 2024, pp. 5515--5519.

\bibitem{zhang2026multimodal}
C.~Zhang, W.~Zhang, G.~Long, Z.~Xue, R.~Xia, and B.~Yang, ``Multimodal-enhanced
  federated recommendation: A group-wise fusion approach,'' in \emph{Proc. ACM
  Web Conf.}, 2026, pp. 6597--6608.

\bibitem{mcmahan2017communication}
H.~B. McMahan, E.~Moore, D.~Ramage, S.~Hampson, and B.~Ag{\"u}era~y Arcas,
  ``Communication-efficient learning of deep networks from decentralized
  data,'' in \emph{Proc. Int. Conf. Artif. Intell. Statist.}, ser. PMLR,
  vol.~54, 2017, pp. 1273--1282.

\bibitem{pillutla2022partial}
K.~Pillutla, K.~Malik, A.-R. Mohamed, M.~Rabbat, M.~Sanjabi, and L.~Xiao,
  ``Federated learning with partial model personalization,'' in \emph{Proc.
  Int. Conf. Mach. Learn.}, ser. PMLR, vol. 162, 2022, pp. 17\,716--17\,758.

\bibitem{perifanis2022federated}
V.~Perifanis and P.~S. Efraimidis, ``Federated neural collaborative
  filtering,'' \emph{Knowl.-Based Syst.}, vol. 242, 2022, {Art. no.} 108441.

\bibitem{kairouz2021advances}
P.~Kairouz \emph{et~al.}, ``Advances and open problems in federated learning,''
  \emph{Found. Trends Mach. Learn.}, vol.~14, no. 1--2, pp. 1--210, 2021.

\bibitem{collins2021fedrep}
L.~Collins, H.~Hassani, A.~Mokhtari, and S.~Shakkottai, ``Exploiting shared
  representations for personalized federated learning,'' in \emph{Proc. Int.
  Conf. Mach. Learn.}, ser. PMLR, vol. 139, 2021, pp. 2089--2099.

\bibitem{zhang2023dual}
C.~Zhang, G.~Long, T.~Zhou, P.~Yan, Z.~Zhang, C.~Zhang, and B.~Yang, ``Dual
  personalization on federated recommendation,'' in \emph{Proc. Int. Joint
  Conf. Artif. Intell.}, 2023, pp. 4558--4566.

\bibitem{zhang2024gpfedrec}
C.~Zhang, G.~Long, T.~Zhou, Z.~Zhang, P.~Yan, and B.~Yang, ``{GPFedRec}:
  Graph-guided personalization for federated recommendation,'' in \emph{Proc.
  ACM SIGKDD Int. Conf. Knowl. Discovery Data Mining}, 2024, pp. 4131--4142.

\bibitem{yuan2024hetefedrec}
W.~Yuan, L.~Qu, L.~Cui, Y.~Tong, X.~Zhou, and H.~Yin, ``{HeteFedRec}: Federated
  recommender systems with model heterogeneity,'' in \emph{Proc. IEEE Int.
  Conf. Data Eng.}, 2024, pp. 1324--1337.

\bibitem{wei2019mmgcn}
Y.~Wei, X.~Wang, L.~Nie, X.~He, R.~Hong, and T.-S. Chua, ``{MMGCN}: Multi-modal
  graph convolution network for personalized recommendation of micro-video,''
  in \emph{Proc. ACM Int. Conf. Multimedia}, 2019, pp. 1437--1445.

\bibitem{he2016vbpr}
R.~He and J.~McAuley, ``{VBPR}: Visual {Bayesian} personalized ranking from
  implicit feedback,'' in \emph{Proc. AAAI Conf. Artif. Intell.}, vol.~30,
  no.~1, 2016, pp. 144--150.

\bibitem{chen2017attentive}
J.~Chen, H.~Zhang, X.~He, L.~Nie, W.~Liu, and T.-S. Chua, ``Attentive
  collaborative filtering: Multimedia recommendation with item- and
  component-level attention,'' in \emph{Proc. Int. ACM SIGIR Conf. Res.
  Develop. Inf. Retrieval}, 2017, pp. 335--344.

\bibitem{he2018nais}
X.~He, Z.~He, J.~Song, Z.~Liu, Y.-G. Jiang, and T.-S. Chua, ``{NAIS}: Neural
  attentive item similarity model for recommendation,'' \emph{IEEE Trans.
  Knowl. Data Eng.}, vol.~30, no.~12, pp. 2354--2366, 2018.

\bibitem{wang2025mmcn}
Z.~Wang, Y.~Feng, X.~Zhang, R.~Yang, and B.~Du, ``Multi-modal correction
  network for recommendation,'' \emph{IEEE Trans. Knowl. Data Eng.}, vol.~37,
  no.~2, pp. 810--822, 2025.

\bibitem{zhang2023micro}
J.~Zhang, Y.~Zhu, Q.~Liu, M.~Zhang, S.~Wu, and L.~Wang, ``Latent structure
  mining with contrastive modality fusion for multimedia recommendation,''
  \emph{IEEE Trans. Knowl. Data Eng.}, vol.~35, no.~9, pp. 9154--9167, 2023.

\bibitem{zhou2023freedom}
X.~Zhou and Z.~Shen, ``A tale of two graphs: Freezing and denoising graph
  structures for multimodal recommendation,'' in \emph{Proc. ACM Int. Conf.
  Multimedia}, 2023, pp. 935--943.

\bibitem{zhou2023bootstrap}
X.~Zhou, H.~Zhou, Y.~Liu, Z.~Zeng, C.~Miao, P.~Wang, Y.~You, and F.~Jiang,
  ``Bootstrap latent representations for multi-modal recommendation,'' in
  \emph{Proc. ACM Web Conf.}, 2023, pp. 845--854.

\bibitem{wei2023mmssl}
W.~Wei, C.~Huang, L.~Xia, and C.~Zhang, ``Multi-modal self-supervised learning
  for recommendation,'' in \emph{Proc. ACM Web Conf.}, 2023, pp. 790--800.

\bibitem{volkovs2017dropoutnet}
M.~Volkovs, G.~Yu, and T.~Poutanen, ``{DropoutNet}: Addressing cold start in
  recommender systems,'' in \emph{Proc. Adv. Neural Inf. Process. Syst.},
  vol.~30, 2017, pp. 4957--4966.

\bibitem{wei2021clcrec}
Y.~Wei, X.~Wang, Q.~Li, L.~Nie, Y.~Li, X.~Li, and T.-S. Chua, ``Contrastive
  learning for cold-start recommendation,'' in \emph{Proc. ACM Int. Conf.
  Multimedia}, 2021, pp. 5382--5390.

\bibitem{zhu2025pckd}
Z.~Zhu and W.~Zhang, ``Preference-consistent knowledge distillation for
  recommender system,'' \emph{IEEE Trans. Knowl. Data Eng.}, vol.~37, no.~4,
  pp. 2071--2084, 2025.

\bibitem{li2020stabilizing}
R.~Li, Z.~Qin, X.~Wang, S.~J. Chen, and D.~Metzler, ``Stabilizing neural search
  ranking models,'' in \emph{Proc. ACM Web Conf.}, 2020, pp. 2725--2732.

\bibitem{wu2022certified}
C.~Wu, R.~Zhang, J.~Guo, W.~Chen, Y.~Fan, M.~de~Rijke, and X.~Cheng,
  ``Certified robustness to word substitution ranking attack for neural ranking
  models,'' in \emph{Proc. ACM Int. Conf. Inf. Knowl. Manage.}, 2022, pp.
  2128--2137.

\bibitem{dean2020recommendations}
S.~Dean, S.~Rich, and B.~Recht, ``Recommendations and user agency: The
  reachability of collaboratively-filtered information,'' in \emph{Proc. Conf.
  Fairness, Accountability, Transparency}, 2020, pp. 436--445.

\bibitem{curmei2021quantifying}
M.~Curmei, S.~Dean, and B.~Recht, ``Quantifying availability and discovery in
  recommender systems via stochastic reachability,'' in \emph{Proc. Int. Conf.
  Mach. Learn.}, ser. PMLR, vol. 139, 2021, pp. 2265--2275.

\bibitem{li2026survey}
\BIBentryALTinterwordspacing
Z.~Li, G.~Long, C.~Zhang, H.~Zhang, J.~Jiang, and C.~Zhang, ``A survey of
  personalized federated foundation models for privacy-preserving
  recommendation,'' in \emph{Proc. Int. Joint Conf. Artif. Intell.}, 2026.
  [Online]. Available: \url{https://arxiv.org/abs/2506.11563}
\BIBentrySTDinterwordspacing

\bibitem{rendle2009bpr}
S.~Rendle, C.~Freudenthaler, Z.~Gantner, and L.~Schmidt-Thieme, ``{BPR}:
  Bayesian personalized ranking from implicit feedback,'' in \emph{Proc. Conf.
  Uncertainty Artif. Intell.}, 2009, pp. 452--461.

\bibitem{bottou2018optimization}
L.~Bottou, F.~E. Curtis, and J.~Nocedal, ``Optimization methods for large-scale
  machine learning,'' \emph{SIAM Rev.}, vol.~60, no.~2, pp. 223--311, 2018.

\bibitem{hou2024bridging}
Y.~Hou, J.~Li, X.~Fu, Z.~He, A.~Yan, X.~Chen, and J.~McAuley, ``Bridging
  language and items for retrieval and recommendation: Benchmarking {LLMs} as
  semantic encoders,'' in \emph{Proc. Annu. Meeting Assoc. Comput.
  Linguistics}, 2026, pp. 3251--3265.

\bibitem{harper2015movielens}
F.~M. Harper and J.~A. Konstan, ``The {MovieLens} datasets: History and
  context,'' \emph{ACM Trans. Interact. Intell. Syst.}, vol.~5, no.~4, pp.
  1--19, 2016.

\bibitem{zhang2024ninerec}
J.~Zhang, Y.~Cheng, Y.~Ni, Y.~Pan, Z.~Yuan, J.~Fu, Y.~Li, J.~Wang, and F.~Yuan,
  ``{NineRec}: A benchmark dataset suite for evaluating transferable
  recommendation,'' \emph{IEEE Trans. Pattern Anal. Mach. Intell.}, vol.~47,
  no.~7, pp. 5256--5267, 2025.

\bibitem{he2017neural}
X.~He, L.~Liao, H.~Zhang, L.~Nie, X.~Hu, and T.-S. Chua, ``Neural collaborative
  filtering,'' in \emph{Proc. ACM Web Conf.}, 2017, pp. 173--182.

\bibitem{krichene2020sampled}
W.~Krichene and S.~Rendle, ``On sampled metrics for item recommendation,'' in
  \emph{Proc. ACM SIGKDD Int. Conf. Knowl. Discovery Data Mining}, 2020, pp.
  1748--1757.

\bibitem{jarvelin2002cumulated}
K.~J{\"a}rvelin and J.~Kek{\"a}l{\"a}inen, ``Cumulated gain-based evaluation of
  {IR} techniques,'' \emph{ACM Trans. Inf. Syst.}, vol.~20, no.~4, pp.
  422--446, 2002.

\bibitem{yu2023multiview}
P.~Yu, Z.~Tan, G.~Lu, and B.-K. Bao, ``Multi-view graph convolutional network
  for multimedia recommendation,'' in \emph{Proc. ACM Int. Conf. Multimedia},
  2023, pp. 6576--6585.

\bibitem{yang2025fitmm}
W.~Yang, R.~Zhong, Y.~Chen, S.~Li, H.~Ping, C.~Lu, and P.~Jiang, ``{FITMM}:
  Adaptive frequency-aware multimodal recommendation via information-theoretic
  representation learning,'' in \emph{Proc. ACM Int. Conf. Multimedia}, 2025,
  pp. 6193--6202.

\bibitem{yu2025principal}
P.~Yu, Z.~Tan, G.~Lu, and B.-K. Bao, ``Mind individual information! {Principal}
  graph learning for multimedia recommendation,'' in \emph{Proc. AAAI Conf.
  Artif. Intell.}, vol.~39, no.~12, 2025, pp. 13\,096--13\,105.

\bibitem{li2026msca}
J.~Li and H.~Wang, ``Multi-view semantic contrastive alignment for multimodal
  recommendation,'' in \emph{Proc. ACM Web Conf.}, 2026, pp. 5941--5952.

\bibitem{liang2022mind}
W.~Liang, Y.~Zhang, Y.~Kwon, S.~Yeung, and J.~Zou, ``Mind the gap:
  Understanding the modality gap in multi-modal contrastive representation
  learning,'' in \emph{Proc. Adv. Neural Inf. Process. Syst.}, vol.~35, 2022,
  pp. 17\,612--17\,625.

\bibitem{benzaken2022bitfit}
E.~Ben~Zaken, Y.~Goldberg, and S.~Ravfogel, ``{BitFit}: Simple
  parameter-efficient fine-tuning for transformer-based masked
  language-models,'' in \emph{Proc. Annu. Meeting Assoc. Comput. Linguistics},
  2022, pp. 1--9.

\bibitem{radford2021clip}
A.~Radford, J.~W. Kim, C.~Hallacy, A.~Ramesh, G.~Goh, S.~Agarwal, G.~Sastry,
  A.~Askell, P.~Mishkin, J.~Clark, G.~Krueger, and I.~Sutskever, ``Learning
  transferable visual models from natural language supervision,'' in
  \emph{Proc. Int. Conf. Mach. Learn.}, ser. PMLR, vol. 139, 2021, pp.
  8748--8763.

\bibitem{akiba2019optuna}
T.~Akiba, S.~Sano, T.~Yanase, T.~Ohta, and M.~Koyama, ``{Optuna}: A
  next-generation hyperparameter optimization framework,'' in \emph{Proc. ACM
  SIGKDD Int. Conf. Knowl. Discovery Data Mining}, 2019, pp. 2623--2631.

\bibitem{mironov2017renyi}
I.~Mironov, ``{R{\'e}nyi} differential privacy,'' in \emph{Proc. IEEE Comput.
  Secur. Found. Symp.}, 2017, pp. 263--275.

\bibitem{balle2020hypothesis}
\BIBentryALTinterwordspacing
B.~Balle, G.~Barthe, M.~Gaboardi, J.~Hsu, and T.~Sato, ``Hypothesis testing
  interpretations and {R{\'e}nyi} differential privacy,'' in \emph{Proc. Int.
  Conf. Artif. Intell. Stat.}, ser. PMLR, vol. 108, 2020, pp. 2496--2506.
  [Online]. Available: \url{https://proceedings.mlr.press/v108/balle20a.html}
\BIBentrySTDinterwordspacing

\end{thebibliography}

\clearpage

\appendices

\section{Notation}
\label{app:notation}

Table~\ref{tab:notation} summarizes the symbols used in SERA's formulation and analysis.
For user $u$ and item $i$, the final score satisfies $s(u,i)=\anchor(u,i)+r(u,i)$ at inference, where the semantic residual score $r(u,i)$ acts as a content correction to the collaborative score $\anchor(u,i)$.
The full item set $\mathcal{I}$ defines the user-adaptive residual scale $\sigma_u$, while the candidate set $\mathcal{C}_u$ contains the items ranked for that user.
The residual budget $\kappa$ scales the semantic residual score relative to $\sigma_u$, and the residual bound $R_u$ caps its magnitude for each item.

\begin{table}[!t]
\centering
\caption{Notation and definitions used in the main text and appendices.}
\label{tab:notation}
\setlength{\tabcolsep}{0pt}
\renewcommand{\arraystretch}{1.1}
\newcommand{\notationcolumns}[1]{%
  {\begin{tabularx}{\linewidth}{@{}p{82pt}@{\hspace{6pt}}>{\raggedright\arraybackslash}X@{}}#1\end{tabularx}}%
}
\begin{tabularx}{\columnwidth}{@{}c@{\hspace{2.5pt}}|@{\hspace{2.5pt}}X@{}}
\toprule
Group & \notationcolumns{Symbol & Meaning \\} \\
\midrule
\rotatebox[origin=c]{90}{Users and items} & \notationcolumns{
$u,\ i,j$ & user and item indices \\
$\mathcal{U},\ N$ & user set and its size \\
$\mathcal{I},\ M$ & item set and its size \\
$\mathcal{C}_u,\ M_u$ & candidate set and its size \\
$\mathcal{H}_u$ & user interaction history \\
} \\
\midrule
\rotatebox[origin=c]{90}{Scores and semantic residuals} & \notationcolumns{
$\anchor(u,i),\ r(u,i),\ s(u,i)$ & collaborative, semantic residual, and final scores \\
$\mathbf{z}_i,\ \mathbf{q}_u$ & content representation and query \\
$\beta_{ui}$ & history weight \\
$h_{\phi}$ & shared content projector \\
$\phi,\ |\phi|$ & projector parameters and count \\
$\theta_{\mathrm s},\theta_u$ & shared and personalized FedRS parameters \\
$\Pi(\cdot),\ \varepsilon$ & L2 normalization and positive floor \\
$c_{ui},\ g_{ui}$ & cosine similarity and bounded shape term \\
$\kappa$ & global residual budget \\
$\sigma_u,\ \hat\sigma_u$ & user-adaptive residual scale and estimate \\
$\bar a_u$ & mean collaborative score over $\mathcal{I}$ \\
$v_u,\ \hat v_u$ & standard deviation of collaborative scores over $\mathcal I$ and estimate \\
$R_u$ & residual bound \\
$\tau_r,\ \tau_q$ & residual and attention temperatures \\
} \\
\midrule
\rotatebox[origin=c]{90}{Ranking and score analysis} & \notationcolumns{
$K$ & number of recommended items \\
$\anchor^{(K)}_{u},\ \anchor^{(K+1)}_{u}$ & collaborative scores at ranks $K$ and $K+1$ \\
$\mathcal T_u,\ \tilde{\mathcal T}_u$ & collaborative and final top-$K$ sets \\
$\gamma_{u,ij}$ & collaborative score gap \\
$\gamma^{(K)}_{u,j}$ & collaborative boundary gap of item $j$ \\
$\delta_{u,ij}$ & pairwise shape-term difference \\
$\kappa_{u,ij}$ & pairwise ranking change point \\
$P_i,\ Q_i$ & competitors that $i$ passes or that pass $i$ \\
$\rho_u(i),\ \tilde\rho_u(i)$ & collaborative and final item ranks \\
$n^{+}_{u,j},\ n^{-}_{u,j}$ & candidates above/below $\anchor(u,j)\pm2R_u$ \\
$\kappa^{\mathrm{in}}_{u,j}$ & entry budget of item $j$ \\
$\bar a_u^{\mathcal C},\ \bar g_u^{\mathcal C},\ \bar s_u^{\mathcal C}$ & means over $\mathcal{C}_u$ \\
$\operatorname{sd}_u(\cdot)$ & standard deviation over $\mathcal{C}_u$ \\
$\chi_u$ & scale ratio $\sigma_u/\operatorname{sd}_u(a)$ \\
$w_u$ & effective mixing weight \\
} \\
\midrule
\rotatebox[origin=c]{90}{Federated training} & \notationcolumns{
$\mathcal{B}_u^{+},\ \mathcal{B}_u^{-}$ & positive and negative mini-batches \\
$\sg[\cdot]$ & stop-gradient operator \\
$\mathcal{L}_u(s),\ \mathcal{L}(s)$ & client risk and federated objective \\
$\sigma(\cdot)$ & logistic function \\
$\mathcal{J}(\phi)$ & SERA residual objective \\
$\mathcal{J}_{\bar{\alpha}}(\phi)$ & residual objective with expected client weights \\
$\mathcal{L}_{\mathrm{base}}(\anchor)$ & collaborative training objective \\
$\mathcal{S}_t$ & clients selected in round $t$ \\
$\mathbf{Z}^t,\ \mathbf{G}_u^t$ & content table and its client gradient \\
$\alpha_u,\alpha_{u,t},\bar{\alpha}_u$ & population, round, and expected client weights \\
$\lambda$ & residual loss weight \\
} \\
\bottomrule
\end{tabularx}
\end{table}

\section{Experimental Settings}
\label{app:experimental-comparisons}

\subsection{Training and Evaluation}
Item features and the federated training setup follow FedVLR~\cite{li2024fedvlr}.
A pretrained CLIP \mbox{ViT-B/32} model~\cite{radford2021clip} extracts frozen 512-dimensional text and image features, and the content representations have 64 dimensions.
All clients participate in every round, and each client runs $E=5$ local steps with a batch size of 2{,}048.
Unless specified otherwise below, training stops after 10 consecutive rounds without improvement in validation NDCG@50 or after 100 rounds, and testing uses the best validation round.
The runs in Fig.~\ref{fig:privacy-noise} instead use a fixed number of rounds $T$ without early stopping or checkpoint selection and report the mean test NDCG@50 over rounds $T-3$ to $T$.
Ranking over $\mathcal I$ masks each user's training and validation items.

\subsection{Model Selection}
For the standard comparisons, we tune SERA with the tree-structured Parzen estimator in Optuna~\cite{akiba2019optuna}, using 500 trials for each backbone and dataset.
The search space covers $\kappa\in\{0,0.25,0.5,1,2,3,4,6,8,16\}$, $\tau_r\in\{0.1,0.5,1,2,5\}$, $\tau_q\in\{0.1,0.5,1,2,5,10\}$, $\lambda$ drawn log-uniformly from $[0.01,3]$, and projectors that are linear or have $2$ or $3$ hidden ReLU layers of width $64$, $128$, or $256$.
For FedAvg and PFedRec on ML, the 10 best hyperparameter settings are each trained with five seeds, and the one with the highest mean validation NDCG@50 is kept.
The search uses seed 2024. All reported results except Fig.~\ref{fig:bias-centering} use seeds 2025 to 2029, and the reselection on ML uses only their validation scores.
Table~\ref{tab:knobs} lists the selected hyperparameter values.

\begin{table}[!t]
  \centering
  \caption{Validation-selected SERA hyperparameters, with $\lambda$ rounded to four decimal places. Width and Depth give the hidden width and layer count, and zeros mark a purely linear projector.}
  \label{tab:knobs}
  \small
  \setlength{\tabcolsep}{2pt}
  \renewcommand{\arraystretch}{1.1}
  \begin{tabular*}{\columnwidth}{@{\extracolsep{\fill}}llrrrrrr@{}}
    \toprule
    Backbone & Dataset & $\kappa$ & $\tau_r$ & $\tau_q$ & $\lambda$ & Width & Depth \\
    \midrule
    \multirow{6}{*}{FedAvg} & Beauty & 16 & 1 & 5 & 0.0925 & 256 & 3 \\
     & Cards & 8 & 1 & 10 & 0.0201 & 64 & 2 \\
     & ML & 0.5 & 0.1 & 0.1 & 0.0400 & 64 & 2 \\
     & Dance & 0.5 & 0.1 & 2 & 0.0182 & 128 & 2 \\
     & Food & 3 & 1 & 0.1 & 0.0626 & 0 & 0 \\
     & Movie & 0.5 & 2 & 10 & 0.0114 & 0 & 0 \\
    \midrule
    \multirow{6}{*}{PFedRec} & Beauty & 16 & 1 & 10 & 0.0143 & 0 & 0 \\
     & Cards & 16 & 5 & 5 & 0.2621 & 0 & 0 \\
     & ML & 4 & 2 & 2 & 0.2536 & 128 & 2 \\
     & Dance & 16 & 2 & 10 & 0.0110 & 128 & 3 \\
     & Food & 16 & 5 & 2 & 0.1931 & 0 & 0 \\
     & Movie & 6 & 5 & 1 & 0.0421 & 256 & 3 \\
    \midrule
    \multirow{6}{*}{FedRAP} & Beauty & 16 & 0.5 & 0.5 & 0.6123 & 256 & 3 \\
     & Cards & 6 & 2 & 10 & 0.2786 & 256 & 3 \\
     & ML & 8 & 1 & 2 & 0.0141 & 0 & 0 \\
     & Dance & 16 & 2 & 1 & 0.0409 & 0 & 0 \\
     & Food & 16 & 2 & 2 & 0.0137 & 0 & 0 \\
     & Movie & 16 & 2 & 0.5 & 0.0545 & 0 & 0 \\
    \bottomrule
  \end{tabular*}
\end{table}

\subsection{Computing Environment}
Experiments run on servers with 16 CPU cores, 125\,GB of RAM, and two NVIDIA RTX A5500 GPUs with 24\,GB of memory each, under Red Hat Enterprise Linux 8.10 with Python 3.12, PyTorch 2.7.1, and CUDA Toolkit 12.8.

\subsection{Noisy-Gradient Experiments and Accounting}
\label{app:privacy-accounting}
The experiments in Fig.~\ref{fig:privacy-noise} use full participation and $E=5$ local steps.
Client $u$ clips its accumulated $\mathbf G_u^t$ to Frobenius norm $C$ and adds independent $\mathcal N(0,\nu^2C^2)$ noise to every entry, including untouched rows.
The server weights these messages by the training interaction counts $n_u$ and divides their sum by $\sum_u n_u$ before updating the projector.
The accounting assumes an honest server without client collusion, honest clients without dropout, and ideal real-valued secure aggregation revealing only this sum, the total weight, and the participant list.
Weights and participation are fixed and public; interaction counts and participation identities are not protected.
The simulator processes individual messages and retains diagnostics; it does not implement this secure-aggregation idealization.

For the conditional aggregation mechanism at a common history, zero-out adjacency replaces one user's data-dependent gradient contribution by zero while retaining its weight, participation, and independent noise.
With all other contributions fixed, the sum's sensitivity is at most $C\max_u n_u$, and each noise coordinate has standard deviation $\nu C\lVert\mathbf n\rVert_2$.
Writing $r=\lVert\mathbf n\rVert_2/\max_u n_u$, the accountant conservatively uses $z=\nu\sqrt{r^2-1}$ instead of $\nu r$, discarding the largest client's noise contribution.
The resulting per-round R\'enyi differential privacy (RDP) cost is at most $\alpha/(2z^2)$~\cite{mironov2017renyi}; full participation provides no subsampling amplification.
We convert the sum of these conditional costs using~\cite{balle2020hypothesis}:
\begin{equation}
\varepsilon=\min_{\alpha\in\mathcal A}
\left[\frac{T\alpha}{2z^2}+\ln\frac{\alpha-1}{\alpha}
-\frac{\ln\delta+\ln\alpha}{\alpha-1}\right],
\label{eq:privacy-accounting}
\end{equation}
where $\mathcal A=\{1.01,1.02,\ldots,200.99\}$ and $\delta=10^{-5}$.
The accompanying \texttt{privacy\_accounting.csv} provides the full-precision inputs and the 72 accounting values.

\begin{table}[!t]
\centering
\caption{Fixed configurations for Fig.~\ref{fig:privacy-noise}. $C$, $r$, and $\nu_1$ are rounded to four significant digits; accounting uses the stored full-precision values. Target SNR $s$ sets $\nu=\nu_1/s$.}
\label{tab:privacy-accounting}
\small
\setlength{\tabcolsep}{2pt}
\renewcommand{\arraystretch}{1.05}
\begin{tabular*}{\columnwidth}{@{\extracolsep{\fill}}llrrrr@{}}
\toprule
Backbone & Dataset & $T$ & $C$ & $r$ & $\nu_1$ \\
\midrule
FedAvg & Beauty & 37 & 0.006852 & 2.853 & 0.004479 \\
PFedRec & Beauty & 64 & 0.009684 & 2.853 & 0.005004 \\
FedRAP & Beauty & 88 & 0.8871 & 2.853 & 0.007278 \\
FedAvg & Cards & 77 & 0.004160 & 3.314 & 0.02474 \\
PFedRec & Cards & 70 & 0.1069 & 3.314 & 0.02125 \\
FedRAP & Cards & 42 & 0.08127 & 3.314 & 0.02351 \\
FedAvg & Food & 100 & 0.04754 & 8.117 & 0.003391 \\
PFedRec & Food & 96 & 0.01354 & 8.117 & 0.005567 \\
FedRAP & Food & 91 & 0.0001715 & 8.117 & 0.03665 \\
\bottomrule
\end{tabular*}
\end{table}

These experiments reuse earlier hyperparameter configurations selected by validation NDCG@10.
A noiseless calibration run with seed 2024 on the same data selects $T$ by the highest four-round mean validation NDCG@10 and sets $C$ to the median, over rounds $1$ to $T$, of the median client gradient norm.
Calibration sets $\nu_1$ to match a median predicted SNR of one, using an estimated clipping effect and an approximation to the aggregate noise norm.
Each target $s$ then fixes $\nu=\nu_1/s$ throughout training, as listed in Table~\ref{tab:privacy-accounting}. The target $s$ is not the SNR measured in the noisy run.

The accounting holds the selected configuration fixed across neighboring inputs and includes only the conditional gradient mechanisms above.
It excludes collaborative releases, same-data calibration, and hyperparameter selection.
Unprotected collaborative messages can affect subsequent gradients, so these costs do not establish a privacy bound for the complete gradient-sum sequence or the full training procedure.

\section{Design Component Comparisons}
\label{app:design-comparisons}
\begin{table*}[!t]
  \centering
  \caption{Performance of SERA and three simpler alternatives on Beauty, Cards, and Food in Recall@50 and NDCG@50, reported as means $\pm$ standard deviations in percent across five paired seeds.
  The Scale group gives all users a shared scale, the Gradient group trains the backbone with the residual loss, and the Weighting group weights history items uniformly.
  In each group, the backbone row reports SERA, run with the same five seeds as the alternative below it.
  \textbf{Bold} marks the higher mean in each pair.}
  \label{tab:design-comparisons}
  \small
  \setlength{\tabcolsep}{4pt}
  \renewcommand{\arraystretch}{0.95}
  \begin{tabular}{c|l cc cc cc}
    \toprule
    \multicolumn{2}{c}{\textbf{Method}} & \multicolumn{2}{c}{\textbf{Beauty}} & \multicolumn{2}{c}{\textbf{Cards}} & \multicolumn{2}{c}{\textbf{Food}} \\
    \multicolumn{2}{c}{} & \textbf{Recall@50} & \textbf{NDCG@50} & \textbf{Recall@50} & \textbf{NDCG@50} & \textbf{Recall@50} & \textbf{NDCG@50} \\
    \midrule
    \rowcolor{RoyalBlue!6}
    \cellcolor{white} & FedAvg & \bestresult{21.03 $\pm$ 3.56} & \bestresult{5.87 $\pm$ 0.94} & 68.97 $\pm$ 2.89 & 21.39 $\pm$ 4.96 & \bestresult{12.78 $\pm$ 3.29} & \bestresult{3.69 $\pm$ 1.38} \\
    & \quad w/ shared scale & 20.08 $\pm$ 6.14 & 5.38 $\pm$ 1.70 & \bestresult{72.84 $\pm$ 3.58} & \bestresult{23.81 $\pm$ 3.31} & 11.75 $\pm$ 3.96 & 3.45 $\pm$ 1.29 \\
    \rowcolor{RoyalBlue!6}
    \cellcolor{white} & PFedRec & 24.58 $\pm$ 2.08 & 7.05 $\pm$ 0.70 & 70.40 $\pm$ 6.30 & 24.64 $\pm$ 6.66 & \bestresult{18.83 $\pm$ 1.45} & \bestresult{6.01 $\pm$ 0.63} \\
    & \quad w/ shared scale & \bestresult{24.90 $\pm$ 1.61} & \bestresult{7.23 $\pm$ 0.57} & \bestresult{74.86 $\pm$ 0.91} & \bestresult{29.45 $\pm$ 0.36} & 16.99 $\pm$ 0.54 & 5.38 $\pm$ 0.15 \\
    \rowcolor{RoyalBlue!6}
    \cellcolor{white} & FedRAP & \bestresult{20.87 $\pm$ 5.15} & 5.80 $\pm$ 1.56 & \bestresult{82.33 $\pm$ 0.61} & \bestresult{33.45 $\pm$ 1.22} & 13.89 $\pm$ 2.73 & 3.94 $\pm$ 0.87 \\
    \multirow{-6}{*}{\rotatebox{90}{Scale}}
    & \quad w/ shared scale & 20.55 $\pm$ 5.21 & \bestresult{5.88 $\pm$ 1.67} & 81.27 $\pm$ 0.97 & 33.04 $\pm$ 1.36 & \bestresult{14.53 $\pm$ 1.07} & \bestresult{4.40 $\pm$ 0.33} \\
    \midrule
    \rowcolor{RoyalBlue!6}
    \cellcolor{white} & FedAvg & \bestresult{21.03 $\pm$ 3.56} & \bestresult{5.87 $\pm$ 0.94} & \bestresult{70.13 $\pm$ 3.36} & \bestresult{23.14 $\pm$ 5.13} & 11.75 $\pm$ 3.45 & 3.35 $\pm$ 1.41 \\
    & \quad w/ joint training & 20.87 $\pm$ 3.29 & 5.77 $\pm$ 0.85 & 70.08 $\pm$ 3.42 & 22.73 $\pm$ 5.30 & \bestresult{12.84 $\pm$ 4.26} & \bestresult{3.85 $\pm$ 1.70} \\
    \rowcolor{RoyalBlue!6}
    \cellcolor{white} & PFedRec & 24.59 $\pm$ 2.08 & 7.05 $\pm$ 0.70 & 70.82 $\pm$ 6.57 & 24.93 $\pm$ 6.92 & \bestresult{19.60 $\pm$ 1.35} & \bestresult{6.28 $\pm$ 0.61} \\
    & \quad w/ joint training & \bestresult{25.30 $\pm$ 1.61} & \bestresult{7.12 $\pm$ 0.72} & \bestresult{74.16 $\pm$ 1.74} & \bestresult{28.56 $\pm$ 0.57} & 17.64 $\pm$ 1.86 & 5.64 $\pm$ 0.65 \\
    \rowcolor{RoyalBlue!6}
    \cellcolor{white} & FedRAP & 20.87 $\pm$ 5.15 & \bestresult{5.80 $\pm$ 1.56} & \bestresult{81.86 $\pm$ 1.55} & \bestresult{33.16 $\pm$ 1.39} & 12.84 $\pm$ 1.63 & 3.61 $\pm$ 0.40 \\
    \multirow{-6}{*}{\rotatebox{90}{Gradient}}
    & \quad w/ joint training & \bestresult{21.82 $\pm$ 2.67} & 5.77 $\pm$ 0.66 & 81.54 $\pm$ 1.39 & 32.52 $\pm$ 1.33 & \bestresult{14.67 $\pm$ 3.07} & \bestresult{4.34 $\pm$ 0.96} \\
    \midrule
    \rowcolor{RoyalBlue!6}
    \cellcolor{white} & FedAvg & 21.03 $\pm$ 3.56 & 5.87 $\pm$ 0.94 & 68.76 $\pm$ 3.30 & 21.53 $\pm$ 5.23 & \bestresult{11.75 $\pm$ 3.45} & \bestresult{3.35 $\pm$ 1.41} \\
    & \quad w/ uniform weighting & \bestresult{21.19 $\pm$ 2.35} & \bestresult{5.98 $\pm$ 0.63} & \bestresult{70.56 $\pm$ 2.74} & \bestresult{24.53 $\pm$ 4.73} & 11.25 $\pm$ 2.08 & 3.10 $\pm$ 0.72 \\
    \rowcolor{RoyalBlue!6}
    \cellcolor{white} & PFedRec & 24.51 $\pm$ 2.07 & 7.01 $\pm$ 0.74 & 70.82 $\pm$ 6.57 & 24.93 $\pm$ 6.92 & 18.48 $\pm$ 1.77 & 5.91 $\pm$ 0.60 \\
    & \quad w/ uniform weighting & \bestresult{24.74 $\pm$ 1.18} & 7.01 $\pm$ 0.43 & \bestresult{72.31 $\pm$ 5.90} & \bestresult{27.05 $\pm$ 6.05} & \bestresult{18.66 $\pm$ 2.12} & \bestresult{5.97 $\pm$ 0.74} \\
    \rowcolor{RoyalBlue!6}
    \cellcolor{white} & FedRAP & \bestresult{20.87 $\pm$ 5.15} & 5.80 $\pm$ 1.56 & \bestresult{81.75 $\pm$ 1.28} & \bestresult{33.09 $\pm$ 0.93} & \bestresult{15.00 $\pm$ 3.08} & \bestresult{4.39 $\pm$ 0.99} \\
    \multirow{-6}{*}{\rotatebox{90}{Weighting}}
    & \quad w/ uniform weighting & 20.32 $\pm$ 4.83 & 5.80 $\pm$ 1.59 & 81.38 $\pm$ 1.69 & 32.79 $\pm$ 1.15 & 14.70 $\pm$ 3.04 & 4.28 $\pm$ 0.91 \\
    \bottomrule
  \end{tabular}
\end{table*}

SERA's semantic residual score has three design components: the user-adaptive residual scale in Eq.~\eqref{eq:residual}, the stop-gradient in Eq.~\eqref{eq:score}, and the history weights in Eq.~\eqref{eq:weights}.
To test each one, we swap it for a simpler alternative and keep the rest of SERA fixed.
The shared scale gives every user the same scale $\bar{\sigma}$, the mean of $\sigma_u$ over the $N$ users of the trained model.
Joint training removes the stop-gradient on the collaborative score, so the residual loss also updates the backbone, while $\beta_{ui}$ and $\sigma_u$ keep their stop-gradient.
Uniform weighting is approximated by multiplying $\tau_q$ by $10^6$.

\paragraph{User-Adaptive Residual Scale}
SERA scales each user's semantic residual scores by $\sigma_u$, the standard deviation of that user's collaborative scores over the item set.
Larger variation in collaborative scores permits larger semantic residual magnitudes. By Proposition~\ref{prop:residual-bound} and Theorem~\ref{thm:reach}, one budget $\kappa$ then gives every user a reach band of equal width in units of $\sigma_u$.
Under the shared scale, these widths differ across users.
Replacing $\sigma_u$ by $\bar{\sigma}$ scales user $u$'s semantic residual scores by the budget ratio $\bar{\sigma}/\sigma_u$, so users with $\sigma_u$ below $\bar{\sigma}$ get a wider reach band than under SERA and those above it a narrower one.

Table~\ref{tab:design-comparisons} shows two consistent differences that keep their sign in every seed.
Both appear on Cards and point in opposite directions: compared with SERA, the shared scale raises Recall@50 for FedAvg and lowers it for FedRAP.
The budget sweeps on Cards in Figs.~\ref{fig:bound-ndcg} and~\ref{fig:budget-users} show a similar pattern.
Above the trained budget, SERA on FedAvg keeps gaining accuracy, whereas SERA on FedRAP loses accuracy as its harmed share of users grows fastest.
The shared scale raises the budget ratio of users with small $\sigma_u$, so for them it acts like a larger budget. On Cards, the shared scale helps SERA on backbones that benefit from larger budgets and hurts SERA on the backbone that suffers from them.
SERA on PFedRec follows the same pattern on Cards. It gains accuracy above its trained budget, and the shared scale gives it higher mean accuracy. Every other difference changes sign across seeds.
SERA's user-adaptive residual scale thus gives every user the same reach band in units of $\sigma_u$, and its effect on accuracy depends on which collaborative backbone it corrects.

\paragraph{Stop-Gradient}
SERA's stop-gradient in Eq.~\eqref{eq:score} passes the collaborative score forward and blocks its gradient, and SERA also applies the stop-gradient to the collaborative scores used for the history weights and $\sigma_u$.
Consequently, the residual loss in Eq.~\eqref{eq:objective_sera} trains only the projector, and the backbone keeps learning from $\mathcal{L}_{\mathrm{base}}(a)$ alone.
All of SERA's content correction therefore sits in the semantic residual score $r(u,i)$, which Proposition~\ref{prop:residual-bound}, Corollary~\ref{cor:certificate}, and Theorem~\ref{thm:reach} bound relative to the trained collaborative scores.
Joint training opens a path from the residual loss to the backbone through $s(u,i)$, and this path changes what the backbone learns.
As Sec.~\ref{sec:collaborative-content} notes, the pairwise loss in Eq.~\eqref{eq:clientrisk} becomes less sensitive to score changes as the margin between the observed item and the sampled negative grows.
Under joint training this margin includes $r(u,i)$, so pairs that $r(u,i)$ already ranks correctly send weaker gradients to the backbone and pairs that it ranks wrongly send stronger ones.
The backbone then learns from errors that already include the semantic residual score, so its collaborative scores carry content information, and the bounds of Sec.~\ref{sec:theory} cover only the content effect left in $r(u,i)$.

Table~\ref{tab:design-comparisons} supports this choice.
Its two consistent differences, both for PFedRec on Food, favor SERA's stop-gradient.
Every other difference changes sign across seeds, including joint training's largest mean advantage, for PFedRec on Cards.
The stop-gradient thus keeps SERA's guarantees in Sec.~\ref{subsec:reach} intact and matches joint training within seed variation.

\paragraph{History Weights}
SERA builds its content query from the user's history and weights each history item by its collaborative score in Eq.~\eqref{eq:weights}.
The temperature $\tau_q$ sets how sharply these weights favor high-scoring items, and a larger $\tau_q$ spreads them more evenly.
Uniform weighting is the large-$\tau_q$ limit. In Table~\ref{tab:knobs}, validation picks one of the two largest $\tau_q$ values in five of the nine configurations.
How much the weights matter depends on the history itself.
The content query in Eq.~\eqref{eq:query} is a normalized weighted average of the representations $\mathbf z_i$ of the history items, so the weights can only mix the directions of these items into $\mathbf q_u$.
When the items point in similar directions, any weighting gives nearly the same query, and Eqs.~\eqref{eq:compat} to~\eqref{eq:residual} then give nearly the same semantic residual scores.
The history weights therefore matter mainly for users whose histories mix dissimilar content.

In Table~\ref{tab:design-comparisons}, every difference between SERA's history weights and uniform weighting changes sign across seeds, including FedAvg on Food, where SERA's validated $\tau_q$ is the grid's smallest.
On these datasets, SERA's history weights thus match uniform weights within seed variation.
The content ablations in Fig.~\ref{fig:content-ablation} show that content drives SERA's gains.

\section{Deployment Cost}
\label{app:deploy-cost}
\begin{table*}[!t]
  \centering
  \caption{Space costs of the federated methods on the six datasets.
  (a) Communication per client, direction, and round in KiB.
  (b) On-device storage in MiB, covering shared parameters, one client's personal state, and the feature or content table read by the client.
  (c) Median peak GPU memory during one training round minus the memory allocated at its start, in MiB.
  SERA rows count the 64-dimensional content table and its gradient, with the projector and the item features on the server.
  ``-'' marks AMMFRS runs that exceed the host memory because the simulator keeps the personal models of all clients.}
  \label{tab:cost-space}
  \setlength{\tabcolsep}{2pt}
  \renewcommand{\arraystretch}{0.95}
  \resizebox{\textwidth}{!}{%
  \begin{tabular}{c|l cccccc cccccc cccccc}
    \toprule
    \multicolumn{2}{c}{\textbf{Method}} & \multicolumn{6}{c}{\textbf{(a) Communication}} & \multicolumn{6}{c}{\textbf{(b) On-device storage}} & \multicolumn{6}{c}{\textbf{(c) Peak GPU memory}} \\
    \cmidrule(lr){3-8}\cmidrule(lr){9-14}\cmidrule(lr){15-20}
    \multicolumn{2}{c}{} & \textbf{Beauty} & \textbf{Cards} & \textbf{ML} & \textbf{Dance} & \textbf{Food} & \textbf{Movie} & \textbf{Beauty} & \textbf{Cards} & \textbf{ML} & \textbf{Dance} & \textbf{Food} & \textbf{Movie} & \textbf{Beauty} & \textbf{Cards} & \textbf{ML} & \textbf{Dance} & \textbf{Food} & \textbf{Movie} \\
    \midrule
    & FedMMR & 105 & 386 & 460 & 4{,}378 & 3{,}276 & 6{,}082 & 0.10 & 0.38 & 0.45 & 4.28 & 3.20 & 5.94 & 2.15 & 5.63 & 4.57 & 104.01 & 67.06 & 158.20 \\
    & AMMFRS & 7{,}162 & 8{,}224 & 4{,}118 & - & - & - & 15.38 & 16.57 & 22.30 & - & - & - & 71.06 & 72.09 & 451.41 & - & - & - \\
    & FedAvg & 44 & 258 & 7{,}300 & 838 & 281 & 1{,}264 & 0.04 & 0.26 & 7.13 & 0.82 & 0.28 & 1.24 & 0.12 & 1.25 & 34.59 & 6.64 & 2.01 & 10.29 \\
    & \quad + FedVLR & 574 & 517 & 1{,}397 & 904 & 766 & 1{,}117 & 2.00 & 1.06 & 15.67 & 7.48 & 5.19 & 11.01 & 4.66 & 2.58 & 116.22 & 5.22 & 3.78 & 7.23 \\
    & \quad + GFMFR & 89 & 32 & 912 & 419 & 281 & 632 & 0.09 & 0.03 & 0.89 & 0.41 & 0.28 & 0.62 & 7.12 & 5.36 & 49.44 & 73.78 & 49.93 & 34.72 \\
    \rowcolor{RoyalBlue!6}
    \cellcolor{white} & \quad + SERA & 178 & 64 & 1{,}825 & 838 & 562 & 1{,}264 & 0.17 & 0.06 & 1.78 & 0.82 & 0.55 & 1.23 & 0.54 & 0.19 & 16.53 & 4.08 & 2.64 & 5.75 \\
    & PFedRec & 89 & 258 & 7{,}300 & 3{,}352 & 1{,}125 & 2{,}528 & 0.09 & 0.25 & 7.13 & 3.28 & 1.10 & 2.47 & 0.44 & 1.50 & 35.67 & 30.23 & 9.43 & 23.48 \\
    & \quad + FedVLR & 574 & 517 & 1{,}397 & 904 & 766 & 1{,}117 & 2.00 & 1.06 & 15.67 & 7.48 & 5.19 & 11.01 & 4.66 & 2.58 & 114.04 & 5.26 & 3.80 & 7.29 \\
    & \quad + GFMFR & 89 & 32 & 912 & 419 & 281 & 632 & 0.09 & 0.03 & 0.89 & 0.41 & 0.27 & 0.62 & 7.12 & 0.13 & 321.54 & 23.97 & 99.16 & 222.57 \\
    \rowcolor{RoyalBlue!6}
    \cellcolor{white} & \quad + SERA & 178 & 64 & 1{,}825 & 838 & 562 & 1{,}264 & 0.17 & 0.06 & 1.78 & 0.82 & 0.55 & 1.23 & 0.47 & 0.19 & 14.39 & 4.28 & 2.66 & 5.81 \\
    & FedRAP & 178 & 64 & 3{,}650 & 210 & 281 & 316 & 0.35 & 0.13 & 7.13 & 0.41 & 0.55 & 0.62 & 0.87 & 0.32 & 24.10 & 1.04 & 1.38 & 1.56 \\
    & \quad + FedVLR & 574 & 517 & 1{,}397 & 904 & 766 & 1{,}117 & 2.09 & 1.09 & 16.56 & 7.89 & 5.47 & 11.63 & 4.70 & 2.66 & 117.55 & 4.65 & 3.06 & 4.54 \\
    & \quad + GFMFR & 122 & 65 & 945 & 452 & 314 & 665 & 0.21 & 0.09 & 1.81 & 0.85 & 0.58 & 1.27 & 14.52 & 13.38 & 2.72 & 37.52 & 49.93 & 4.57 \\
    \rowcolor{RoyalBlue!6}
    \cellcolor{white}\multirow{-14}{*}{\rotatebox{90}{Federated}}
    & \quad + SERA & 178 & 64 & 1{,}825 & 838 & 562 & 1{,}264 & 0.26 & 0.09 & 2.67 & 1.23 & 0.82 & 1.85 & 0.64 & 0.25 & 20.23 & 2.06 & 1.38 & 3.11 \\
    \bottomrule
  \end{tabular}%
  }
\end{table*}

\begin{table*}[!t]
  \centering
  \caption{Time costs of the federated methods on the six datasets as medians in milliseconds.
  (a) Local training of one client in one round, with the shared-model and personal updates summed for AMMFRS.
  (b) Server-side work per round beyond aggregation: FedMMR's server-side multimodal model, GFMFR's clustering, fusion training, and group teachers, and SERA's projection of the full item set with one backward pass and one optimizer step.
  (c) Scoring of all items per user, averaged over one validation pass.
    GFMFR's server stage runs only with a nonzero distillation weight, so its entries are 0.00 where validation selects a zero weight.
  ``-'' marks AMMFRS runs that exceed the host memory because the simulator keeps the personal models of all clients.}
  \label{tab:cost-time}
  \setlength{\tabcolsep}{2pt}
  \renewcommand{\arraystretch}{0.95}
  \resizebox{\textwidth}{!}{%
  \begin{tabular}{c|l cccccc cccccc cccccc}
    \toprule
    \multicolumn{2}{c}{\textbf{Method}} & \multicolumn{6}{c}{\textbf{(a) Client computation}} & \multicolumn{6}{c}{\textbf{(b) Server stage}} & \multicolumn{6}{c}{\textbf{(c) Inference per user}} \\
    \cmidrule(lr){3-8}\cmidrule(lr){9-14}\cmidrule(lr){15-20}
    \multicolumn{2}{c}{} & \textbf{Beauty} & \textbf{Cards} & \textbf{ML} & \textbf{Dance} & \textbf{Food} & \textbf{Movie} & \textbf{Beauty} & \textbf{Cards} & \textbf{ML} & \textbf{Dance} & \textbf{Food} & \textbf{Movie} & \textbf{Beauty} & \textbf{Cards} & \textbf{ML} & \textbf{Dance} & \textbf{Food} & \textbf{Movie} \\
    \midrule
    & FedMMR & 3.63 & 3.62 & 3.77 & 3.79 & 3.69 & 3.93 & 1.20 & 1.25 & 4.09 & 2.66 & 2.54 & 3.89 & 0.43 & 0.41 & 0.47 & 0.48 & 0.47 & 0.48 \\
    & AMMFRS & 82.45 & 73.99 & 49.69 & - & - & - & 0.00 & 0.00 & 0.00 & - & - & - & 4.99 & 3.79 & 3.87 & - & - & - \\
    & FedAvg & 5.76 & 5.65 & 11.71 & 6.55 & 6.47 & 6.93 & 0.00 & 0.00 & 0.00 & 0.00 & 0.00 & 0.00 & 0.23 & 0.27 & 0.28 & 0.22 & 0.21 & 0.20 \\
    & \quad + FedVLR & 13.85 & 14.45 & 19.58 & 14.94 & 14.71 & 15.27 & 0.00 & 0.00 & 0.00 & 0.00 & 0.00 & 0.00 & 0.67 & 0.80 & 0.75 & 0.68 & 0.64 & 0.65 \\
    & \quad + GFMFR & 6.86 & 5.07 & 11.41 & 6.26 & 6.11 & 6.40 & 321.65 & 194.01 & 590.80 & 834.69 & 562.80 & 936.52 & 0.23 & 0.22 & 0.23 & 0.19 & 0.19 & 0.19 \\
    \rowcolor{RoyalBlue!6}
    \cellcolor{white} & \quad + SERA & 12.13 & 11.70 & 22.33 & 13.22 & 12.99 & 13.50 & 0.61 & 0.48 & 0.53 & 0.52 & 0.36 & 0.48 & 0.45 & 0.44 & 0.45 & 0.47 & 0.43 & 0.45 \\
    & PFedRec & 5.63 & 5.54 & 11.75 & 8.09 & 6.46 & 7.69 & 0.00 & 0.00 & 0.00 & 0.00 & 0.00 & 0.00 & 0.20 & 0.19 & 0.19 & 0.22 & 0.17 & 0.18 \\
    & \quad + FedVLR & 13.70 & 14.24 & 20.77 & 14.80 & 14.71 & 15.08 & 0.00 & 0.00 & 0.00 & 0.00 & 0.00 & 0.00 & 0.67 & 0.77 & 0.75 & 0.68 & 0.64 & 0.63 \\
    & \quad + GFMFR & 6.56 & 5.02 & 11.44 & 6.21 & 5.90 & 6.44 & 330.20 & 0.00 & 17{,}451.20 & 280.28 & 1{,}117.99 & 2{,}498.41 & 0.21 & 0.19 & 0.18 & 0.18 & 0.17 & 0.17 \\
    \rowcolor{RoyalBlue!6}
    \cellcolor{white} & \quad + SERA & 12.24 & 11.66 & 22.19 & 12.79 & 12.38 & 12.66 & 0.35 & 0.41 & 0.36 & 0.61 & 0.41 & 0.69 & 0.45 & 0.42 & 0.49 & 0.44 & 0.40 & 0.38 \\
    & FedRAP & 6.98 & 6.73 & 13.34 & 7.43 & 6.99 & 7.63 & 0.00 & 0.00 & 0.00 & 0.00 & 0.00 & 0.00 & 0.25 & 0.24 & 0.51 & 0.22 & 0.19 & 0.19 \\
    & \quad + FedVLR & 15.49 & 14.52 & 20.86 & 16.35 & 15.49 & 15.86 & 0.00 & 0.00 & 0.00 & 0.00 & 0.00 & 0.00 & 0.74 & 0.71 & 0.83 & 0.69 & 0.61 & 0.60 \\
    & \quad + GFMFR & 8.17 & 6.92 & 12.94 & 7.99 & 7.82 & 8.20 & 795.66 & 194.21 & 0.00 & 461.31 & 606.74 & 0.00 & 0.25 & 0.27 & 0.37 & 0.23 & 0.23 & 0.22 \\
    \rowcolor{RoyalBlue!6}
    \cellcolor{white}\multirow{-14}{*}{\rotatebox{90}{Federated}}
    & \quad + SERA & 13.00 & 13.72 & 23.26 & 14.65 & 13.56 & 14.23 & 1.38 & 0.63 & 0.47 & 0.33 & 0.32 & 0.33 & 0.43 & 0.49 & 0.64 & 0.49 & 0.41 & 0.44 \\
    \bottomrule
  \end{tabular}%
  }
\end{table*}

Tables~\ref{tab:cost-space} and~\ref{tab:cost-time} report the space and time costs of SERA and the other federated methods in Tables~\ref{tab:main} and~\ref{tab:main-shortvideo}.
In these runs, FedVLR uses its default hyperparameters, and every other method uses those of the main comparison.
Communication and on-device storage are exact byte counts. Time and memory come from short runs on an idle RTX A5500 GPU, with one warm-up round and three timed rounds of 128 sampled clients each, followed by one warm-up pass and two timed inference passes over all of the validation users.

\paragraph{Space}
Algorithm~\ref{alg:fedsera} keeps the projector on the server, which sets SERA's space cost.
Beyond its backbone, a SERA client exchanges one $M\times d$ table in each direction and stores one copy.
Because this table depends only on $M$ and $d$, SERA's added cost stays the same for any history length, projector size $|\phi|$, and backbone, and SERA's three rows in Table~\ref{tab:cost-space} have the same communication, proportional to the number of items.
This content table also explains how SERA compares with the other adapters.
Compared with GFMFR, which distills server-side multimodal content into each client's backbone using group preference targets, SERA also exchanges its content table.
In return, SERA keeps content in $r(u,i)$, which leaves the backbone's training unchanged.
FedVLR uses the same collaborative item table as SERA, so their communication differs by SERA's $M\times d$ content table and FedVLR's fusion network of about $1.24\times10^5$ parameters.
Consequently, SERA exchanges less than FedVLR on the four datasets with fewer than about $1.9\times10^3$ items and more on ML and Movie.
SERA stores $d$-dimensional rows, a sixteenth of the width of the $d_x$-dimensional features that FedVLR stores, so SERA needs far less on-device storage.
The server-side projector also keeps SERA's GPU memory low.
Because the projector's parameters, optimizer states, and $d_x$-dimensional activations stay on the server, SERA's peak GPU memory stays close to that of a backbone of the same width.
SERA's added tensors, namely the content table, its gradient, and the collaborative scores over the item set, all grow with $M$, so its peak GPU memory on Dance, Food, and Movie rises roughly in proportion to the number of items.
At each history position, SERA keeps a $d$-dimensional row, whereas FedVLR keeps $d_x$-dimensional features and fusion outputs. The long histories of ML multiply this gap, so SERA's peak GPU memory on ML is a small fraction of FedVLR's.

\paragraph{Time}
SERA's time cost follows from the size of its operations.
Its client batches hold one user, and its server projects at most a few thousand items per round.
At these sizes, fixed kernel launch overhead sets the duration of each GPU operation.
SERA's added time therefore depends on how many operations it adds and changes little with the embedding width or the number of items.
Table~\ref{tab:cost-time} confirms this: outside ML, SERA's client computation stays nearly constant across item sets differing in size by over an order of magnitude.

Operation counts also explain how SERA compares with the other methods at each stage.
In local training, SERA adds the parameter-free steps of Eqs.~\eqref{eq:weights} to~\eqref{eq:residual} to its backbone. These steps need fewer operations than FedVLR's fusion network, so SERA's client computation is lower than FedVLR's outside ML.
On the server, SERA runs one forward pass, one backward pass, and one optimizer step per round for all clients, so its server stage is on the same order of magnitude as FedMMR's and more than two orders of magnitude below GFMFR's clustering and fusion rounds.
At inference, SERA computes all similarities of a user with one matrix-vector product between the content table and $\mathbf q_u$.
This is more work than GFMFR, whose inference uses the backbone alone, and less than FedVLR, whose inference runs its fusion network, so SERA's time per user lies between theirs everywhere.

\paragraph{Exception on ML}
ML departs from this pattern because two effects raise SERA's client computation there.
First, the data loading and negative sampling that SERA shares with its backbone grow with each user's interactions, which are more than an order of magnitude more numerous on ML.
Second, SERA pads each history with item~0 to the longest one, and in the backward pass of its table lookup, positions that share a row add their gradients one at a time.
Since the longest ML history is an order of magnitude longer than elsewhere, most positions of an ML history are padding and write in turn into the row of item~0.
Profiling confirms this second effect. The accumulation takes about 4.6\,ms per ML client with SERA's server-side projector and about 0.3\,ms with a client-side projector. Reading the table with \texttt{index\_select} instead cuts SERA's client computation on ML by 17\% and changes it on Beauty by about 3\%, so the serialized accumulation matters only when user histories are long.

\paragraph{Equal-Width Comparison}
The backbones in Table~\ref{tab:cost-time} use their own tuned hyperparameters, so the extra time of each SERA row mixes SERA's cost with differences in hyperparameters.
To measure SERA's own cost, we also time SERA against a backbone of the same 64-dimensional width on Beauty, Cards, and Food.
With its server-side projector, SERA needs 1.8 to 2.0 times the client computation of this backbone and 1.1 to 1.9 times its inference time per user, and SERA's peak GPU memory stays within 0.3\,MiB of the backbone's.
Moving SERA's projector to the client raises its client computation and peak GPU memory, which confirms the cost benefit of placing the projector on the server.

\section{Effective Mixing Weight}
\label{app:mixing}
Eqs.~\eqref{eq:shape} and \eqref{eq:residual} scale the semantic residual score through $\kappa$, $\sigma_u$, and $\tau_r$.
Subtracting the candidate mean from every final score and dividing by the standard deviation of the collaborative scores over $\mathcal C_u$ leave the ranking unchanged. After this rescaling, each final score equals the standardized collaborative score plus $w_u$ times the standardized shape term, so $w_u$ measures how much content counts relative to the collaborative score for user $u$.
For the candidate set $\mathcal{C}_u\subseteq\mathcal{I}$ of size $M_u$, define the following three score averages:
\begin{equation}
\begin{aligned}
\bar a_u^{\mathcal C} &= \frac{1}{M_u}\sum_{i\in\mathcal C_u}\anchor(u,i),\\
\bar g_u^{\mathcal C} &= \frac{1}{M_u}\sum_{i\in\mathcal C_u}g_{ui},\\
\bar s_u^{\mathcal C} &= \frac{1}{M_u}\sum_{i\in\mathcal C_u}s(u,i).
\end{aligned}
\end{equation}
The superscript $\mathcal C$ separates these means from $\bar a_u$ over $\mathcal I$.
The operator $\operatorname{sd}_u(\cdot)$ denotes the population standard deviation over $\mathcal{C}_u$, whereas $\sigma_u$ is still computed over the full item set $\mathcal I$.

\begin{proposition}[Mixing Weight]
\label{prop:mixing}
Fix a user $u$ with $\operatorname{sd}_u(a)>0$ and $\operatorname{sd}_u(g)>0$, and write
$\check a_{ui}=\big(\anchor(u,i)-\bar a_u^{\mathcal C}\big)/\operatorname{sd}_u(a)$ and
$\check g_{ui}=\big(g_{ui}-\bar g_u^{\mathcal C}\big)/\operatorname{sd}_u(g)$.
Then, for every item $i\in\mathcal{C}_u$:
\begin{equation}
\frac{s(u,i)-\bar s_u^{\mathcal C}}{\operatorname{sd}_u(a)}
= \check a_{ui} + w_u\,\check g_{ui},
\;
w_u = \kappa\,\chi_u\operatorname{sd}_u(g),
\label{eq:mixing}
\end{equation}
where $\chi_u=\sigma_u/\operatorname{sd}_u(a)$.
The map on the left is strictly increasing in $s(u,i)$ and preserves both item order and ties.
\end{proposition}
\begin{proof}
Fix a user $u$ and take the displayed means and $\operatorname{sd}_u$ over $\mathcal{C}_u$. The user-adaptive residual scale $\sigma_u$ is still computed over the full item set.
By Eqs.~\eqref{eq:score} and \eqref{eq:residual}, the final score is:
\begin{equation}
s(u,i)=\anchor(u,i)+\kappa\sigma_u g_{ui}.
\label{eq:app-mix-score}
\end{equation}
With $\kappa$ and $\sigma_u$ constant in $i$, averaging Eq.~\eqref{eq:app-mix-score} over $\mathcal{C}_u$ gives:
\begin{equation}
\bar s_u^{\mathcal C}=\bar a_u^{\mathcal C}+\kappa\sigma_u\bar g_u^{\mathcal C}.
\label{eq:app-mix-mean}
\end{equation}
Subtracting Eq.~\eqref{eq:app-mix-mean} from Eq.~\eqref{eq:app-mix-score} centers both terms:
\begin{equation}
  s(u,i)-\bar s_u^{\mathcal C}
  =\big(\anchor(u,i)-\bar a_u^{\mathcal C}\big)
  +\kappa\sigma_u\big(g_{ui}-\bar g_u^{\mathcal C}\big).
\label{eq:app-mix-centered}
\end{equation}
Since $\operatorname{sd}_u(a)>0$ and $\operatorname{sd}_u(g)>0$, dividing Eq.~\eqref{eq:app-mix-centered} by $\operatorname{sd}_u(a)$ and multiplying the second term by $\operatorname{sd}_u(g)/\operatorname{sd}_u(g)$ yields:
\begin{equation}
\frac{s(u,i)-\bar s_u^{\mathcal C}}{\operatorname{sd}_u(a)}
=\check a_{ui}
+\kappa\cdot\frac{\sigma_u}{\operatorname{sd}_u(a)}\cdot\operatorname{sd}_u(g)\cdot\check g_{ui},
\label{eq:app-mix-final}
\end{equation}
which is Eq.~\eqref{eq:mixing} with $\chi_u=\sigma_u/\operatorname{sd}_u(a)$ and $w_u=\kappa\chi_u\operatorname{sd}_u(g)$.
The map $\xi\mapsto\big(\xi-\bar s_u^{\mathcal C}\big)/\operatorname{sd}_u(a)$ is affine with slope $1/\operatorname{sd}_u(a)>0$ and preserves both order and ties.
\end{proof}

\noindent\emph{Remark on what the weight absorbs.}
Eq.~\eqref{eq:app-mix-final} shows that $\kappa$ and $\sigma_u$ enter the standardized score only through $w_u$.
The temperature $\tau_r$ enters $w_u$ through $\operatorname{sd}_u(g)$ and also changes $\check g_{ui}$, so its effect extends beyond $w_u$.
The bound $\operatorname{sd}_u(g)\le\tanh(1/\tau_r)$ follows from Eq.~\eqref{eq:shape} with $|c_{ui}|\le1$ and Popoviciu's inequality, which bounds the standard deviation of a bounded variable by half the width of the interval that contains its values.

The factor $\chi_u$ compares $\sigma_u$ over $\mathcal I$ with the standard deviation of collaborative scores over $\mathcal C_u$, so it reflects masking and the floor $\varepsilon$.
A change of $\tau_r$ can rescale $\operatorname{sd}_u(g)$ differently across users, so preserving each user's mixing weight can require user-specific changes in $\kappa$.
For fixed similarities, $c\mapsto\tanh(c/\tau_r)$ preserves the order by similarity for every $\tau_r>0$.
As $\tau_r\to\infty$, $\tanh(c_{ui}/\tau_r)=c_{ui}/\tau_r+O(\tau_r^{-3})$.
For nonconstant similarities, standardization cancels the leading factor $1/\tau_r$, so $\check g_{ui}$ approaches the standardized similarity.

\section{Sampling Estimate of the User-Adaptive Residual Scale}
\label{app:sampling}
The user-adaptive residual scale requires all $M$ collaborative scores. Sampling approximates it with the same floor $\varepsilon$.
Write $\bar a_u=\frac1M\sum_{i'\in\mathcal I}\anchor(u,i')$ and define $v_u\ge0$ by $v_u^2=\frac1M\sum_{i'\in\mathcal I}(\anchor(u,i')-\bar a_u)^2$. Thus $\sigma_u=\max\{v_u,\varepsilon\}$.
We study $v_u$ and the floor separately for a fixed trained model.

\subsection{The Estimator and Its Target}
For an integer $m\ge1$, draw indices $i_1,\dots,i_m$ uniformly from $\mathcal{I}$ with replacement, and we write $A_k=\anchor(u,i_k)$.

\begin{lemma}[Sampling Moments]
\label{lem:sampling-iid}
The variables $A_1,\dots,A_m$ are independent and identically distributed with moments:
\begin{equation}
\mathbb{E}[A_k]=\bar a_u,
\;
\operatorname{Var}(A_k)=v_u^{2}.
\label{eq:app-sampling-moments}
\end{equation}
\end{lemma}
\begin{proof}
Uniform sampling on $\mathcal I$ yields the empirical score distribution over the item set.
This distribution has mean $\bar a_u$ and variance $v_u^2$, giving Eq.~\eqref{eq:app-sampling-moments}.
Sampling with replacement produces independent scores with this common distribution.
\end{proof}

Define the uncorrected sample variance of the draws by:
\begin{equation}
\hat v_u^{2}=\frac1m\sum_{k=1}^{m}\big(A_k-\bar A\big)^{2},
\;
\bar A=\frac1m\sum_{k=1}^{m}A_k.
\label{eq:app-sampling-estimator}
\end{equation}
For its expectation, write $Y_k=A_k-\bar a_u$, $\bar Y=m^{-1}\sum_{k=1}^mY_k$, and $B_m=m^{-1}\sum_{k=1}^mY_k^2$.
Then $\hat v_u^2=B_m-\bar Y^2$.
The $Y_k$ are independent with mean zero and variance $v_u^2$, which gives:
\begin{equation}
\mathbb E[B_m]=v_u^2,
\;
\mathbb E[\bar Y^2]=\frac{v_u^2}{m}.
\end{equation}
Subtracting gives $\mathbb E[\hat v_u^2]=(m-1)v_u^2/m$.
The strong law gives $B_m\to v_u^2$ and $\bar Y\to0$ almost surely, so $\hat v_u^2\to v_u^2$ almost surely.
By continuity, $\hat\sigma_u=\max\{\hat v_u,\varepsilon\}$ converges almost surely to $\sigma_u$. The floor can shift the mean of $\hat\sigma_u^2$ upward.

\subsection{Size of the Error}

\begin{proposition}[Sampling Error]
\label{prop:sampling-error}
Let $\mu_4=\frac1M\sum_{i'\in\mathcal{I}}\big(\anchor(u,i')-\bar a_u\big)^{4}$ and assume $v_u>0$. Then, as the sample size $m$ tends to infinity, $\hat v_u$ then satisfies:
\begin{equation}
\frac{\sqrt{\operatorname{Var}(\hat v_u)}}{v_u}
=\sqrt{\frac{\mu_4-v_u^{4}}{4m\,v_u^{4}}}+o\big(m^{-1/2}\big).
\label{eq:app-sampling-error}
\end{equation}
When $\mu_4\approx3v_u^4$ as for a Gaussian, Eq.~\eqref{eq:app-sampling-error} gives the Gaussian reference $1/\sqrt{2m}$ for this relative error of the estimate.
\end{proposition}
\begin{proof}
We first compute the variance at finite $m$.
Expanding products of the centered variables $Y_k$ removes every term with an index that occurs only once, because independence then supplies a factor $\mathbb E[Y_k]=0$.
The remaining terms give:
\begin{equation}
\begin{aligned}
\mathbb E[B_m^2]
&=\frac{\mu_4+(m-1)v_u^4}{m},\\
\mathbb E[B_m\bar Y^2]
&=\frac{\mu_4+(m-1)v_u^4}{m^2},\\
\mathbb E[\bar Y^4]
&=\frac{\mu_4+3(m-1)v_u^4}{m^3}.
\end{aligned}
\end{equation}
Using $\hat v_u^2=B_m-\bar Y^2$ and subtracting $(\mathbb E[\hat v_u^2])^2$ yields:
\begin{equation}
\begin{aligned}
\operatorname{Var}\big(\hat v_u^2\big)
&=\frac{(m-1)^2}{m^3}\mu_4-\frac{(m-1)(m-3)}{m^3}v_u^4\\
&=\frac{\mu_4-v_u^4}{m}+O(m^{-2}).
\end{aligned}
\label{eq:app-sampling-var2}
\end{equation}

We next expand the square root to first order.
Let $\Delta_m=\hat v_u^2-v_u^2$ and define the remainder $\mathcal R_m$ by the exact identity:
\begin{equation}
\begin{aligned}
\hat v_u-v_u&=\frac{\Delta_m}{2v_u}+\mathcal R_m,\\
\mathcal R_m&=-\frac{\Delta_m^2}{2v_u(\hat v_u+v_u)^2}.
\end{aligned}
\label{eq:app-sampling-root-remainder}
\end{equation}
Since $v_u>0$ and $\hat v_u\ge0$, this gives $|\mathcal R_m|\le\Delta_m^2/(2v_u^3)$ even if $\hat v_u=0$.
To bound its second moment, set $W_k=Y_k^2-v_u^2$.
The finite item set makes $Y_k$ and $W_k$ bounded, and the $W_k$ have mean zero.
The same product expansion gives:
\begin{equation}
\begin{aligned}
\mathbb E[(B_m-v_u^2)^4]
&=\frac{\mathbb E[W_1^4]}{m^3}+\frac{3(m-1)}{m^3}(\mathbb E[W_1^2])^2\\
&=O(m^{-2}).
\end{aligned}
\end{equation}
In the expansion of $\mathbb E[\bar Y^8]$, each surviving index occurs at least twice, so at most four distinct indices remain.
There are $O(m^4)$ such terms before division by $m^8$, giving $\mathbb E[\bar Y^8]=O(m^{-4})$.
Since $\Delta_m=(B_m-v_u^2)-\bar Y^2$, the inequality $|\Delta_m|^4\le8(|B_m-v_u^2|^4+\bar Y^8)$ now gives $\mathbb E[\Delta_m^4]=O(m^{-2})$.
Thus $\mathbb E[\mathcal R_m^2]=O(m^{-2})$.
Eq.~\eqref{eq:app-sampling-var2} gives $\operatorname{Var}(\Delta_m)=O(m^{-1})$, so the Cauchy--Schwarz inequality bounds $|\operatorname{Cov}(\Delta_m,\mathcal R_m)|$ by $O(m^{-3/2})$.
Taking the variance of both sides of Eq.~\eqref{eq:app-sampling-root-remainder} therefore gives:
\begin{equation}
\begin{aligned}
\operatorname{Var}\big(\hat v_u\big)
&=\frac{1}{4v_u^2}\operatorname{Var}\big(\hat v_u^2\big)+o\!\left(\frac1m\right)\\
&=\frac{\mu_4-v_u^{4}}{4m\,v_u^{2}}+o\!\left(\frac1m\right).
\end{aligned}
\label{eq:app-sampling-var1}
\end{equation}
Dividing by $v_u^2$ and taking square roots gives Eq.~\eqref{eq:app-sampling-error}. The Gaussian reference sets $\mu_4=3v_u^4$ in that variance term.
\end{proof}

The reference is not a bound. The unfloored estimate varies more at larger kurtosis $\mu_4/v_u^4$.
The floored estimate satisfies $|\hat\sigma_u-\sigma_u|\le|\hat v_u-v_u|$, so flooring cannot increase the error. For $v_u>\varepsilon$, concentration makes $\hat v_u<\varepsilon$ exponentially unlikely, so the leading variance term holds for $\hat\sigma_u$. For $v_u\le\varepsilon$, the expansion need not describe $\hat\sigma_u$. At $v_u=0$, both $\sigma_u$ and $\hat\sigma_u$ equal $\varepsilon$, so the sampled estimate is exact for such users.

\subsection{Effect on the Results of Secs.~\ref{sec:method} and~\ref{sec:theory}}
The estimate replaces the item-independent $\sigma_u$, so its error rescales every semantic residual score of the user alike.

\begin{proposition}[Sampled Scale]
\label{prop:sampling-budget}
Write $\hat\sigma_u=(1+\zeta_u)\sigma_u$ with $\zeta_u>-1$. Hold the collaborative scores and shape terms fixed, and let $\hat r$ and $\hat R_u$ use $\hat\sigma_u$ in place of $\sigma_u$. Then for each $i$:
\begin{equation}
\hat r(u,i)=\kappa(1+\zeta_u)\,\sigma_u\,g_{ui},
\label{eq:app-sampling-budget}
\end{equation}
so $\hat r(u,i)$ equals the semantic residual score with $\sigma_u$ at budget $\kappa(1+\zeta_u)$, and $\hat R_u=(1+\zeta_u)R_u$.
Under the nonzero-variance conditions of Proposition~\ref{prop:mixing}, the mixing weight computed with $\hat\sigma_u$ also satisfies $\hat w_u=(1+\zeta_u)w_u$.
If neither budget produces a tie, a pair with $\delta_{u,ij}\ne0$ changes ranking under the sampled scale exactly when $\kappa_{u,ij}$ lies strictly between $\kappa$ and $\kappa(1+\zeta_u)$.
\end{proposition}
\begin{proof}
Substituting $\hat\sigma_u=(1+\zeta_u)\sigma_u$ into Eq.~\eqref{eq:residual} gives Eq.~\eqref{eq:app-sampling-budget} at budget $\kappa(1+\zeta_u)$.
The bound $R_u$ is proportional to $\sigma_u$.
Under the conditions of Proposition~\ref{prop:mixing}, $w_u=\kappa\chi_u\operatorname{sd}_u(g)$ has the same dependence through $\chi_u=\sigma_u/\operatorname{sd}_u(a)$.
By Eq.~\eqref{eq:flip}, the sign of $\kappa'-\kappa_{u,ij}$ determines the pair's order at budget $\kappa'$. Crossing $\kappa_{u,ij}$ changes the ranking of the two items.
\end{proof}

Without the floor, $\zeta_u=\hat v_u/v_u-1$, so Proposition~\ref{prop:sampling-error} gives the asymptotic standard deviation of $\zeta_u$.
The identity holds for a fixed model, and training with sampled scales can also change the learned shape terms and their change points.

\section{Detailed Proofs}
\label{app:proofs}

Conditions (A1) to (A5) govern the ranking proofs.
This section also states and proves the coverage condition, the pairwise ranking balance, the entry budget, and the projector's convergence that the main text summarizes.
Fix a user $u$ and a candidate set $\mathcal C_u\subseteq\mathcal I$ with $1\le K<M_u$.
The collaborative and final top-$K$ lists rank this set, whereas $\sigma_u$ uses all $M$ items in $\mathcal I$.
The stop-gradient in Eq.~\eqref{eq:score} preserves the value of $\anchor(u,i)$, so the ranking proofs use $s(u,i)=\anchor(u,i)+r(u,i)$.

\paragraph{Ranking Conditions}

\begin{enumerate}[label=(A\arabic*),leftmargin=*]
  \item \textbf{Normalized content vectors.} The normalization $\Pi(\mathbf{v})=\mathbf{v}/\max\{\|\mathbf{v}\|_2,\varepsilon\}$ used in Eqs.~\eqref{eq:basis} and \eqref{eq:query} gives $\lVert\mathbf{z}_i\rVert_2\le1$ and $\lVert\mathbf{q}_u\rVert_2\le1$ for all inputs, including zero.
  \item \textbf{Positive residual bound.} Assume $\kappa>0$ and $\tau_r>0$. The floor gives $\sigma_u\ge\varepsilon>0$, so the residual bound satisfies $R_u=\kappa\sigma_u\tanh(1/\tau_r)>0$ for every user.
  \item \textbf{Deterministic ranking.} The same fixed tie-breaking rule by item index makes each top-$K$ ranking on $\mathcal C_u$ unique.
  \item \textbf{Strict ranking changes.} The score difference changes sign. Condition (A3) breaks change-point ties. In Eq.~\eqref{eq:pq}, $P_i$ and $Q_i$ count only strict ranking changes.
  \item \textbf{Fixed model.} As $\kappa$ varies, the FedRS, projector, histories, and candidate set stay fixed, so $\gamma_{u,ij}$, $\sigma_u$, and $\delta_{u,ij}$ do too. One trained model serves every budget.
\end{enumerate}
When $\kappa=0$, Eq.~\eqref{eq:residual} gives $r(u,i)=0$ for every item. The final ranking equals the collaborative ranking, so Corollary~\ref{cor:certificate} and the reach condition hold directly. The subsequent ranking proofs therefore assume (A2) with a positive budget $\kappa>0$.

\subsection{Proof of Proposition~\ref{prop:residual-bound}}
\begin{proof}
Fix an arbitrary item $i\in\mathcal{I}$. By (A1), we have:
\begin{equation}
\lVert\mathbf{q}_u\rVert_2\le 1
\;\text{and}\;
\lVert\mathbf{z}_i\rVert_2\le 1.
\label{eq:app-vector-bound}
\end{equation}
The Cauchy--Schwarz inequality applied to the two vectors in Eq.~\eqref{eq:app-vector-bound} bounds the content similarity $c_{ui}$ of Eq.~\eqref{eq:compat}:
\begin{equation}
\left|c_{ui}\right|
\le
\lVert\mathbf{q}_u\rVert_2\lVert\mathbf{z}_i\rVert_2
\le 1.
\label{eq:app-compatibility-bound}
\end{equation}
Because $\tau_r>0$ by (A2), dividing Eq.~\eqref{eq:app-compatibility-bound} by $\tau_r$ preserves the inequality and yields $|c_{ui}/\tau_r|\le1/\tau_r$.
The hyperbolic tangent is odd and increasing, so $|\tanh(\xi)|=\tanh(|\xi|)$ for every real number~$\xi$, and Eq.~\eqref{eq:shape} therefore bounds the shape term:
\begin{equation}
\left|g_{ui}\right|
=
\tanh\!\left(\frac{|c_{ui}|}{\tau_r}\right)
\le
\tanh\!\left(\frac{1}{\tau_r}\right).
\label{eq:app-tanh-bound}
\end{equation}
Taking the absolute value of Eq.~\eqref{eq:residual} and using Eq.~\eqref{eq:app-tanh-bound} gives:
\begin{equation}
|r(u,i)|
= \kappa\sigma_u\left|g_{ui}\right|
\le \kappa\sigma_u\tanh\!\left(\frac{1}{\tau_r}\right).
\label{eq:app-residual-intermediate-bound}
\end{equation}
By (A2), $\kappa$ and $\sigma_u$ are positive. Since $1/\tau_r$ is finite and positive, $0<\tanh(1/\tau_r)<1$. Multiplying by $\kappa\sigma_u$ gives:
\begin{equation}
R_u = \kappa\sigma_u\tanh\!\left(\frac{1}{\tau_r}\right)
< \kappa\sigma_u.
\label{eq:app-envelope-bound}
\end{equation}
Eqs.~\eqref{eq:app-residual-intermediate-bound} and \eqref{eq:app-envelope-bound} prove Eq.~\eqref{eq:bound} for every item $i\in\mathcal{I}$.
\end{proof}

\subsection{Proof of Lemma~\ref{lem:flip}}
\begin{proof}
Fix a user $u$ and items $i,j$ with $\delta_{u,ij}=g_{ui}-g_{uj}\ne0$. Under (A5), $\gamma_{u,ij}$, $\sigma_u$, and $\delta_{u,ij}$ stay fixed as $\kappa$ varies.

\emph{Step 1. The score difference is affine in the budget.}
On the forward pass, Eqs.~\eqref{eq:score} and \eqref{eq:residual} give $s(u,i)=\anchor(u,i)+\kappa\sigma_u g_{ui}$. Writing the same identity for $j$ and subtracting gives:
\begin{equation}
s(u,i)-s(u,j)=\gamma_{u,ij}+\kappa\,\sigma_u\,\delta_{u,ij},
\label{eq:app-flip-affine}
\end{equation}
where $\gamma_{u,ij}$ is the collaborative score gap of Eq.~\eqref{eq:pairgap}. The right-hand side is affine in $\kappa$, with nonzero slope $\sigma_u\delta_{u,ij}$ because $\sigma_u>0$ by (A2) and $\delta_{u,ij}\ne0$ by assumption.

\emph{Step 2. Factoring out the slope.}
Dividing the constant term of Eq.~\eqref{eq:app-flip-affine} by the nonzero slope and factoring gives:
\begin{equation}
\begin{aligned}
\gamma_{u,ij}+\kappa\sigma_u\delta_{u,ij}
&=\sigma_u\delta_{u,ij}\left(\kappa+\frac{\gamma_{u,ij}}{\sigma_u\delta_{u,ij}}\right)\\
&=\sigma_u\delta_{u,ij}\big(\kappa-\kappa_{u,ij}\big).
\end{aligned}
\label{eq:app-flip-factored}
\end{equation}
The second line substitutes Eq.~\eqref{eq:flippoint}, proving Eq.~\eqref{eq:flip}.

\emph{Step 3. Order around $\kappa_{u,ij}$.}
By (A2), Eq.~\eqref{eq:app-flip-factored} gives:
\begin{equation}
\operatorname{sign}\big(s(u,i)-s(u,j)\big)
=\operatorname{sign}\big(\delta_{u,ij}\big)\cdot\operatorname{sign}\big(\kappa-\kappa_{u,ij}\big).
\label{eq:app-flip-sign}
\end{equation}
The first factor does not depend on $\kappa$. The second equals $-1$ for $\kappa<\kappa_{u,ij}$, equals $0$ at $\kappa=\kappa_{u,ij}$, and equals $+1$ for $\kappa>\kappa_{u,ij}$. The pair therefore keeps one order for all $\kappa<\kappa_{u,ij}$ and the opposite order for all $\kappa>\kappa_{u,ij}$, and its two final scores are equal only at the change point $\kappa=\kappa_{u,ij}$ of Eq.~\eqref{eq:flippoint}.

\emph{Step 4. Ranking change at a positive budget.}
Admissible budgets satisfy $\kappa\ge0$, so both sides of $\kappa_{u,ij}$ are reachable only when $\kappa_{u,ij}>0$. By Eq.~\eqref{eq:flippoint}, $\kappa_{u,ij}>0$ holds exactly when $\gamma_{u,ij}$ and $\sigma_u\delta_{u,ij}$ carry opposite signs, and since $\sigma_u>0$ this is the condition $\gamma_{u,ij}\delta_{u,ij}<0$. Canceling the sign gives $\kappa_{u,ij}=|\gamma_{u,ij}|/(\sigma_u|\delta_{u,ij}|)$. When instead $\gamma_{u,ij}\delta_{u,ij}>0$, the change point is negative and no budget reaches it, so the pair keeps its collaborative order for every budget $\kappa\ge0$.
\end{proof}

\noindent\emph{Remark on degenerate pairs.}
If $\delta_{u,ij}=0$, then Eq.~\eqref{eq:app-flip-affine} is the constant $\gamma_{u,ij}$ and the pair never changes ranking.
If $\gamma_{u,ij}=0$ while $\delta_{u,ij}\ne0$, the factorization still holds with $\kappa_{u,ij}=0$. The final scores then differ for every $\kappa>0$ and can change the order that tie-breaking gave the pair, which is not a strict ranking change.
Corollary~\ref{cor:certificate} covers $\delta_{u,ij}=0$ under its gap condition, and its premise $|\gamma_{u,ij}|>2R_u>0$ rules out ties.

\subsection{Proof of Corollary~\ref{cor:certificate}}
\begin{proof}
Fix a user $u$ and two items $i$ and $j$ satisfying:
\begin{equation}
|\gamma_{u,ij}|> 2R_u.
\label{eq:app-anchor-gap}
\end{equation}
Define the residual difference of the two items as follows:
\begin{equation}
\Delta_r=r(u,i)-r(u,j).
\label{eq:app-pairwise-differences}
\end{equation}
The difference between their two final scores then satisfies:
\begin{equation}
s(u,i)-s(u,j)=\gamma_{u,ij}+\Delta_r.
\label{eq:app-final-gap}
\end{equation}
By the triangle inequality and Eq.~\eqref{eq:bound}, it then satisfies:
\begin{equation}
\begin{aligned}
|\Delta_r|
&=|r(u,i)-r(u,j)| \\
&\le |r(u,i)|+|r(u,j)| \\
&\le R_u+R_u=2R_u.
\end{aligned}
\label{eq:app-residual-difference-bound}
\end{equation}
With the collaborative score gap in Eq.~\eqref{eq:app-anchor-gap}, this yields:
\begin{equation}
|\Delta_r|\le 2R_u< |\gamma_{u,ij}|.
\label{eq:app-difference-comparison}
\end{equation}
Since $R_u>0$, Eq.~\eqref{eq:app-anchor-gap} gives $\gamma_{u,ij}\ne0$. The sign of $\gamma_{u,ij}$ determines which of the two items initially ranks higher.

If $\gamma_{u,ij}>0$, then $\Delta_r\ge-|\Delta_r|>-\gamma_{u,ij}$ by Eq.~\eqref{eq:app-difference-comparison}. Adding $\gamma_{u,ij}$ to both sides gives $\gamma_{u,ij}+\Delta_r>0$.
If $\gamma_{u,ij}<0$, then $\Delta_r\le|\Delta_r|<|\gamma_{u,ij}|=-\gamma_{u,ij}$. Adding $\gamma_{u,ij}$ to both sides gives $\gamma_{u,ij}+\Delta_r<0$.
In both cases, Eq.~\eqref{eq:app-final-gap} has the sign of $\gamma_{u,ij}$, so $i$ and $j$ keep their collaborative order.
\end{proof}

\noindent\emph{Remark on the residual bound.}
The proof uses the bound $|r(u,i)|\le R_u$ of Eq.~\eqref{eq:bound}, which is attained at $c_{ui}=\pm1$.
Learned representations may never reach $c_{ui}=\pm1$, so realized residual differences can be smaller than $2R_u$.
The corollary holds for every semantic residual score within the bound, and because the bound is attained, $2R_u$ is the smallest threshold that ensures order preservation for any possible content.

\subsection{Proof of Theorem~\ref{thm:reach}}
\begin{proof}
Let $\mathcal T_u\subseteq\mathcal C_u$ denote the deterministic collaborative top-$K$ set. Thus $|\mathcal T_u|=K$, and every item $i\in\mathcal T_u$ satisfies:
\begin{equation}
\anchor(u,i)\ge\anchor^{(K)}_{u}.
\label{eq:app-topk-anchor-bound}
\end{equation}
Let $\tilde{\mathcal T}_u\subseteq\mathcal C_u$ be the deterministic final top-$K$ set. Since $j\in\tilde{\mathcal T}_u\setminus\mathcal T_u$ and both sets have $K$ items, some $i\in\mathcal T_u\setminus\tilde{\mathcal T}_u$ exists.

The final ranking places $j$ inside its top-$K$ set and $i$ outside it, so the final score of $j$ is at least the final score of $i$:
\begin{equation}
s(u,j)\ge s(u,i).
\label{eq:app-promotion-order}
\end{equation}
The inequality allows equality because the deterministic tie-breaking rule in (A3) can decide the order of equal scores.

Applying Eq.~\eqref{eq:bound} to $j$ gives $r(u,j)\le R_u$, hence:
\begin{equation}
s(u,j)=\anchor(u,j)+r(u,j)\le\anchor(u,j)+R_u.
\label{eq:app-target-upper-bound}
\end{equation}
Applying Eq.~\eqref{eq:bound} to the displaced item gives $r(u,i)\ge-R_u$. Combining this inequality with Eq.~\eqref{eq:app-topk-anchor-bound} yields the bound:
\begin{equation}
\begin{aligned}
s(u,i)
&=\anchor(u,i)+r(u,i) \\
&\ge\anchor(u,i)-R_u \\
&\ge\anchor^{(K)}_{u}-R_u.
\end{aligned}
\label{eq:app-displaced-lower-bound}
\end{equation}
Eqs.~\eqref{eq:app-target-upper-bound}, \eqref{eq:app-promotion-order}, and \eqref{eq:app-displaced-lower-bound} form a chain of inequalities:
\begin{equation}
\anchor(u,j)+R_u
\ge s(u,j)
\ge s(u,i)
\ge \anchor^{(K)}_{u}-R_u.
\label{eq:app-promotion-chain}
\end{equation}
The outer terms of the chain in Eq.~\eqref{eq:app-promotion-chain} therefore satisfy:
\begin{equation}
\anchor(u,j)+R_u\ge\anchor^{(K)}_{u}-R_u.
\label{eq:app-anchor-necessary-condition}
\end{equation}
Rearranging Eq.~\eqref{eq:app-anchor-necessary-condition} gives $\gamma^{(K)}_{u,j}=\anchor^{(K)}_{u}-\anchor(u,j)\le2R_u$ by the definition in Eq.~\eqref{eq:gap}, which then proves Theorem~\ref{thm:reach}.

By contraposition, no semantic residual score satisfying Eq.~\eqref{eq:bound} lets $j$ enter the top-$K$ list when $\gamma^{(K)}_{u,j}>2R_u$.
\end{proof}

\noindent\emph{Remark on necessity versus sufficiency.}
The reach condition is necessary for entry.
Entry also requires $j$ to outrank enough competitors by final score, and these scores depend on the content similarities in Eq.~\eqref{eq:compat}.
Proposition~\ref{prop:rank-path} states this requirement exactly through the rank counts $P_j$ and $Q_j$.

\subsection{Proof of Corollary~\ref{cor:rank-range}}
\begin{proof}
Fix a user $u$, an item $j\in\mathcal C_u$, and any content.

\emph{Step 1. The residual bound holds for any content.}
The map $\Pi$ satisfies (A1) for every input, whatever the item features, projector parameters, and history.
The proof of Proposition~\ref{prop:residual-bound} uses only (A1) and (A2), so $|r(u,i)|\le R_u$ for every $i\in\mathcal C_u$.

\emph{Step 2. Distant candidates keep their side of $j$.}
Define the sets of candidates far above and far below $j$ as follows:
\begin{equation}
\begin{aligned}
\mathcal N^{+}_{u,j}&=\{i\in\mathcal C_u:\anchor(u,i)>\anchor(u,j)+2R_u\},\\
\mathcal N^{-}_{u,j}&=\{i\in\mathcal C_u:\anchor(u,i)<\anchor(u,j)-2R_u\},
\end{aligned}
\label{eq:app-range-sets}
\end{equation}
so that $n^{\pm}_{u,j}=|\mathcal N^{\pm}_{u,j}|$.
Every $i\in\mathcal N^{+}_{u,j}$ satisfies $\gamma_{u,ij}>2R_u$, and the proof of Corollary~\ref{cor:certificate} gives $s(u,i)>s(u,j)$.
Every $i\in\mathcal N^{-}_{u,j}$ satisfies $\gamma_{u,ji}>2R_u$, which gives $s(u,i)<s(u,j)$.
Both inequalities are strict, so the tie-breaking rule in (A3) leaves the order of these pairs unchanged at every budget.

\emph{Step 3. Counting positions.}
The final ranking places all of $\mathcal N^{+}_{u,j}$ above $j$ and all of $\mathcal N^{-}_{u,j}$ below it, so at least $n^{+}_{u,j}$ and at most $M_u-1-n^{-}_{u,j}$ items precede $j$, which gives Eq.~\eqref{eq:rank-range}.
\end{proof}

\noindent\emph{Remark on Theorem~\ref{thm:reach}.}
At least $K$ candidates exceed \mbox{$\anchor(u,j)+2R_u$} exactly when the collaborative top-$K$ boundary score $\anchor^{(K)}_u$ does.
Hence $n^{+}_{u,j}<K$ holds exactly when $\anchor^{(K)}_u-\anchor(u,j)\le2R_u$, which for an outside item $j$ with $\tilde\rho_u(j)\le K$ is Eq.~\eqref{eq:reachcond}.
Theorem~\ref{thm:reach} is therefore the top-$K$ case of Eq.~\eqref{eq:rank-range}.

\noindent\emph{Remark on accuracy.}
Each user has one test item, so accuracy changes follow the rank range of that item.
Let $j$ be the test item of user $u$.
Candidates above $\anchor(u,j)+2R_u$ also precede $j$ in the collaborative ranking and candidates below $\anchor(u,j)-2R_u$ follow it, so the collaborative rank $\rho_u(j)$ lies in the range of Eq.~\eqref{eq:rank-range} as well.
Recall@$K$ records whether the rank is at most $K$, so it can change only if the range contains both $K$ and $K+1$.
By the previous remark, the lower end of the range is at most $K$ exactly when $\anchor^{(K)}_u-2R_u\le\anchor(u,j)$.
Symmetrically, at least $M_u-K$ candidates fall below $\anchor(u,j)-2R_u$ exactly when the collaborative score $\anchor^{(K+1)}_u$ at rank $K+1$ does. The upper end is thus at least $K+1$ exactly when $\anchor(u,j)\le\anchor^{(K+1)}_u+2R_u$.
NDCG@$K$ is a function of $\tilde\rho_u(j)$ that vanishes beyond rank $K$, so it can change only if the range holds several ranks and its lower end is at most $K$.
For any content, the semantic residual scores can therefore change the accuracy of at most the share of users whose test items meet these conditions.

\noindent\emph{Remark on untrusted content.}
The range needs only (A1), (A2), and collaborative scores that content cannot change.
A client that receives content representations it cannot verify can reject nonfinite values and apply $\Pi$ to the rest before scoring, which keeps Eq.~\eqref{eq:rank-range}.
For finite representations without this normalization, $\tanh$ still gives $|g_{ui}|<1$, so Eq.~\eqref{eq:rank-range} holds with the larger bound $\kappa\sigma_u$ replacing the residual bound $R_u$.

\subsection{Coverage Condition}
Corollary~\ref{cor:certificate} and the reach band of Theorem~\ref{thm:reach} become uninformative once the residual budget is large enough.

\begin{proposition}[Coverage Condition]
\label{prop:nonvacuous}
For user $u$ and items $i,j\in\mathcal{I}$, the collaborative score gap obeys $|\gamma_{u,ij}|\le\sigma_u\sqrt{2M}$.
If the hyperparameters satisfy the following condition:
\begin{equation}
2\kappa\tanh\!\left(\frac{1}{\tau_r}\right)\;\ge\;\sqrt{2M},
\label{eq:vacuous}
\end{equation}
then Corollary~\ref{cor:certificate} applies to no pair and the reach band includes every item outside the collaborative top-$K$ list.
\end{proposition}
\begin{proof}
We first bound the collaborative score gap.
Let $\bar a_u=\frac1M\sum_{i'\in\mathcal{I}}\anchor(u,i')$ be the mean over the item set. The scale $\sigma_u$ is at least the population standard deviation by its definition in Sec.~\ref{subsec:residual}.
Fix $i\ne j$ and write $\tilde a_i=\anchor(u,i)-\bar a_u$ and $\tilde a_j=\anchor(u,j)-\bar a_u$.
The sum of squared deviations over $\mathcal I$ contains the distinct terms $\tilde a_i^2$ and $\tilde a_j^2$, so the scale satisfies:
\begin{equation}
M\sigma_u^{2}\ge\sum_{i'\in\mathcal{I}}\big(\anchor(u,i')-\bar a_u\big)^2\ge \tilde a_i^{2}+\tilde a_j^{2}.
\label{eq:app-gapbound-sum}
\end{equation}
Since $(\tilde a_i+\tilde a_j)^2\ge0$, the parallelogram identity gives:
\begin{equation}
(\tilde a_i-\tilde a_j)^{2}\le(\tilde a_i-\tilde a_j)^{2}+(\tilde a_i+\tilde a_j)^{2}=2\big(\tilde a_i^{2}+\tilde a_j^{2}\big)\le2M\sigma_u^{2}.
\label{eq:app-gapbound-parallelogram}
\end{equation}
Taking square roots in Eq.~\eqref{eq:app-gapbound-parallelogram} yields $|\gamma_{u,ij}|=|\tilde a_i-\tilde a_j|\le\sigma_u\sqrt{2M}$. For $i=j$, the zero gap also satisfies this bound.

The constant is tight. For $M\ge2$ and any $b>\varepsilon$, set $\anchor(u,i)=\bar a_u+b\sqrt{M/2}$, $\anchor(u,j)=\bar a_u-b\sqrt{M/2}$, and all remaining scores to $\bar a_u$. The gap is $\sigma_u\sqrt{2M}$ since $\sigma_u=b$.

Now assume Eq.~\eqref{eq:vacuous}. Multiplying it by $\sigma_u>0$ gives:
\begin{equation}
2R_u=2\kappa\sigma_u\tanh\!\left(\frac{1}{\tau_r}\right)\;\ge\;\sigma_u\sqrt{2M}\;\ge\;|\gamma_{u,ij}|
\label{eq:app-vacuous-chain}
\end{equation}
for every pair $i,j$. The premise $|\gamma_{u,ij}|>2R_u$ of Corollary~\ref{cor:certificate} therefore fails for every pair, so both Corollary~\ref{cor:certificate} and the reach condition of Theorem~\ref{thm:reach} become uninformative for user $u$.
\end{proof}

\noindent\emph{Remark on the scope of the coverage condition.}
Eq.~\eqref{eq:app-gapbound-sum} uses the item-set size $M$ and $\sigma_u$, so it bounds every pair in $\mathcal I$, including the pairs in $\mathcal C_u$ that define the reach condition.
Replacing $M$ with the candidate-set size $M_u$ generally gives a false inequality.
The coverage condition is sufficient but not necessary for this. When Eq.~\eqref{eq:vacuous} fails, $2R_u$ lies below the largest possible gap $\sigma_u\sqrt{2M}$, and the collaborative scores decide how many pairs keep their order.
With a sampled scale $\hat\sigma_u$, Proposition~\ref{prop:sampling-budget} replaces $\kappa$ by $\kappa(1+\zeta_u)$ in Eq.~\eqref{eq:vacuous}.

\subsection{Pairwise Ranking Balance}
\begin{proposition}[Pairwise Ranking Balance]
\label{prop:zerosum}
For a fixed user $u$, trained model, and budget $\kappa$, $P_i$ and $Q_i$ have equal totals over the candidate set $\mathcal{C}_u$.
Their differences therefore satisfy:
\begin{equation}
  \sum_{i\in\mathcal{C}_u}\big[P_i-Q_i\big]=0.
  \label{eq:zerosum}
\end{equation}
\end{proposition}
\begin{proof}
Fix a user $u$, a trained model, and a budget $\kappa$, and consider the ordered pairs $(i,j)$ with $i,j\in\mathcal{C}_u$ and $i\ne j$. Define the indicators of ranking changes favoring and opposing $i$ by:
\begin{equation}
\begin{aligned}
f(i,j)&=\mathbf 1\{\anchor(u,j)>\anchor(u,i),\ s(u,j)<s(u,i)\},\\
f'(i,j)&=\mathbf 1\{\anchor(u,j)<\anchor(u,i),\ s(u,j)>s(u,i)\}.
\end{aligned}
\label{eq:app-zerosum-indicators}
\end{equation}
By Eq.~\eqref{eq:pq}, $P_i=\sum_{j\ne i}f(i,j)$ and $Q_i=\sum_{j\ne i}f'(i,j)$. Summing both counts over all target items $i\in\mathcal{C}_u$ gives:
\begin{equation}
\sum_{i\in\mathcal{C}_u}P_i=\sum_{i}\sum_{j\ne i}f(i,j),
\;
\sum_{i\in\mathcal{C}_u}Q_i=\sum_{i}\sum_{j\ne i}f'(i,j).
\label{eq:app-zerosum-sums}
\end{equation}
Swapping $i$ and $j$ in the definition of $f'$ in Eq.~\eqref{eq:app-zerosum-indicators} shows $f(i,j)=f'(j,i)$, because the conditions ``$\anchor(u,j)>\anchor(u,i)$ and $s(u,j)<s(u,i)$'' are ``$\anchor(u,i)<\anchor(u,j)$ and $s(u,i)>s(u,j)$'' with each inequality written from the other side.
The swap $(i,j)\mapsto(j,i)$ is a bijection on the set of ordered pairs with $i\ne j$, so summing this identity over that set matches the two totals in Eq.~\eqref{eq:app-zerosum-sums} term by term and proves Eq.~\eqref{eq:zerosum}.
\end{proof}

\noindent\emph{Remark on ties and the scope of the identity.}
Both indicators in Eq.~\eqref{eq:app-zerosum-indicators} use strict inequalities, so $f$ and $f'$ both exclude a tied pair and the bijection is unaffected.
The identity requires neither tie-breaking nor distinct scores.
With ties, the rank of $i$ can also change through pairs that tie in either score, while $P_i-Q_i$ counts only strict ranking changes.
Corollary~\ref{cor:certificate} uses no relevance labels. It bounds $P_i+Q_i$, and hence $|P_i-Q_i|$, by the number of pairs involving $i$ whose collaborative score gap is at most $2R_u$. Proposition~\ref{prop:zerosum} adds that the net changes $P_i-Q_i$ sum to zero over all targets.
If a target's collaborative score is distinct, Proposition~\ref{prop:rank-path} equates $P_i-Q_i$ with its rank improvement at every budget except its change points.

\subsection{Proof of Proposition~\ref{prop:rank-path}}
\begin{proof}
Fix a user $u$, a trained model under (A5), a target $i\in\mathcal C_u$ whose collaborative score differs from all others, and a budget $\kappa>0$ that avoids every positive change point of $i$.

\emph{Step 1. Every pair with $i$ is strictly ordered.}
If $\delta_{u,ij}=0$, Eq.~\eqref{eq:app-flip-affine} gives $s(u,i)-s(u,j)=\gamma_{u,ij}\ne0$.
If $\delta_{u,ij}\ne0$, Eq.~\eqref{eq:app-flip-factored} gives a nonzero difference because $\kappa\ne\kappa_{u,ij}$.
Both rankings therefore order each pair involving $i$ strictly, and Eq.~\eqref{eq:pq} counts every ranking change of these pairs.

\emph{Step 2. The rank identity.}
The final rank of $i$ is one plus the number of competitors scoring above it.
Of the \mbox{$\rho_u(i)-1$} competitors with higher collaborative scores, $P_i$ fall below $i$ and the rest stay above it. By Eq.~\eqref{eq:pq}, exactly $Q_i$ lower competitors rise above $i$, which gives the identity in Eq.~\eqref{eq:rank-path}.

\emph{Step 3. Change points determine the counts.}
Take a competitor $j$ with $\anchor(u,j)>\anchor(u,i)$, so that $\gamma_{u,ij}<0$.
By Step~4 of the proof of Lemma~\ref{lem:flip} and the remark after it, this pair changes ranking at some $\kappa>0$ only if $\delta_{u,ij}>0$, and then $\kappa_{u,ij}>0$.
Eq.~\eqref{eq:app-flip-sign} shows that $j$ lies below $i$ exactly when $\kappa>\kappa_{u,ij}$.
Therefore, $P_i$ counts the competitors with higher collaborative scores whose change points with $i$ lie in $(0,\kappa)$.
Symmetrically, for $\gamma_{u,ij}>0$ and $\delta_{u,ij}<0$, $Q_i$ counts the competitors with lower collaborative scores whose change points lie in $(0,\kappa)$.

\emph{Step 4. The large-budget limit.}
Let $\kappa$ exceed every positive change point of $i$.
Then $P_i$ counts the higher competitors with $g_{uj}<g_{ui}$, and $Q_i$ counts the lower competitors with $g_{uj}>g_{ui}$.
Since no competitor ties the target $i$ in collaborative score, substituting these counts into Eq.~\eqref{eq:rank-path} gives its final rank:
\begin{equation}
\begin{aligned}
\tilde\rho_u(i)-1
&=\big|\{j:g_{uj}>g_{ui}\}\big|\\
&\quad+\big|\{j:g_{uj}=g_{ui},\ \anchor(u,j)>\anchor(u,i)\}\big|.
\end{aligned}
\label{eq:app-rank-limit}
\end{equation}
Because $\tanh$ is strictly increasing and $\tau_r>0$, $g_{uj}>g_{ui}$ holds exactly when $c_{uj}>c_{ui}$.
Eq.~\eqref{eq:app-rank-limit} is therefore the rank of $i$ by similarity, with ties broken by collaborative order.
\end{proof}

\subsection{Entry Budget}
\begin{corollary}[Entry Budget]
\label{cor:entry-budget}
Under the conditions of Proposition~\ref{prop:rank-path}, an item $j$ with $\rho_u(j)>K$ lies in the final top-$K$ list exactly when $P_j-Q_j\ge\rho_u(j)-K$.
We call this condition the entry rule. If it holds at some budget, let $\kappa^{\mathrm{in}}_{u,j}$ be the infimum of such budgets. This entry budget is the smallest positive change point of $j$ with the rule holding just above it, and the collaborative boundary gap of $j$ bounds it from below:
\begin{equation}
\kappa^{\mathrm{in}}_{u,j}\ge\frac{\gamma^{(K)}_{u,j}}{2\sigma_u\tanh(1/\tau_r)}.
\label{eq:entry-budget}
\end{equation}
\end{corollary}
\begin{proof}
Fix an item $j$ with $\rho_u(j)>K$ as in Proposition~\ref{prop:rank-path}.

\emph{Step 1. The entry rule.}
By Eq.~\eqref{eq:rank-path}, the final rank of $j$ satisfies $\tilde\rho_u(j)\le K$ exactly when $P_j-Q_j\ge\rho_u(j)-K$.

\emph{Step 2. The entry budget is a change point.}
By Proposition~\ref{prop:rank-path}, $P_j$ and $Q_j$ stay constant between consecutive positive change points of $j$ and above the largest one, so the membership of $j$ in the final top-$K$ list is constant there too.
Below the smallest positive change point both counts vanish and $\tilde\rho_u(j)=\rho_u(j)>K$.
The budgets other than change points that place $j$ in the list therefore form a union of these open intervals. Their infimum is the left end of the first one, which is the smallest positive change point of $j$ with the entry rule holding just above it.

\emph{Step 3. The residual bound.}
At every budget $\kappa$ that places $j$ in the list, Theorem~\ref{thm:reach} gives $\gamma^{(K)}_{u,j}\le2\kappa\sigma_u\tanh(1/\tau_r)$.
Dividing by $2\sigma_u\tanh(1/\tau_r)>0$ and taking the infimum over these budgets proves the lower bound stated in Eq.~\eqref{eq:entry-budget}.
\end{proof}

\noindent\emph{Remark on a sharper bound.}
The entry rule requires \mbox{$P_j\ge\rho_u(j)-K$}, so $j$ must pass at least $\rho_u(j)-K$ competitors with higher collaborative scores.
Passing such a competitor $i$ requires $\delta_{u,ji}>0$ by Step~4 of the proof of Lemma~\ref{lem:flip}. Eq.~\eqref{eq:app-tanh-bound} gives $\delta_{u,ji}\le2\tanh(1/\tau_r)$, so the change point satisfies $\kappa_{u,ij}=-\gamma_{u,ij}/(\sigma_u\delta_{u,ij})=\gamma_{u,ij}/(\sigma_u\delta_{u,ji})\ge\gamma_{u,ij}/(2\sigma_u\tanh(1/\tau_r))$.
The entry budget is thus at least the $(\rho_u(j)-K)$th smallest positive change point of $j$ with a higher competitor.
The competitors that $j$ can pass form a subset of the higher competitors, so their $(\rho_u(j)-K)$th smallest gap is at least $\gamma^{(K)}_{u,j}$, its value over all higher competitors. This sharper bound thus implies Eq.~\eqref{eq:entry-budget} and is at least as tight.

\noindent\emph{Remark on leaving the list.}
Change points with lower competitors raise $Q_j$, so $j$ can leave the list after entering it.
Take three candidates with collaborative scores $1$, $0.8$, and $0$, and shape terms $0$, $0.2\tanh(1/\tau_r)$, and $0.5\tanh(1/\tau_r)$.
For $K=1$, the second candidate leads between its change points with the first and the third candidates, and the third leads beyond.

\subsection{Convergence of the Projector Updates}
\label{app:convergence-proof}
We analyze the server's SGD on the projector with fixed collaborative parameters, without momentum or clipping, and condition on the initial projector $\phi^0$.
The client sampling distribution is fixed and independent of the past iterates.
For each nonempty participating set, let $\alpha_{u,t}\ge0$ sum to one over its members and depend only on that set and fixed client data counts.
Set $\alpha_{u,t}=0$ outside $\mathcal{S}_t$.
Then $\bar{\alpha}_u=\mathbb{E}[\alpha_{u,t}]$ is independent of $t$, $\sum_u\bar{\alpha}_u=1$, and the objective $\mathcal{J}_{\bar{\alpha}}$ satisfies:
\begin{equation}
\mathcal{J}_{\bar{\alpha}}(\phi)
=\lambda\sum_u \bar{\alpha}_u\mathcal{L}_u(s)
=\mathbb{E}_{\mathcal{S}_t}
 \left[\lambda\sum_{u\in\mathcal{S}_t}\alpha_{u,t}\mathcal{L}_u(s)\right].
\label{eq:sera-population-objective}
\end{equation}
These weights $\bar{\alpha}_u$ can differ from the weights $\alpha_u$ of Eq.~\eqref{eq:objective_sera}.
For example, suppose each round selects $m_c$ clients uniformly without replacement and weights them by their data counts $|\mathcal H_u|$.
The aggregation weight of client $u$ in each round is:
\begin{equation}
\alpha_{u,t}
=\frac{\mathbf 1\{u\in\mathcal S_t\}|\mathcal H_u|}{\sum_{u'\in\mathcal S_t}|\mathcal H_{u'}|}.
\label{eq:app-sampling-weight}
\end{equation}
With equal data counts, each selected client receives weight $1/m_c$ and is selected with probability $m_c/N$.
Consequently, $\bar\alpha_u=(m_c/N)(1/m_c)=1/N$.
With unequal data counts, the normalization depends on the sampled set.
Selecting one client per round still gives $\bar\alpha_u=1/N$, whereas full participation gives $\bar\alpha_u=|\mathcal H_u|/\sum_{u'}|\mathcal H_{u'}|$.
Participation can thus change the optimized objective even with unbiased local gradients.

\begin{assumption}[Conditional Regularity]
\label{ass:sera-convergence}
The fixed objective $\mathcal{J}_{\bar{\alpha}}$ is bounded below by $\mathcal{J}_{\bar{\alpha},\inf}$ and has an $L$-Lipschitz gradient.
Given the sampled clients, the sampled gradients of $\lambda\mathcal{L}_u(s)$ are unbiased with finite second moments.
\end{assumption}
Smoothness must be assumed, because the derivative of $\Pi$ in Eqs.~\eqref{eq:basis} and~\eqref{eq:query} can jump where its input has norm $\varepsilon$.

Let $\mathbf g_{u,t,e}$ be client $u$'s sampled gradient of $\lambda\mathcal L_u(s)$ with respect to $\phi$ at local step $e$ of round $t$.
All local steps read the table $\mathbf Z^t$, so their gradients are taken at $\phi^t$. By the chain rule, the server update in Algorithm~\ref{alg:fedsera} moves $\phi$ along the direction:
\begin{equation}
\widehat{\mathbf g}_t
=\sum_{u\in\mathcal S_t}\alpha_{u,t}\frac1E\sum_{e=1}^{E}\mathbf g_{u,t,e}.
\label{eq:agg-direction}
\end{equation}
Let $\mathbb E_t$ condition on all randomness before round $t$, and let $V_t=\mathbb E_t\|\widehat{\mathbf g}_t-\mathbb E_t[\widehat{\mathbf g}_t]\|_2^2$ be the variance of this direction.

\begin{theorem}[Stationarity Bound]
\label{thm:sera-convergence}
Under Assumption~\ref{ass:sera-convergence}, let $\bar\eta=\eta_{\mathrm s}E>0$ with $L\bar\eta\le1$ and $\phi^{t+1}=\phi^t-\bar\eta\widehat{\mathbf g}_t$.
Then:
\begin{equation}
\begin{aligned}
&\frac1T\sum_{t=0}^{T-1}
 \mathbb E\|\nabla\mathcal J_{\bar{\alpha}}(\phi^t)\|_2^2\\
\le{}&
\frac{2[\mathcal J_{\bar{\alpha}}(\phi^0)-\mathcal J_{\bar{\alpha},\inf}]}{\bar\eta T}
+\frac{L\bar\eta}{T}\sum_{t=0}^{T-1}\mathbb E V_t.
\end{aligned}
\label{eq:convergence-bound}
\end{equation}
\end{theorem}
\begin{proof}
Every local step of client $u$ reads the same table $\mathbf Z^t$ computed at $\phi^t$.
By the chain rule through $h_\phi$, the gradient $\mathbf G_u^t$ that client $u$ uploads in Algorithm~\ref{alg:fedsera} then satisfies:
\begin{equation}
\sum_{i\in\mathcal I}\Big(\frac{\partial\mathbf z_i}{\partial\phi}\Big|_{\phi^t}\Big)^{\!\top}[\mathbf G_u^t]_i
=\sum_{e=1}^{E}\mathbf g_{u,t,e},
\label{eq:app-table-chain}
\end{equation}
where $[\mathbf G_u^t]_i$ is row $i$ of $\mathbf G_u^t$ and $\mathbf g_{u,t,e}$ is the sampled gradient of $\lambda\mathcal L_u(s)$ at $\phi^t$ for the pairs of local step $e$.
Back-propagation is linear, so the server update of Algorithm~\ref{alg:fedsera} takes the form:
\begin{equation}
\phi^{t+1}
=\phi^t-\eta_{\mathrm s}\sum_{u\in\mathcal S_t}\alpha_{u,t}\sum_{e=1}^{E}\mathbf g_{u,t,e}
=\phi^t-\bar\eta\,\widehat{\mathbf g}_t,
\label{eq:app-global-projector-update}
\end{equation}
with $\bar\eta=\eta_{\mathrm s}E$ and the aggregated direction $\widehat{\mathbf g}_t$ from Eq.~\eqref{eq:agg-direction}.
Write $\bar{\mathbf g}_t=\nabla\mathcal{J}_{\bar{\alpha}}(\phi^t)$.
Given the sampled clients $\mathcal S_t$, each $\mathbf g_{u,t,e}$ is unbiased for the gradient of $\lambda\mathcal L_u(s)$ at $\phi^t$ by Assumption~\ref{ass:sera-convergence}, and $\alpha_{u,t}$ depends only on $\mathcal S_t$, so $\widehat{\mathbf g}_t$ is unbiased:
\begin{equation}
\mathbb E_t[\widehat{\mathbf g}_t]
=\lambda\sum_u\bar\alpha_u\nabla_\phi\mathcal L_u(s)\big|_{\phi=\phi^t}
=\bar{\mathbf g}_t.
\label{eq:app-unbiased-direction}
\end{equation}
With $V_t$ as in Theorem~\ref{thm:sera-convergence}, this gives $\mathbb{E}_t[\|\widehat{\mathbf{g}}_t\|_2^2]=\|\bar{\mathbf g}_t\|_2^2+V_t$.

By the $L$-smoothness of $\mathcal{J}_{\bar{\alpha}}$, the descent lemma gives:
\begin{equation}
\mathcal{J}_{\bar{\alpha}}(\phi^{t+1})
\le{}\mathcal{J}_{\bar{\alpha}}(\phi^t)
-\bar{\eta}\langle\bar{\mathbf g}_t,\widehat{\mathbf{g}}_t\rangle
+\frac{L\bar{\eta}^2}{2}\|\widehat{\mathbf{g}}_t\|_2^2.
\label{eq:app-descent-lemma}
\end{equation}
The conditional expectation of Eq.~\eqref{eq:app-descent-lemma} with Eq.~\eqref{eq:app-unbiased-direction} gives:
\begin{equation}
\mathbb{E}_t[\mathcal{J}_{\bar{\alpha}}(\phi^{t+1})]
\le{}\mathcal{J}_{\bar{\alpha}}(\phi^t)
-\Big(\bar{\eta}-\frac{L\bar{\eta}^2}{2}\Big)\|\bar{\mathbf g}_t\|_2^2
+\frac{L\bar{\eta}^2}{2}V_t.
\label{eq:app-expected-descent}
\end{equation}
Under $\bar{\eta}>0$ and $L\bar{\eta}\le1$, the coefficient of $\|\bar{\mathbf g}_t\|_2^2$ is at least $\bar{\eta}/2$, so rearranging and dividing by $\bar{\eta}/2$ gives the bound:
\begin{equation}
\|\bar{\mathbf g}_t\|_2^2
\le{}\frac{2}{\bar{\eta}}
 \left(\mathcal{J}_{\bar{\alpha}}(\phi^t)
 -\mathbb{E}_t[\mathcal{J}_{\bar{\alpha}}(\phi^{t+1})]\right)
 +L\bar{\eta}V_t.
\label{eq:app-gradient-bound}
\end{equation}
Taking total expectation and summing over rounds gives:
\begin{equation}
\begin{aligned}
\sum_{t=0}^{T-1}\mathbb{E}\|\bar{\mathbf g}_t\|_2^2
\le{}&\frac{2}{\bar{\eta}}
 \left(\mathcal{J}_{\bar{\alpha}}(\phi^0)-\mathbb{E}\mathcal{J}_{\bar{\alpha}}(\phi^T)\right)
\\
&+L\bar{\eta}\sum_{t=0}^{T-1}\mathbb{E}V_t.
\end{aligned}
\label{eq:app-telescoped-bound}
\end{equation}
The lower bound $\mathcal{J}_{\bar{\alpha}}\ge\mathcal{J}_{\bar{\alpha},\inf}$ of Assumption~\ref{ass:sera-convergence} gives:
\begin{equation}
\mathcal{J}_{\bar{\alpha}}(\phi^0)-\mathbb{E}\mathcal{J}_{\bar{\alpha}}(\phi^T)
\le\mathcal{J}_{\bar{\alpha}}(\phi^0)-\mathcal{J}_{\bar{\alpha},\inf}.
\label{eq:app-objective-lower-bound}
\end{equation}
Dividing both sides by $T$ proves the bound in Eq.~\eqref{eq:convergence-bound}.
\end{proof}

\paragraph{Convergence Rate}
The bound separates the initial objective gap from the variance of stochastic sampling.
To obtain an explicit rate, choose the constant step size $\bar\eta=T^{-1/2}$ for a run of $T\ge L^2$ rounds, which meets $L\bar\eta\le1$, and assume $\mathbb E V_t\le V$ for a constant $V<\infty$.
Eq.~\eqref{eq:convergence-bound} then gives:
\begin{equation}
\frac1T\sum_{t=0}^{T-1}\mathbb E\|\bar{\mathbf g}_t\|_2^2
\le\frac{2[\mathcal J_{\bar\alpha}(\phi^0)-\mathcal J_{\bar\alpha,\inf}]+LV}{\sqrt T}.
\label{eq:app-convergence-rate}
\end{equation}

\end{document}